\documentclass[aps, prl, reprint, superscriptaddress]{revtex4-2}
\def\onlymaintext{0}

\def\isprl{1}
\input{style/preamble.sty}

\newcommand{\supp}{\textup{supp}}
\newcommand{\Ber}{\mathrm{Ber}}   %
\newcommand{\KL}{\mathrm{KL}}     %
\newcommand{\TV}{\mathrm{TV}}     %

\newcommand{\bigO}{\mathcal{O}}

\newcommand{\I}{I}
\newcommand{\ii}{\textup{i}}

\newcommand{\ham}{H}
\newcommand{\hdm}{\mathrm{H}}
\newcommand{\eps}{\epsilon}

\newcommand{\poly}{\textup{poly}}

\newcommand{\deltat}{\delta t}
\newcommand{\tmax}{T}
\newcommand{\pfU}{\mathscr{U}}
\newcommand{\trsteps}{R}

\newcommand{\cmm}{\textup{comm}}

\newcommand{\vbx}{\vb{x}}

\newcommand{\vbq}{\vb{q}}
\newcommand{\vbp}{\vb{p}}
\newcommand{\vbzero}{\vb{0}}

\newcommand{\vertiii}[1]{{\left\vert\kern-0.25ex\left\vert\kern-0.25ex\left\vert #1
		\right\vert\kern-0.25ex\right\vert\kern-0.25ex\right\vert}}
\newcommand{\Vertiii}[1]{{\vert\kern-0.25ex\vert\kern-0.25ex\vert #1
		\vert\kern-0.25ex\vert\kern-0.25ex\vert}}

\newcommand{\ceil}[1]{\lceil{#1}\rceil}

\definecolor{addgreen}{RGB}{0,128,64}

\newcommand{\aqcp}{\textup{\textsc{Approx-Qcirc-Prob}}}

\newcommand{\IQP}{\textup{IQP}}
\newcommand{\dqpt}{\textup{\textsc{Dqpt}}}

\newcommand{\DQPT}{\textup{DQPT}}

\newcommand{\losamp}{\mathcal{G}}
\newcommand{\los}{\mathcal{L}}
\newcommand{\rate}{r}
\newcommand{\drate}{\dot{r}}

\newcommand{\dratep}{\dot{r}_+}
\newcommand{\dratem}{\dot{r}_-}
\newcommand{\krate}{r_k}
\newcommand{\dkrate}{\dot{r}_k}
\newcommand{\ddkrate}{\ddot{r}_k}
\newcommand{\bpmf}{\textup{\textsf{bin}}}
\newcommand{\bpf}{\textup{\textsf{Bin}}}
\newcommand{\projk}{P^{(k)}}
\newcommand{\hdkrate}{\hat{\vphantom{A}\dot{r}}_k}
\newcommand{\hdlos}[1][k]{\hat{\vphantom{A}\dot{\los}}_{#1}}

\newcommand{\koverlap}{\textup{\textsc{$k$-Overlap}}}

\newcommand{\yes}{\textup{\textsc{YES}}}
\newcommand{\no}{\textup{\textsc{NO}}}

\newcommand{\Input}{\textup{input}}
\newcommand{\cc}{\mathrm{\bf c}}
\newcommand{\oo}{\mathrm{\bf o}}

\newcommand{\BQP}{\textsf{\textup{BQP}}}
\newcommand{\BPP}{\textsf{\textup{BPP}}}

\newcommand{\NP}{\textsf{\textup{NP}}}
\newcommand{\gapP}{\textsf{\textup{GapP}}}
\newcommand{\gap}{\mathrm{gap}}
\newcommand{\maj}{\mathrm{maj}}
\newcommand{\ngap}{\mathrm{ngap}}
\newcommand{\sharpP}{\textsf{\textup{\texttt{\#}P}}}
\newcommand{\complete}{\textsf{\textup{complete}}}
\newcommand{\hard}{\textsf{\textup{hard}}}
\newcommand{\hardness}{\textsf{\textup{hardness}}}

\newcommand{\local}{\textup{\textsc{Local}}}

\newcommand{\promise}{\textup{\textsc{promise}}}

\newcommand{\estimate}{\textup{\textsc{Estimate}}}

\newcommand{\search}{\textup{\textsc{Search}}}

\providecommand{\onlymaintext}{0}
\newcommand{\apdcref}[1]{%
  \ifnum\onlymaintext=0 \cref{#1}\else \cite{seesm}\fi
}

\input{style/nsp_style.sty}

\begin{document}

\ifnum\onlymaintext=0
\let\oldacl\addcontentsline
\renewcommand{\addcontentsline}[3]{}
\fi

\newcommand{\mytitle}{Provable Quantum Advantage for Dynamical Phase Transition}

\title{\mytitle}
\author{Jue Xu}
\email[]{xujue@connect.hku.hk}
\affiliation{QICI Quantum Information and Computation Initiative, Department of Computer Science,
The University of Hong Kong, Pokfulam Road, Hong Kong}
\author{Xiao Yuan}
\email[]{xiaoyuan@pku.edu.cn}
\affiliation{Center on Frontiers of Computing Studies, Peking University, Beĳing 100871, China}
\affiliation{School of Computer Science, Peking University, Beĳing 100871, China}
\author{Qi Zhao}
\email[]{zhaoqi@cs.hku.hk}
\affiliation{QICI Quantum Information and Computation Initiative, Department of Computer Science,
The University of Hong Kong, Pokfulam Road, Hong Kong}

\date{\today}
\begin{abstract}
    The universal scaling of critical behavior in phase transitions is a cornerstone of physics.
    Dynamical quantum phase transitions (DQPTs) are their nonequilibrium analogues: 
    abrupt nonanalyticities that emerge as a quantum system evolves in time.
    Yet the hardness and cost of detecting this phenomenon remain largely unexplored.
    We prove that estimating DQPT to a certain precision is
    intractable even for quantum computers, 
    whereas deciding a subsystem variant of DQPT is as hard as simulating generic quantum circuits, implying a provable exponential quantum advantage.
    Furthermore, to search for critical times of local DQPTs, 
    we show a quadratically faster quantum algorithm that estimates observables of Hamiltonian dynamics at multiple time points 
    with Heisenberg-limited precision and sublinear scaling in the number of time points.
    Moreover, through encoding classical evolution into quantum dynamics,
    our framework enables broader quantum speedups for detecting anomalous phenomena in classical systems.
\end{abstract}
\maketitle
\setcounter{secnumdepth}{0}

\mysection[1]{Introduction}
Characterizing quantum many-body dynamics far from equilibrium is a central challenge across physics, chemistry, and materials science.
In equilibrium, phase transitions arise as nonanalyticities of the free energy at critical points, separating distinct phases of matter.
The similar critical behaviors of phase transitions shared by vastly different physical systems are known as universality classes~\cite{sachdevQuantumPhaseTransitions2011}.
To extend this framework to nonequilibrium dynamics, where no free energy is defined and the state evolves unitarily, Heyl, Polkovnikov, and Kehrein~\cite{heylDynamicalQuantumPhase2013,heylDynamicalQuantumPhase2018} introduced the dynamical quantum phase transition (DQPT).
The central observable is the Loschmidt echo, 
the return probability of a quantum state to itself after a sudden quench~\cite{gorinDynamicsLoschmidtEchoes2006}, which plays the role of a dynamical partition function.
Its intensive logarithm, the rate function, develops sharp nonanalyticities at critical times in close analogy to a free-energy density.
These cusps obey scalings whose exponents depend only on an underlying universality class~\cite{heylDynamicalQuantumPhase2015}.
DQPTs now provide an organizing principle for nonequilibrium quantum matter~\cite{heylDynamicalQuantumPhase2018}, spanning broken-symmetry, long-range, topological, and driven systems~\cite{heylDynamicalQuantumPhase2014,vandammeAnatomyDynamicalQuantum2023,zunkovicDynamicalQuantumPhase2018,budichDynamicalTopologicalOrder2016,huangDynamicalQuantumPhase2016,langDynamicalQuantumPhase2018,zacheDynamicalTopologicalTransitions2019,hamazakiExceptionalDynamicalQuantum2021} 
and linking to entanglement and information scrambling~\cite{denicolaEntanglementViewDynamical2021,heylDetectingEquilibriumDynamical2018}.

\begin{figure*}[t!]
    \centering
    \includegraphics[width=0.97\linewidth]{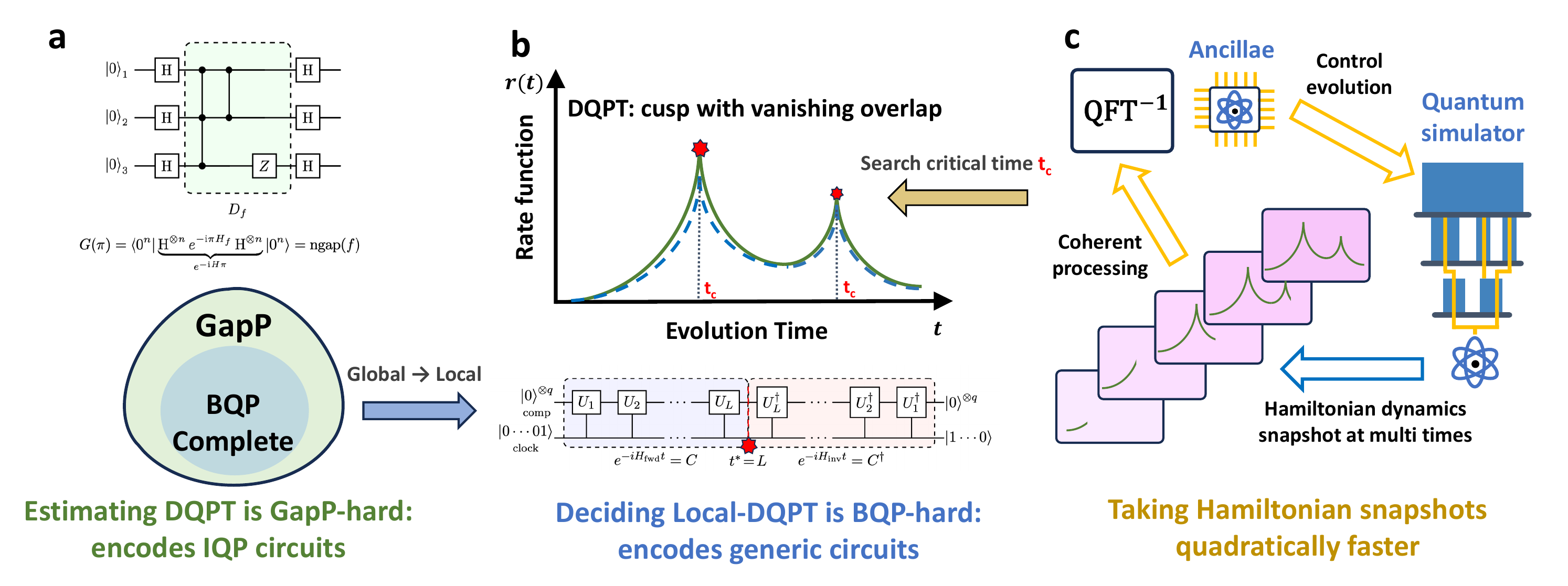}
    \caption{Illustration of practical quantum advantage for DQPT. 
    (a) Estimating rate function is \gapP{}-\hard{} by reduction from the degree-3 polynomial over $\mathbb{F}_2$ to a Hamiltonian dynamics that is equivalent to IQP circuits (above).
    (b) Deciding a local DQPT (solid line) that encodes a general polynomial-size quantum circuit (below) is \BQP{}-\complete{}, 
    implying a provable quantum advantage.
    (c) The coherent quantum algorithm searches local DQPT critical times using the gradient-estimation based Alg.~\ref{alg:multi_time}, which utilizes controlled evolution and Quantum Fourier Transform (QFT). 
    It achieves the Heisenberg-limited precision and sublinear scaling in the number of time points, which is nearly optimal for quantum algorithms and quadratically faster than classical methods.
    }
    \label{fig:dqpt_main}
\end{figure*}

Experimentally, DQPTs have been observed across quantum simulation platforms spanning trapped ions \cite{jurcevicDirectObservationDynamical2017,zhangObservationManyBodyDynamical2017,muellerQuantumComputationDynamical2023,deNonequilibriumCriticalScaling2025}, 
superconducting qubits \cite{xuProbingDynamicalPhase2020}, cold atoms \cite{bernienProbingManybodyDynamics2017}, photonic systems \cite{wangSimulatingDynamicQuantum2019}, 
etc~\cite{flaschnerObservationDynamicalVortices2018,tianObservationDynamicalQuantum2020,nieExperimentalObservationEquilibrium2020}.
Whether these observations translate into a practical quantum advantage of quantum simulation 
\cite{lloydUniversalQuantumSimulators1996,feynmanSimulatingPhysicsComputers1982,daleyPracticalQuantumAdvantage2022,zhangExponentialQuantumAdvantages2025,chenLocalMinimaQuantum2025}, 
has nevertheless remained unclear, 
because the Loschmidt echo is exponentially suppressed in system size and resolving it demands exponentially many samples.
We confirm this obstruction from complexity theory perspective: 
estimating the global rate function to a certain precision is $\gapP$-\hard,
likely out of reach even for quantum algorithms,
under standard complexity assumptions \cite{valiantComplexityComputingPermanent1979,aaronsonComputationalComplexityLinear2011,bremnerAveragecaseComplexityApproximate2016}.

To bypass the barrier from its definition, 
a subsystem variant of DQPT was proposed.
Instead of measuring the entire system, 
one looks only at the DQPT within a subsystem,
which has recently been systematically observed in a Bose--Hubbard gas~\cite{karchProbingQuantumManybody2025}.
To study its computational complexity, 
we formalize the task of determining whether a DQPT occurs in a subsystem at a specific time as a decision problem,
called \nameref{prm:promise_local_dqpt}.
We prove that it is $\BQP$-\complete:
a quantum computer solves it efficiently, 
while it is widely believed to be intractable for classical computers 
\cite{bennettStrengthsWeaknessesQuantum1997}.
This places local DQPT detection squarely in a sweet spot of practical quantum advantage.

Beyond the critical time decision problem, 
searching critical times within a given range is a practically motivated task.
We construct a quantum gradient-estimation-based protocol 
that estimates the subsystem Loschmidt echo at multiple time points with quadratically fewer queries to the quantum state preparation.
This algorithm is nearly optimal for this task and can be extended to any bounded observables. 
We also show that both the critical time and the universal scaling exponent are robust to Trotter error and noise, 
placing DQPT detection within reach of near-term hardware.
Since the underlying gradient encoding applies to any bounded Hermitian observable under Hamiltonian evolution, 
the same quadratic speedup carries over to processing dynamical information in classical systems such as coupled oscillators \cite{babbushExponentialQuantumSpeedup2023}, 
differential equations \cite{bravyiQuantumSimulationNoisy2025}, 
once they are embedded into quantum dynamics \cite{jinQuantumSimulationPartial2024,anLinearCombinationHamiltonian2023}.

\mysection[1]{Computational complexity of DQPT}
We first formalize DQPTs.
Given an $n$-qubit local Hamiltonian $H$, initial state $\ket{\psi_0}$, and evolution time $t$, 
the Loschmidt echo is $\los(t):= \abs{\losamp(t)}^2$ with amplitude $\losamp(t):= \bra{\psi_0} e^{-\ii \ham t} \ket{\psi_0}$, 
the (finite-size) rate function is
\begin{equation}\label{eq:global_echo}
    \rate(t) := -\frac{1}{n} \log \los(t).
\end{equation}
A DQPT occurs when $\rate(t)$ becomes nonanalytic at a critical time $t_c$ in the thermodynamic limit $n\to\infty$~\cite{heylDynamicalQuantumPhase2013}.
The time derivative of the rate function $\drate(t)$ plays the role of an order parameter for DQPT. 
It is analogous to the magnetization at an equilibrium ferromagnetic transition, 
which is the first-order derivative of the free energy with respect to an external field.
To detect $t_c$ with finite sizes, 
we formalize the critical-time condition as an exponentially small echo together with 
a sudden jump of $\drate(t)$ across $t_c$.
\begin{definition}[dynamical critical time]\label{def:dqpt}
    Given an $n$-qubit local Hamiltonian $H$, initial state $\ket{\psi_0}$, and $t \in \poly(n)$, 
    a \emph{critical time} $t_c$ of DQPT satisfies the rate function $\rate(t_c) \ge \xi \in \Theta(1)$ (exponentially small echo) and 
    $\dratem(t_c) - \dratep(t_c) \ge \eta > 0$ 
    (a slope discontinuity, i.e.\ a dynamical-susceptibility spike),
    where $\dratem(t_c):=\drate(t_c-\deltat)$ and $\dratep(t_c):=\drate(t_c+\deltat)$  are the time derivatives at $t_c$ with offset $\deltat$.
\end{definition}

Detecting a DQPT requires evaluating the rate $\rate(t)$ precisely, which is naturally tied to counting complexity~\cite{aaronsonComputationalComplexityLinear2011} via the Loschmidt echo $\los(t)$.
The relevant complexity class is $\gapP$~\cite{valiantComplexityComputingPermanent1979}, 
the set of differences of two $\sharpP$ functions.
Formally, for a Boolean function $f:\{0,1\}^n\to\{0,1\}$ with $N_b := |\{x : f(x) = b\}|$, 
computing the normalized gap $\ngap(f) := (N_0 - N_1)/2^n$ is $\gapP$-\complete, 
which is closely related to the amplitude of quantum circuits.
To reduce this counting problem to DQPT, 
we invoke the instantaneous quantum polynomial-time (IQP) circuit~\cite{bremnerClassicalSimulationCommuting2011}, 
which realizes $f$ as a 3-local Hamiltonian $H$ with the amplitude satisfying
\begin{equation}\label{eq:iqp_amplitude}
    \losamp(t) = \frac{1}{2^n} \sum_{x \in \{0,1\}^n} e^{-\ii t f(x)},
    \quad \losamp(\pi) = \ngap(f).
\end{equation}
The identity at $t = \pi$ 
converts an additive error $\eps$ in the rate function $\rate(\pi)$ into a multiplicative factor $e^{n\eps/2}$ in $|\ngap(f)|$ as in \cref{fig:dqpt_main}(a).
Ref.~\cite{bremnerAveragecaseComplexityApproximate2016} showed that computing the absolute value $|\ngap(f)|$ to constant multiplicative error remains $\gapP$-\hard{} for degree-3 polynomials over $\mathbb{F}_2$.
Therefore, we obtain the $\gapP$-\hardness{} of DQPT.

\begin{theorem}
    [$\gapP$-\hardness]
    \label{thm:sharp_hardness}
    There exists an $n$-qubit $3$-local Hamiltonian $\ham$
    and an evolution time $t$
    such that for the initial state $\ket{\psi_0} = \ket{0^n}$,
    estimating the rate function $\rate(t)$ to additive error $\eps\in\bigO(1/n)$ is $\gapP$-$\hard$ under polynomial-time Turing reductions.
\end{theorem}

We defer the proof to Methods.
This hardness reflects a statistical bottleneck: the exponentially small echo required by \nameref{def:dqpt} demands exponentially many samples to resolve, obstructing classical and likely quantum computation.

To bypass this obstacle,
we define the Loschmidt echo of a subsystem with constant size $k\in\bigO(1)$ and the $k$-local rate function 
following~\cite{halimehLocalMeasuresDynamical2021,bandyopadhyayObservingDynamicalQuantum2021,karchProbingQuantumManybody2025}, 
\begin{equation}\label{eq:local_echo}
    \los_k(t) := 
    \bra{\psi(t)} \prod_{j=1}^k P_j\, \ket{\psi(t)},
    \;\krate(t) := -\tfrac{1}{k}\log\los_k(t),
\end{equation}
where $\ket{\psi(t)} := e^{-\ii H t} \ket{0^n}$,
and the projector $P_j := \op{0}_j$ for the initial state $\ket{0^n}$. 
It recovers the (global) rate function \cref{eq:global_echo} when $k = n$. 
Because $\los_k = \Theta(e^{-k})$ near a critical point, estimating it to multiplicative error needs only $\bigO(1)$ samples.
\cref{fig:local_dqpt} illustrates the cusp in $\krate(t)$ and the jump of $\dkrate(t)$ for an analytically solvable Ising instance.
\begin{figure}[t]
    \centering
    \includegraphics[width=.96\linewidth]{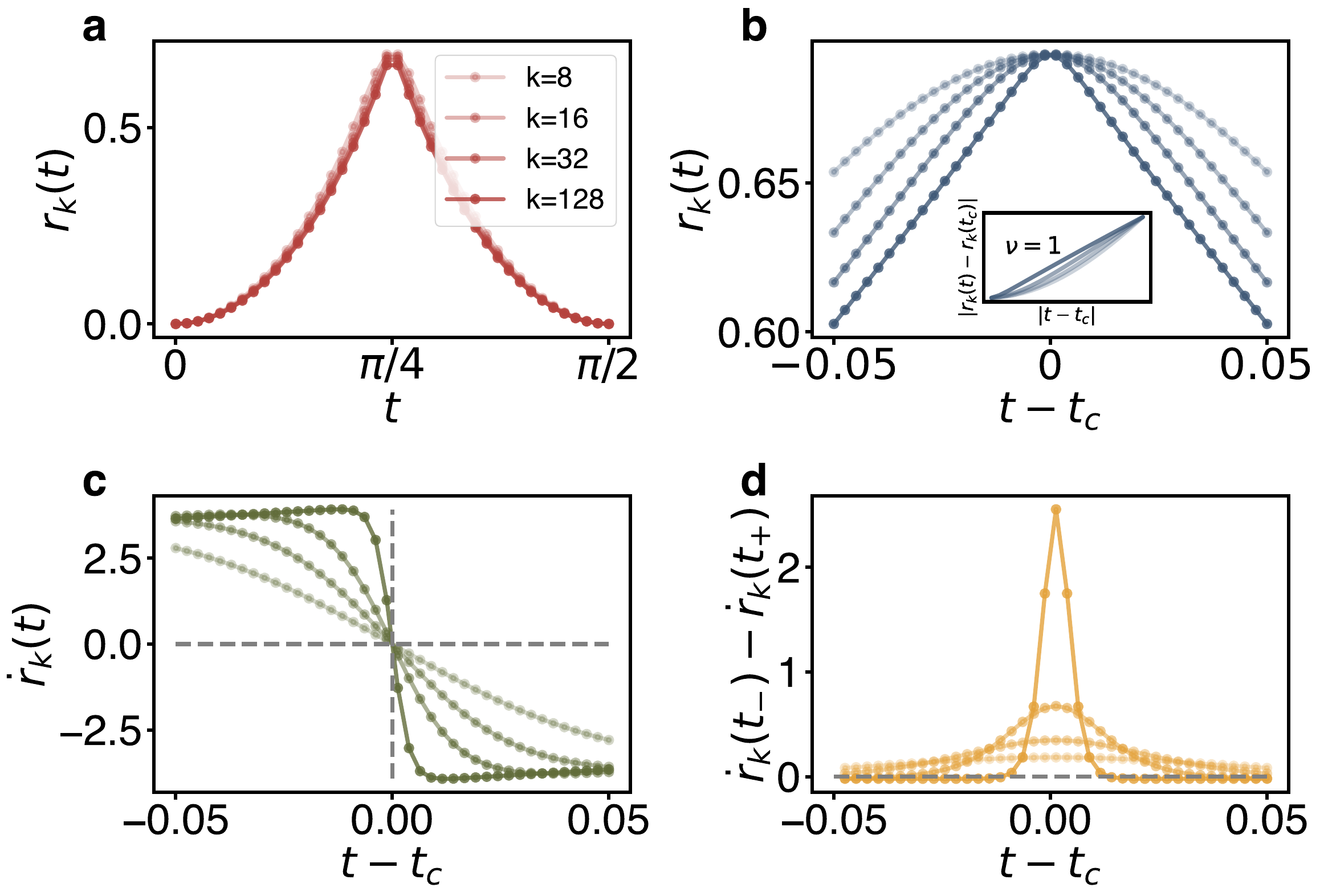}
    \caption{
    The local version of rate function 
    for the Ising Hamiltonian $H = J\sum_{j=1}^{n-1}  X_j X_{j+1}$ of $n=128$ spins and the initial state $\ket{\psi_0} = \ket{0}^{\otimes n}$, where the Loschmidt echo can be solved analytically \cite{halimehLocalMeasuresDynamical2021}.
    (a) The local rate function $\krate(t)$ for different subsystem sizes $k$.
    (b) The zoom-in plot of $\krate(t)$ near the critical time $t_c$.
    The inset shows the universal scaling with the exponent $\nu = 1$.
    (c) The time derivative of the local rate function $\dkrate(t)$.
    (d) 
    The difference of the left and right derivatives of the local rate function,
    i.e., $\dkrate(t_-)-\dkrate(t_+)$, 
    which quantifies the jump of order-parameter $\dkrate(t)$.
    }
    \label{fig:local_dqpt}
\end{figure}

With this local (subsystem) variant,
we formulate a decision problem 
called \nameref{prm:promise_local_dqpt} 
asks whether both the magnitude and susceptibility conditions of \nameref{def:dqpt}
hold at $t$ for a $k$-local subsystem with a promised gap, 
i.e., whether $t$ is a local critical time or not.

\begin{problem}[\local-\dqpt]\label{prm:promise_local_dqpt}
    Given a local Hamiltonian $H$, an initial state $\ket{\psi_0}$, evolution time $t$ with
    offset $\deltat$, subsystem size $k$, and magnitude thresholds $\xi_1 \ge \xi_0 > 0$, 
    and susceptibility thresholds $J_1 > J_0 \ge 0$,
    this problem is a promise problem to determine
    \begin{itemize}
        \item $\yes$: $\krate(t) \ge \xi_1$ and 
        $J_k(t):= \dot{r}_{k,-}(t) - \dot{r}_{k,+}(t)  \ge J_1$
        (i.e., it is a critical time of local DQPT),
        \item $\no$: $\krate(t) \le \xi_0$ or 
        $J_k(t):= \dot{r}_{k,-}(t) - \dot{r}_{k,+}(t)  \le J_0$,
    \end{itemize}
    with promised gaps 
    $J_1/J_0 \ge c$ for some constant $c > 1$
    and $\xi_1 - \xi_0 \in \Omega(1)$ with $\xi_1 \in \Omega(1)$,
    where $\dot{r}_{k,\pm}(t) := \dkrate(t \pm \deltat)$
    are the derivatives at the offset points $t\pm\deltat$. 
\end{problem}
    Here $J_k(t) $ quantifies the jump of the order parameter across the window, 
    i.e., the integrated dynamical susceptibility $ \int_{t-\deltat}^{t+\deltat}\chi_k(s)\,ds$ of density $\chi_k(t) := -\ddkrate(t)$.
We show this problem is $\BQP$-$\complete$, 
implying a provable quantum advantage for deciding local DQPTs.
The $\BQP$-$\hardness$ is by reduction from the standard problem of deciding the output of a generic polynomial-size quantum circuit $U$.
With $k$ parallel repetitions and a majority vote,
we obtain a new circuit $U'$
so that $\bra{0^n}U'^\dagger \prod_{j=1}^k P_j U'\ket{0^n}$ is either $\le e^{-\Omega(k)}$ or $\ge 1 - e^{-\Omega(k)}$. 
Then, we construct a palindromic circuit $V:=U'^\dagger  U'$ as shown in \cref{fig:dqpt_main}(b), 
which yields a special time at the midpoint of $V$,
emulating the critical time of a local DQPT depending on the output of the original circuit $U$.
To convert $V$ to Hamiltonian dynamics, 
we construct the Feynman--Kitaev Hamiltonian $H_{\mathrm{FK}}$ with binomial clock amplitudes. 
At $t = \pi/4$, 
a sharp cusp 
appears with $e^{-\ii H_{\mathrm{FK}}t}$ (satisfying both promise conditions of \nameref{prm:promise_local_dqpt}) 
if and only if the original circuit rejects.
See more details in Supplementary Materials \cite{seesm}.

The membership in $\BQP$ 
(i.e., the efficiency of the quantum algorithm) 
follows because a quantum computer estimates $\los_k(t)$ to constant multiplicative precision with $\bigO(1)$ samples 
(i.e., magnitude check) and 
estimates $\dot{r}_{k,\pm}$ via the commutator identity $\dot{\los}_k(t) = \ii\bra{\psi(t)}[\ham,\prod_{j=1}^k P_j ]\ket{\psi(t)}$ 
(i.e., susceptibility check, forming the jump $J_k = \dot{r}_{k,-} - \dot{r}_{k,+}$).
Thus any classical algorithm for \nameref{prm:promise_local_dqpt} would imply $\BPP = \BQP$, 
which is widely believed unlikely.
So, it is a provable quantum advantage.

\begin{theorem}[Informal]
\label{thm:local_bqp_complete}
    \nameref{prm:promise_local_dqpt} is $\BQP$-$\complete$ for any constant $k \ge 1$.
\end{theorem}

\mysection[1]{Coherent critical time search algorithm}
While we have shown the quantum advantage for deciding the existence of a local DQPT at a given time $t$,
one wants to \emph{search} for critical times $t_c$ in a time range $[0, \tmax]$.
Assuming consecutive critical times are separated by at least $\deltat_{\min}\in \Omega(1/\poly(n))$,
a natural strategy samples $\los_k$ on a grid of $M = \tmax/\deltat_{\min}$ time points and bisects candidate intervals where $\dkrate$ changes sign.
A near-term implementation (see \cref{fig:trotter}) estimates each grid point independently by Trotter evolution with direct measurement and achieves shot-noise cost $\bigO(M/\eps^2)$ samples.

Another key quantity of phase transitions is the exponent of the universal scaling.
In 1D systems, DQPTs generically exhibit power-law nonanalyticities 
near a critical time $t_c$ 
\begin{equation}\label{eq:scaling_form}
    \krate(t) \simeq \krate(t_c) - A\,|t - t_c|^\nu,
\end{equation}
where the critical exponent $0 < \nu \le 1$ characterizes the universality class~\cite{heylDynamicalQuantumPhase2013,heylDynamicalQuantumPhase2015},
while higher-dimensional systems such as the 2D Ising model exhibit logarithmic nonanalyticity~\cite{vajnaTopologicalClassificationDynamical2015}.
To estimate the power-law exponent $\hat\nu$ with additive error $\eps$,
we use a log-spaced fit to $\los_k(t_j)$ at $\bigO(1)$ points,
which also requires evaluating expectation values at multiple time points with $\eps$ precision.

\begin{figure}[t]
    \centering
    \includegraphics[width=.95\linewidth]{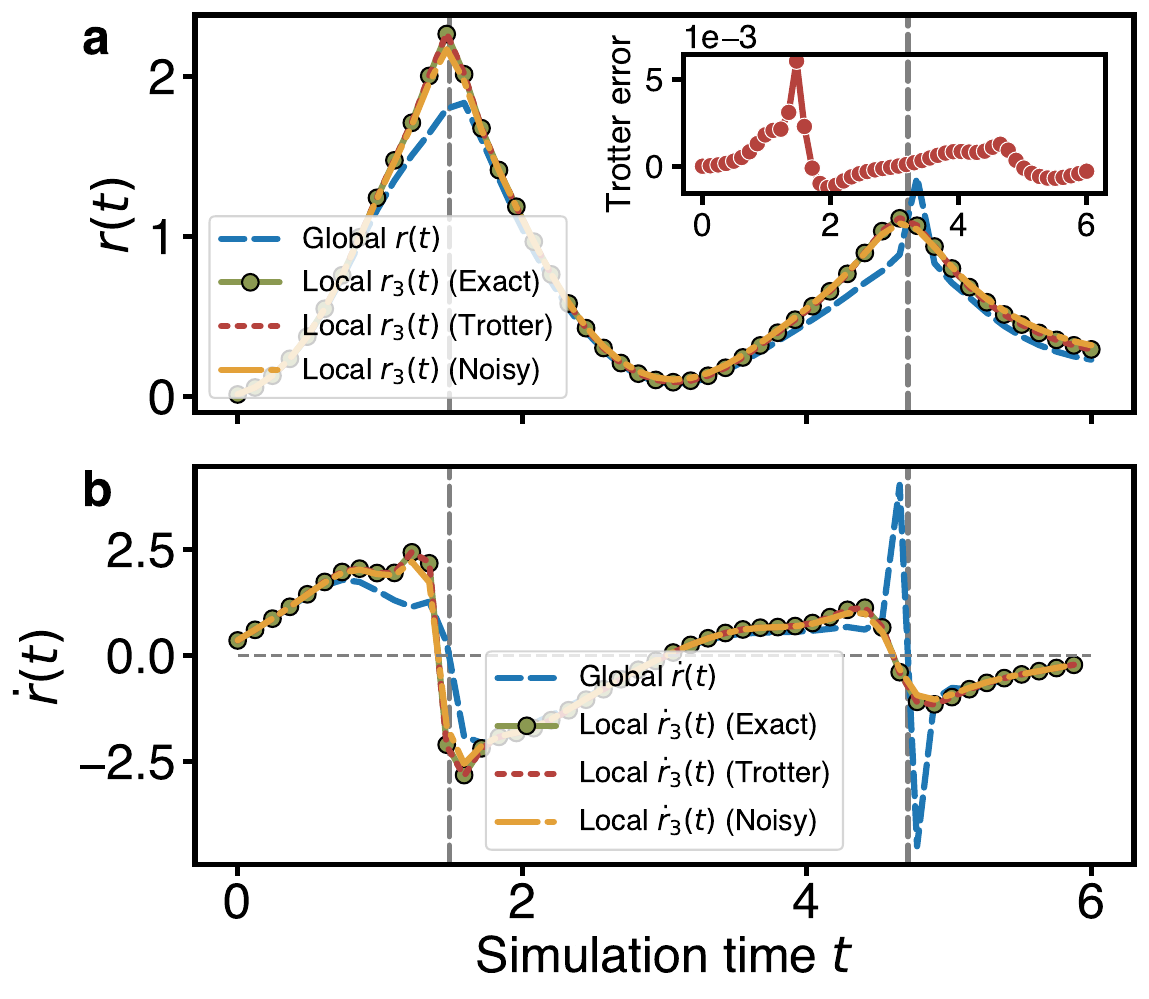}
    \caption{
    Search for local DQPT critical times and the robustness of the near-term protocol.
    (a) The local rate function $\krate(t)$ of the transverse-field Ising Hamiltonian $H=J\sum_{j}Z_jZ_{j+1}+h\sum_{j}X_j$ with $J=1$, $h=2$, $n=10$ spins and subsystem size $k=3$ (green circle line) compared with the evolution (red line) approximated by 50-step second-order Trotter formula and the one with local depolarizing noise at rate $10^{-3}$ (yellow line).
    The inset shows the Trotter error in the rate function.
    (b) The time derivative $\dkrate(t)$ is evaluated by finite differences of the rate function and exhibits a jump discontinuity at $t_c$.
    The simulated local critical times are close to the ideal one.
    }
    \label{fig:trotter}
\end{figure}

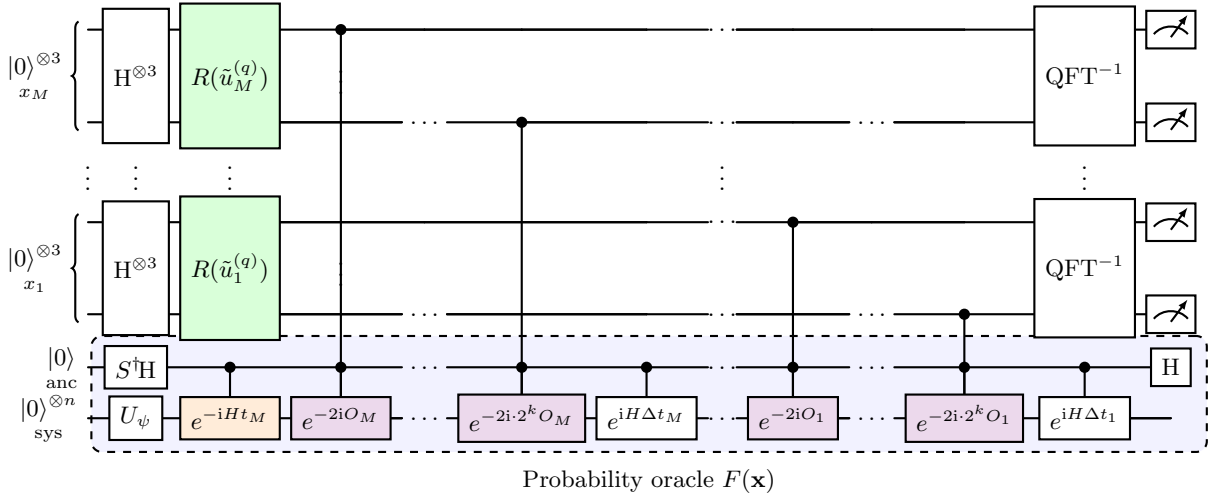
\begin{figure*}[t]
    \centering
    \begin{quantikz}[row sep=0.08cm, column sep=0.15cm]
        \lstick[wires=3]{$\ket{0}^{\otimes 3}$\\[-2pt]{\scriptsize $x_M$}}
            & \gate[wires=3]{\hdm^{\otimes 3}} & \gate[wires=3, style={fill=green!15}]{R(\tilde{u}_M^{(q)})} & \ctrl{8} & \qw & \qw & \qw & \ldots
            & \qw & \qw & \qw & \gate[wires=3]{\mathrm{QFT}^{-1}} & \meter{} \\[-0.1cm]
        \setwiretype{n}
            & & & \push{\vdots} & & & & & & & & & \\[-0.05cm]
        \setwiretype{q}
            & & & \qw & \ \ldots\ \qw & \ctrl{6} & \qw & \ldots
            & \qw & \ \ldots\ \qw & \qw & & \meter{} \\[-0.05cm]
        \setwiretype{n}
            \push{\vdots} & \vdots & \vdots & & & & & \vdots & & & & \vdots & \\
        \lstick[wires=3]{$\ket{0}^{\otimes 3}$\\[-2pt]{\scriptsize $x_1$}}
            & \gate[wires=3]{\hdm^{\otimes 3}} & \gate[wires=3, style={fill=green!15}]{R(\tilde{u}_1^{(q)})} & \qw & \qw & \qw & \qw & \ldots
            & \ctrl{4} & \qw & \qw & \gate[wires=3]{\mathrm{QFT}^{-1}} & \meter{} \\[-0.1cm]
        \setwiretype{n}
            & & & \push{\vdots} & & & & & & & & & \\[-0.05cm]
        \setwiretype{q}
            & & & \qw & \ \ldots\ \qw & \qw & \qw & \ldots
            & \qw & \ \ldots\ \qw & \ctrl{2} & & \meter{} \\
        \lstick{$\ket{0}$\\[-2pt]{\scriptsize anc}}
            & \gate{S^\dagger\!\hdm}
            \gategroup[wires=2, steps=12, style={dashed, rounded corners, fill=blue!5, inner xsep=1pt, inner ysep=0.5pt}, background, label style={label position=below, anchor=north, yshift=-0.3cm}]{Probability oracle $F(\vbx)$}
            & \ctrl{1} & \ctrl{1} & \ \ldots\  & \ctrl{1} & \ctrl{1} & \ldots
            & \ctrl{1} & \ \ldots\  & \ctrl{1} & \ctrl{1} & \gate{\hdm} \\
        \lstick{$\ket{0}^{\otimes n}$\\[-2pt]{\scriptsize sys}}
            & \gate{U_\psi}
            & \gate[style={fill=orange!15}]{e^{-\ii\ham t_M}}
            & \gate[style={fill=violet!15}]{e^{-2\ii O_M}}
            & \ \ldots\ \qw
            & \gate[style={fill=violet!15}]{e^{-2\ii\cdot 2^k O_M}}
            & \gate{e^{\ii\ham\Delta t_M}}
            & \ldots
            & \gate[style={fill=violet!15}]{e^{-2\ii O_1}}
            & \ \ldots\ \qw
            & \gate[style={fill=violet!15}]{e^{-2\ii\cdot 2^k O_1}}
            & \gate{e^{\ii\ham\Delta t_1}}
            & \qw
    \end{quantikz}
    \caption{The quantum circuit for one round of adaptive snapshots estimation (Alg.~\ref{alg:multi_time}).
    The algorithm repeats this circuit for rounds $q = 0, 1, \ldots, \ceil{\log_2(1/\eps)}$.
    The dashed box encloses the probability oracle $F(\vbx)$, which is unchanged across rounds:
    an ancilla qubit implements the Hadamard test on $U(\vbx)$, mapping
    $f(\vbx) = \frac{1}{2}(1 - \Im\bra{\psi_0}U(\vbx)\ket{\psi_0})$
    to measurement probability.
    Each index register $x_j$ has $p = 3$ qubits, independent of~$\eps$.
    At each round, classical phase corrections $R(\tilde{u}_j^{(q)})$ (shaded green) shift the index registers
    based on previous estimates; after inverse QFT and measurement, the estimates are updated as
    $\tilde{u}_j^{(q+1)} = \tilde{u}_j^{(q)} + \pi\,2^{-q} g_j^{(q)}$.
    Bit~$k$ of $x_j$ controls a rotation $e^{-2\ii\cdot 2^k O_j}$ (shaded violet), doubly controlled with the ancilla.
    The initial evolution $e^{-\ii\ham t_M}$ (shaded orange) differs from the incremental steps $e^{\ii\ham\Delta t_j}$.
    }
    \label{fig:circuit_gradient}
\end{figure*}

Therefore, it is a desideratum to achieve the Heisenberg limit $\bigO(1/\eps)$ and a scaling sublinear in $M$ when taking $M$ snapshots of Hamiltonian dynamics.
We show that faster searching of local critical times is feasible with a coherent quantum algorithm.
We first consider a general dynamical estimation task, 
i.e., take snapshots of Hamiltonian dynamics by measuring observables on the evolved states at multiple time points.
\begin{problem}[Hamiltonian snapshots estimation]\label{prm:multi_time_obs}
    Given an $n$-qubit Hamiltonian $\ham$, an initial state $\ket{\psi_0}$ prepared by $U_{\psi}$,
    $M$ time points $\{t_j\}_{j=1}^M$ in a range $[0,T]$, 
    bounded Hermitian observables $\{O_j\}_{j=1}^M$ with $\|O_j\| \le 1$,
    and precision $\eps > 0$,
    output estimates $\tilde{o}_j$ satisfying
    $|\tilde{o}_j - \bra{\psi_0} e^{\ii \ham t_j} O_j e^{-\ii \ham t_j} \ket{\psi_0}| \le \eps$ for all $j \in [M]$.
    The goal is to minimize total queries to $U_\psi$ and total Hamiltonian simulation time.
\end{problem}
The main idea of the fast quantum algorithm is to encode expectation values at $M$ distinct times as the gradient of a single scalar function, 
then extract all components simultaneously via quantum gradient estimation~\cite{gilyenOptimizingQuantumOptimization2019,hugginsNearlyOptimalQuantum2022,wadaHeisenbergLimitedAdaptiveGradient2025}.
For bounded Hermitian observables $\{O_j\}$ with $\|O_j\|\le 1$, we interleave short-time evolutions with auxiliary rotations to form the parameterized unitary
\begin{equation}\label{eq:param_unitary}
    U(\vbx) := \Bigl(\prod_{j=1}^{M} e^{\ii \ham \Delta t_j}\, e^{-2\ii x_j O_j}\Bigr) e^{-\ii \ham \tmax},
\end{equation}
with $\Delta t_j := t_j - t_{j-1}$, 
and set $f(\vbx) := \tfrac{1}{2}[1 - \Im\bra{\psi_0}U(\vbx)\ket{\psi_0}]$.
Differentiating \cref{eq:param_unitary} gives the gradient--observable identity\label{lem:oracle_construction}
\begin{equation}\label{eq:gradient_echo}
    \left.\frac{\partial f}{\partial x_j}\right|_{\vbx=\vbzero} = \bra{\psi_0} O_j(t_j) \ket{\psi_0},
    \quad j = 1,\ldots,M,
\end{equation}
with $O_j(t_j) := e^{\ii \ham t_j}\,O_j\,e^{-\ii \ham t_j}$ the $j$th observable evolved for time $t_j$ in the Heisenberg picture.

A key component of this quantum gradient estimation algorithm is the probability oracle
    $U_f(\vbx)$,
    which maps $\ket{\vbx} \ket{\vb{0}}$
    to $\ket{\vbx} \qty(\sqrt{1 - f(\vbx)}\ket{0}\ket{\phi_0(\vbx)} + \sqrt{f(\vbx)}\ket{1}\ket{\phi_1(\vbx)})$.
A Hadamard test on the controlled-$U(\vbx)$ circuit
(the dashed box in \cref{fig:circuit_gradient})
realizes this probability oracle
at a cost of one call to $U_{\psi}$,
total simulation time $2\tmax$,
and $M$ controlled rotations $e^{-2\ii x_j O_j}$ per query.
Embedding $F(\vbx)$ in index registers of $p$
qubits each and applying adaptive quantum gradient estimation~\cite{wadaHeisenbergLimitedAdaptiveGradient2025} extracts all $M$ components simultaneously.
\cref{fig:circuit_gradient} presents the circuit of Alg.~\ref{alg:multi_time}.

\begin{theorem}
    [Adaptive snapshots estimation]\label{thm:multi_time_estimation}
    For $M$ time points $\qty{t_j}_{j=1}^{M}$ in $[0,\tmax]$ and bounded observables $\|O_j\|\le 1$, 
    there is a quantum algorithm outputting estimators $\hat{o}_j$ with mean-squared error $\max_j \mathbb{E}[(\hat{o}_j - \bra{\psi_0}O_j(t_j)\ket{\psi_0})^2] \le \eps^2$ using $\bigO(\eps^{-1}\sqrt{M}\log M)$ oracle queries, total simulation time $\bigO(\eps^{-1}\sqrt{M}\,\tmax \log M)$, and $\bigO(M+n)$ qubits.
\end{theorem}

Specifically, the Heisenberg-limited scaling $\bigO(1/\eps)$ arises from coherent amplitude estimation on the probability oracle,
quadratically better than the $\bigO(1/\eps^2)$ standard quantum limit. 
The $\sqrt{M}$ scaling exploits quantum parallelism, with index registers in superposition probing all $M$ directions at amplitude $\sim 1/\sqrt{M}$ per query.

Consequently, we have an algorithm for searching critical times of local DQPTs.
Setting $O_j = P$ (the $k$-local projector) gives $\tilde\los_j \approx \los_k(t_j)$, 
with each oracle decomposing into $\bigO(kM)$ extra single-qubit-controlled gates.
Plugging this into the two-phase search---a coarse screening on the $M = \tmax/\deltat_{\min}$ grid followed by $\bigO(\log(\deltat_{\min}/\eps))$ bisections---completes the critical-time search with the following cost.
\begin{corollary}[Complexity of \search-\dqpt]\label{thm:optimal_search_dqpt}\label{prm:search_dqpt}
    Given an $n$-qubit local Hamiltonian $\ham$, 
    initial state $\ket{\psi_0}$, time range $\tmax \in \poly(n)$, and precision $\eps > 0$,
    assume consecutive critical times are separated by at least $\deltat_{\min} \in \Omega(1/\poly(n))$,
    there is a quantum algorithm 
    outputting estimates of critical times $\tilde{t}_c$ satisfying $|\tilde{t}_c - t_c| \le \eps$ for every critical time $t_c \in [0,\tmax]$ of \nameref{prm:promise_local_dqpt}, 
    using total Hamiltonian simulation time $\widetilde{\bigO}(\tmax^{3/2}/\deltat_{\min}^{1/2}) + \bigO(\tmax\log(\deltat_{\min}/\eps))$ and $\bigO(\tmax/\deltat_{\min} + n)$ qubits.
\end{corollary}

\mysection[1]{Optimality: quantum lower bound}
While we have established the upper bound of Hamiltonian snapshots estimation in \cref{thm:multi_time_estimation},
we can show both scalings in $M$ and $\eps$ are nearly tight for quantum algorithms.
Since it is known that estimating $M$ observables to additive error $\eps$ in the state-preparation oracle model requires $\Omega(\sqrt{M}/\eps)$ queries~\cite{hugginsNearlyOptimalQuantum2022,vanapeldoornQuantumProbabilityOracles2021},
we prove the quantum lower bound of \nameref{prm:multi_time_obs} by a reduction from multi-observable estimation.
\begin{theorem}[Quantum lower bounds]\label{thm:multi_lower_bound}
    Given observables $\qty{O_j}$ and $M$ time points $\qty{t_j}$ in the range $[0,T]$, 
    the oracle to prepare initial state $U_{\psi}$,
    and the oracle to Hamiltonian dynamics $e^{-\ii H t}$,
    any quantum algorithm estimating $\bra{\psi_0}O_j(t_j)\ket{\psi_0}$ at $M$ time points to additive error $\eps\in (0, 1/(3\sqrt{M}))$ requires $\Omega(\sqrt{M}/\eps)$ queries to $U_{\psi}$.
    And the Hamiltonian simulation time is $\Omega(T)$.
\end{theorem}

\mysection[1]{Quantum speedup: classical lower bound}
On the other hand, we also show this quantum algorithm exhibits speedups over classical algorithms under the prepare-and-measure model.
Classically, each $\los_k(t_j) \in [0,1]$ is a Bernoulli mean requiring $\Omega(1/\eps^2)$ classical samples for a time $t_j$~\cite{dagumOptimalAlgorithmMonte2000}.
Since the Heisenberg-picture projectors $e^{\ii\ham t_j}\, P\, e^{-\ii\ham t_j}$ generically fail to commute across distinct times, 
each prepared sample resolves only one time point,
so the $M$ near-balanced snapshots must be estimated separately, 
compounding to $\Omega(M/\eps^2)$ samples in total.
\begin{theorem}[Classical lower bound]\label{thm:classical_lower_bound}
    In the classical prepare-and-measure model,
    where each round prepares one copy of $\ket{\psi_0}$ with a single $U_\psi$ call, 
    evolves it under $e^{-\ii\ham t}$ for a chosen time $t$, and measures a fixed projector $P$ for one outcome bit, 
    there exist an $\bigO(1)$-local Hamiltonian $\ham$ on $n=\bigO(M^2)$ qubits, a $1$-local projector $P$, and $M$ time points $\qty{t_m}$ 
    such that estimating all $M$ snapshots $\bra{\psi_0}\,e^{\ii\ham t_m}\,P\,e^{-\ii\ham t_m}\ket{\psi_0}$ to additive error $\eps$ (with success probability $\ge 2/3$, 
    for $\eps$ below an absolute constant) requires $\Omega(M/\eps^2)$ rounds, i.e.\ calls to $U_\psi$.
\end{theorem}

The proof embeds $M$ independent near-$\tfrac12$ parameters into a single $\bigO(1)$-local Hamiltonian through a Feynman--Kitaev clock, 
so that measuring $P$ after evolving for a time $t_m$ returns one bit whose bias is a fixed weighted mixture of the $M$ unknown parameters; 
recovering all $M$ snapshots therefore forces recovering all $M$ parameters.
A multi-parameter estimation lower bound, 
Assouad's lemma applied to a hypercube of near-$\tfrac12$ instances, 
then yields the $\Omega(M/\eps^2)$ bound. 
See Supplementary Materials \cite{seesm} for more details.
The coherent quantum algorithm, by contrast, extracts $\bigO(\sqrt{M})$ snapshot directions per $U_\psi$ call via gradient interference, which is what evades the one-bit-per-round classical cost.
Thus, \cref{thm:multi_time_estimation} is a quadratic improvement in both $M$ and $\eps$ over classical approaches (\cref{thm:classical_lower_bound}).

\mysection[1]{Beyond quantum dynamics}
Because the gradient-observable identity \cref{eq:gradient_echo} holds for any bounded Hermitian observable, the protocol extends directly to quantum encodings of classical dynamics.
A striking instance is the simulation of $N=2^n$ coupled classical oscillators governed by Newton's equations 
\cite{babbushExponentialQuantumSpeedup2023}, 
whose phase-space state is encoded into $\bigO(n)$ qubits with velocity and position amplitudes scaled by the conserved total energy $E$.
The kinetic-energy fraction of a subset $V\subseteq[N]$ is the expectation $K_V(t)/E = \bra{\psi(t)}P_V\ket{\psi(t)}$ of a velocity-sector projector $P_V$, 
and, in the thermodynamic limit, a nonanalyticity of the rate function $r_V(t):=-\log K_V(t)/E$
signals a sharp redistribution of energy among oscillators.
Single-point estimation of $K_V(t)/E$ is already $\BQP$-complete~\cite{babbushExponentialQuantumSpeedup2023}; 
substituting $O_j=P_V$ into \cref{thm:multi_time_estimation} tracks kinetic energy at $M$ times in $\widetilde{\bigO}(\sqrt{M}/\eps)$ queries,
a quadratic speedup over independent estimation.

More broadly, Hamiltonian encoding methods extend the framework to nonunitary and nonlinear classical evolution.
For instance, linear flows $\partial_t\vec u = -A\vec u$, covering heat, Fokker--Planck, and Black--Scholes dynamics, lift to unitary evolution on an extended Hilbert space via Schr\"odingerization~\cite{jinQuantumSimulationPartial2024} and the linear combination of Hamiltonian simulation (LCHS)~\cite{anLinearCombinationHamiltonian2023}, 
where the survival fraction 
is proportional to a qubit-local projector overlap. 
Even nonlinear dissipative ODEs, including the Navier--Stokes equations, embed into Fock-space Hamiltonians via the backward Kolmogorov equation, 
where the simulation problem is itself $\BQP$-\complete{}~\cite{bravyiQuantumSimulationNoisy2025}. 
Our quantum protocol 
opens a more efficient route to detecting the onset of intriguing phenomena such as a $\DQPT$-like singularity in classical systems.

\mysection[1]{Discussion}
This work places DQPT detection alongside ground-state energy estimation as a physically motivated source of provable quantum advantage.
The lower bounds imply that our protocol is not merely optimal in quantum cases but also faster than classical approaches. 
This work opens a new avenue for quantum algorithms in classical dynamical information processing.
Meanwhile, there remain several directions for future work.
Does $\BQP$-\textsf{completeness} persist under average-case hardness for geometrically two-local, translation-invariant Hamiltonians 
where experimental DQPTs actually live?
It is also of practical interest to reduce the cost of the search algorithm with prior knowledge.

\begin{acknowledgments}
J.X. thanks Wenjun Yu, Daochen Wang, Zhengfeng Ji, Boyang Chen, for helpful discussions. 
We also acknowledge Claude Code for suggestion.
J.X. and Q.Z. acknowledge funding from Innovation Program for Quantum Science and Technology via Project 2024ZD0301900, National Natural Science Foundation of China (NSFC) via Project No. 12347104 and No. 12305030, Guangdong Basic and Applied Basic Research Foundation via Project 2023A1515012185, Hong Kong Research Grant Council (RGC) via No. 27300823, N\_HKU718/23, and R6010-23, Guangdong Provincial Quantum Science Strategic Initiative No. GDZX2303007, HKU Seed Fund for Basic Research for New Staff via Project 2201100596. 
\end{acknowledgments}

\bibliographystyle{style/truncate_ref}
\bibliography{bib/ref_aps, bib/seesm}
\clearpage

\section{Methods}

\subsection{GapP hardness of global DQPT}

\begin{proof}[Proof sketch of \cref{thm:sharp_hardness}]
    Estimating $\rate(\pi)$ to additive error $\bigO(1/n)$ is at least as hard as multiplicatively approximating $|\ngap(f)|$ for degree-3 polynomials $f$ over $\mathbb{F}_2$, which is $\gapP$-\hard{}~\cite{bremnerAveragecaseComplexityApproximate2016}.
    The reduction has two ingredients:
    (i) the IQP Hamiltonian construction,
    which maps $f$ in $\poly(n)$ time to a 3-local $\ham$ whose Loschmidt amplitude at $t=\pi$ equals $\ngap(f)$, and
    (ii) an additive-to-multiplicative error amplification through the logarithm in $\rate$.

    Replacing each Boolean variable $x_i$ in $f$ by the projector $P_i = (I-Z_i)/2$ produces the diagonal problem Hamiltonian $H_f = \sum_{ijk}\alpha_{ijk}P_iP_jP_k + \sum_{ij}\beta_{ij}P_iP_j + \sum_i\gamma_i P_i$, whose eigenvalues on the computational basis are precisely the polynomial values: $H_f\ket{x} = f(x)\ket{x}$.
    Since $f$ has degree 3, each summand involves at most three projectors and $H_f$ is 3-local.
    Conjugating by Hadamards, $\ham := \hdm^{\otimes n} H_f \hdm^{\otimes n}$, rotates $Z \mapsto X$ and reproduces the IQP structure---Hadamards sandwich a diagonal phase---when applied to $\ket{0^n}$.
    At $t=\pi$, the phase factor $e^{-\ii\pi f(x)} = (-1)^{f(x)}$ collapses to $\pm 1$, reducing the amplitude to the signed counting fraction:
    \begin{align}\label{eq:method_iqp_at_pi}
        \losamp(\pi) = \frac{1}{2^n}\sum_{x\in\{0,1\}^n}(-1)^{f(x)} = \frac{N_0 - N_1}{2^n} = \ngap(f).
    \end{align}

    Inverting $\rate(\pi) = -(2/n)\log|\ngap(f)|$ yields $|\ngap(f)| = e^{-n\rate(\pi)/2}$, so an additive-$\eps$ estimate $\tilde{r}$ of $\rate(\pi)$ obeys
        $e^{-n\eps/2} \le \frac{e^{-n\tilde{r}/2}}{|\ngap(f)|} \le e^{n\eps/2}$,
    a multiplicative blow-up by $e^{n\eps/2}$ on the gap estimate.
    Choosing $\eps = \bigO(1/n)$ tightens this to $\bigO(1)$ multiplicative precision, sufficient for $\gapP$-\hardness{}.
    See \apdcref{apd:sec:sharp} for the full proof.
\end{proof}

\subsection{BQP completeness of local DQPT}
\begin{proof}[Proof sketch of \cref{thm:local_bqp_complete}]
    \textsf{BQP-hardness}:
    Given a $\BQP$ circuit $C$, we run $m = \Theta(k)$ independent copies,
    majority-vote their outputs, and fan out the result via \textsc{cnot}
    to $k$ fresh answer qubits at positions $1,\dots,k$;
    call the resulting compute-and-fan-out circuit $W$ (with $\ell$ gates total).
    By the Chernoff bound, the $k$-local overlap
    $\bra{0^n}W^\dagger P^{(k)}\, W\ket{0^n}$
    (with $P^{(k)} = \prod_{j=1}^{k}\op{0}_j$) 
    is either
    $\le\epsilon$ or $\ge 1-\epsilon$ with
    $\epsilon = e^{-\Omega(k)}$,
    so this $k$-Overlap problem is $\BQP$-\complete{}.
    We then 
    insert $w$ idle (identity) gates between the forward and inverse halves of $W$, holding the fanned-out answer in the middle of the palindrome:
    \begin{equation}
        V' := \bigl(W_1,\dots,W_\ell,\ \underbrace{I,\dots,I}_{w \text{ idle}},\ W_\ell^\dagger,\dots,W_1^\dagger\bigr),
    \end{equation}
    as in \cref{fig:dqpt_main}(b) and embed $V'$
    into a time-independent local Hamiltonian $H$ via the Feynman--Kitaev
    construction \apdcref{thm:circuit_hamiltonian},
    whose clock register produces binomial amplitudes
    $|c_j(t)|^2 = \binom{N'}{j}\sin^{2j}(t)\cos^{2(N'-j)}(t)$ and $N':= 2\ell+w$.
    At $t^* = \pi/4$, 
    the idle gates hold the fanned-out answer across the segment $B = \{\ell, \dots, \ell+w\}$,
    so the $k$-local echo reads the overlap value $a$ throughout $B$; 
    the window is set by the exact dial $\sin(2\deltat) = w/N'$, with $t^* = \pi/4$ (clock step $N'/2$) centred in $B$.
    If the $k$-Overlap instance is \no{} (small overlap),
    the $k$-local echo $\los_k(t^*)$ collapses to $e^{-\Omega(k)}$,
    creating a cusp in the rate function with
    $J_k(t^*) =\dkrate(t^*-\deltat) - \dkrate(t^*+\deltat) \ge \Omega(\sqrt{N'}/k)$.
    The magnitude condition is also satisfied: 
    $\krate(\pi/4) = \Omega(1) \ge \xi_1$.
    If the instance is \yes{} (the output passes the overlap test),
    $\los_k(t) \ge 1 - e^{-\Omega(k)}$ uniformly in $t$,
    so 
    the jump is suppressed to $\bigO(\epsilon\sqrt{N'}/k)$ with $\epsilon = e^{-\Omega(k)}$
    and $\krate(t) \le \bigO(e^{-\Omega(k)}/k) \le \xi_0$.
    Both promise gaps are satisfied:
    $J_1/J_0 = \Theta(1/\epsilon) \ge c > 1$ and
    $\xi_1 - \xi_0 \ge \Omega(1)$
    (since $k = \bigO(1)$ and $\epsilon < 1/2$),
    completing the reduction.
    See \apdcref{apd:thm:bqp_hard} for the complete proof.

\textsf{BQP membership}:
    The algorithm simulates $e^{-\ii \ham t}\ket{\psi_0}$ and performs two checks.
    \emph{Magnitude check:} estimate $\los_k(t)$ to constant multiplicative error
    by measuring $P^{(k)}$; 
    for constant $k$ and $\xi_1 \in \Omega(1)$,
    $\los_k = e^{-k\xi_1} = e^{-\Theta(k)} = \Theta(1)$,
    so $\bigO(1)$ repetitions suffice.
    \emph{Susceptibility check:} estimate $\dot{r}_{k,\pm}(t)$ via the
    commutator method:
    by the Heisenberg equation $\dv*{\expval{O}}{t} = \ii\expval{[\ham, O]}$, one expresses
    $\dot{\los}_k(t) = \bra{\psi(t)} Q \ket{\psi(t)}$ with $Q = \ii[\ham, P^{(k)}]$
    and requires $\bigO(\norm{Q}_1^2/(k\los_{\min}\eps)^2)$ samples.
    Under the \no{} case's OR condition, it suffices to detect
    either 
    $J_k(t) \le J_0$ or $\krate(t) \le \xi_0$;
    the promise gaps 
    $J_1/J_0 \ge c > 1$ and $\xi_1 - \xi_0 \ge \Omega(1)$
    ensure constant relative precision resolves the decision.
\end{proof}

\subsection{Quantum algorithm cost for searching critical times}

\begin{proof}[Proof sketch of \cref{thm:multi_time_estimation}]
    \emph{Gradient encoding.}
    The algorithm encodes the $M$ Heisenberg-picture expectation values $\bra{\psi_0}O_j(t_j)\ket{\psi_0}$ as the gradient at the origin of a single scalar function $f(\vbx) = \tfrac{1}{2}[1-\Im\bra{\psi_0}U(\vbx)\ket{\psi_0}]$, then extracts all components in parallel via adaptive quantum gradient estimation~\cite{wadaHeisenbergLimitedAdaptiveGradient2025}.
    At $\vbx = \vbzero$, every rotation $e^{-2\ii\cdot 0\cdot O_j}$ reduces to the identity, so $U(\vbzero) = I$ by telescoping of \cref{eq:param_unitary} and $f(\vbzero) = 1/2$ lies in the analytic regime that licenses gradient estimation.
    Differentiating $U(\vbx)$ at the origin gives
    \begin{align}\label{eq:method_param_derivative}
        \left.\frac{\partial U}{\partial x_\ell}\right|_{\vbzero} = -2\ii\, e^{\ii\ham t_\ell}\, O_\ell\, e^{-\ii\ham t_\ell},
    \end{align}
    which, substituted into $f$, yields the gradient--observable identity \cref{eq:gradient_echo}.
    The bound $\|O_\ell\|\le 1$ propagates to all multi-derivatives, $|\partial_\alpha f(\vbzero)| \le 2^{|\alpha|}\,|\alpha|^{|\alpha|/2}$, satisfying the derivative growth condition required by quantum gradient estimation~\cite{gilyenOptimizingQuantumOptimization2019}.

    \emph{Adaptive gradient estimation.}
    The adaptive quantum gradient estimation algorithm Alg.~\ref{alg:multi_time} uses $p = 3$ index qubits per time point and iterates $q = 0, 1, \ldots, \ceil{\log_2(1/\eps)}$ rounds.
    At round $q$, the gradient is estimated to precision $2^{-q}$ using $\bigO(2^q\sqrt{M}\log M)$ queries to the probability oracle $F$; 
    summing the geometric series gives $\bigO(\eps^{-1}\sqrt{M}\log M)$ total queries.
    The shift by temporal estimates $\tilde{u}_j^{(q)}$ at each round is a classical phase correction on the index registers, 
    leaving the oracle circuit unchanged, so each query consumes simulation time $2\tmax$ and one call to $U_\psi$.
    The MSE guarantee follows from geometrically scheduled failure probabilities $\delta^{(q)} = c_0/8^{q_{\max}-q}$, bounding each round's variance contribution.
    See \apdcref{apd:cor:adaptive_snapshots} for the full proof.

    \emph{Cost of searching critical times.}
    Setting $O_j = P^{(k)}$ specializes the snapshot estimator to $\tilde\los_j \approx \los_k(t_j)$, each oracle adding $\bigO(kM)$ controlled rotations.
    The two-phase search of \cref{thm:optimal_search_dqpt} then proceeds in
    (i) a coarse \emph{screening} over the $M_0 = \tmax/\deltat_{\min}$ uniform grid, resolving $\los_k$ to constant precision in $\widetilde\bigO(\sqrt{M_0})$ oracle queries and total simulation time $\widetilde\bigO(\tmax^{3/2}/\deltat_{\min}^{1/2})$;
    (ii) a subsequent \emph{bisection} refining each surviving cusp in $\bigO(\log(\deltat_{\min}/\eps))$ rounds at simulation time $\bigO(\tmax)$ per round.
    The screening phase dominates, giving the cost of \cref{thm:optimal_search_dqpt}.
\end{proof}

\begin{algorithm}[t!]
\caption{Adaptive snapshots estimation}\label{alg:multi_time}
\DontPrintSemicolon
\SetKwInOut{Input}{Input}
\SetKwInOut{Output}{Output}
\Input{Hamiltonian $\ham$, state preparation $U_\psi$ for $\ket{\psi_0}$, time points $\{t_j\}_{j=1}^M$,\newline
bounded observables $\{O_j\}_{j=1}^M$ with $\|O_j\| \le 1$, precision $\eps$}
\Output{Estimates $\hat{o}_j$ with $\max_j \mathbb{E}[(\hat{o}_j - \bra{\psi_0} O_j(t_j) \ket{\psi_0})^2] \le \eps^2$}
\BlankLine
Construct the parameterized unitary $U(\vbx)$ (\cref{eq:param_unitary}) and the probability oracle $F(\vbx)$ via Hadamard test\;
Set $q_{\max} \gets \ceil{\log_2(1/\eps)}$ and $\tilde{u}_j^{(0)} \gets 0$ for all $j \in [M]$\tcp*{$\tilde{u}_j$ estimates $\bra{\psi_0}O_j(t_j)\ket{\psi_0}$}
\For{$q = 0, 1, \ldots, q_{\max}$}{
    Prepare $M$ index registers ($p = 3$ qubits each) in uniform superposition\;
    Apply the classical phase corrections $R(\tilde{u}_j^{(q)})$ to the index registers\;
    Apply controlled-$F$ $\bigO(2^q\sqrt{M})$ times to write phases $\propto \nabla f$ onto them\;
    Inverse-QFT and measure each register\; 
    Repeat $\bigO(\log M)$ times and take the median to obtain gradient estimates $g_j^{(q)}$\tcp*{$\bigO(2^q\sqrt{M}\log M)$ queries to $F$}
    Update $\tilde{u}_j^{(q+1)} \gets \tilde{u}_j^{(q)} + \pi\, 2^{-q}\, g_j^{(q)}$ for all $j$\tcp*{refine estimates}
}
\Return $\hat{o}_j := \tilde{u}_j^{(q_{\max}+1)} \approx \bra{\psi_0}O_j(t_j)\ket{\psi_0}$ via \cref{eq:gradient_echo}\;
\end{algorithm}

\subsection{Lower bound for Hamiltonian snapshots estimation}

\begin{proof}[Proof sketch of \cref{thm:multi_lower_bound}]
    The bound chains three reductions: multidimensional amplitude estimation~\cite{vanapeldoornQuantumProbabilityOracles2021} $\to$ multi-observable estimation~\cite{hugginsNearlyOptimalQuantum2022} $\to$ multi-time snapshot estimation.

    \emph{Prior lower bounds.}
    Apeldoorn~\cite{vanapeldoornQuantumProbabilityOracles2021} showed that for a hard sign matrix $A \in \{-1,+1\}^{M\times M}$, 
    computing $\tilde{\vbq}$ with $\|A\vbp - \tilde{\vbq}\|_\infty \le \eps$ from a probability oracle for $\vbp \in \Delta^M$ requires $\Omega(\sqrt{M}/\eps)$ coherent queries.
    Huggins et al.~\cite{hugginsNearlyOptimalQuantum2022} promoted this to multi-observable estimation: there exist Pauli observables $\{Z_j\}_{j=1}^{M}$ and a state-preparation oracle $U_\psi$ such that estimating every $\bra{\psi}Z_j\ket{\psi}$ to additive error $\eps$ requires $\Omega(\sqrt{M}/\eps)$ queries to $U_\psi$.

    \emph{Reduction to multi-time snapshots.}
    We encode the $M$ independent observables $\{Z_j\}$ as a single $1$-local observable $P = \op{0}_1 \otimes I$ probed at $M$ different times.
    Take $n = M+1$ qubits and let $\ham = H_\Pi$ generate the cyclic left-shift $\Pi\ket{x_1,\ldots,x_n} = \ket{x_2,\ldots,x_n,x_1}$, so $e^{-\ii H_\Pi} = \Pi$.
    At integer times $t_j = j$, the Heisenberg-picture conjugation rotates $P$ across the register,
    \begin{align}\label{eq:method_cyclic_shift}
        e^{\ii H_\Pi t_j}\, P\, e^{-\ii H_\Pi t_j} = \Pi^{-j}\, P\, \Pi^j = \op{0}_{j+1} \otimes I,
    \end{align}
    so $o(t_j) = \bra{\psi_0}\op{0}_{j+1}\ket{\psi_0}$ exposes the $(j{+}1)$-th $\op{0}$-projector, and via $Z_{j+1} = 2\op{0}_{j+1} - I$ each $\bra{\psi_0}Z_{j+1}\ket{\psi_0}$.
    Estimating all $M$ snapshots to precision $\eps/2$ therefore solves multi-observable estimation to precision $\eps$, forcing $\Omega(\sqrt{M}/\eps)$ queries.

    \emph{Local Hamiltonian construction.}
    $H_\Pi$ is itself non-local---its spectral projectors $\mathcal{P}_k$ act on global momentum eigenspaces---so the bound above is existential over general Hamiltonians.
    A constructive $\bigO(1)$-local hard instance follows by embedding $V = \Pi^M$ (a depth-$M^2$ circuit of nearest-neighbour SWAPs) into a Feynman--Kitaev clock Hamiltonian $H_{\mathrm{FK}}$ on $\bigO(M^2)$ qubits.
    Choosing time points $\sin^2(t_m) = (2m{-}1)/(2M)$ centres the binomial clock weight on a distinct $M$-gate block per snapshot;
    the resulting block-weight matrix is strictly diagonally dominant, so a constant-conditioned inversion recovers each Pauli expectation from the snapshots, transferring the $\Omega(\sqrt{M}/\eps)$ bound to the local setting.
    See \apdcref{apd:sec:quantum_lower_bounds} for the complete proof.
\end{proof}

\ifnum\onlymaintext=0
\clearpage
\onecolumngrid

\appendix
\setcounter{secnumdepth}{3}  
\begin{center}
{\bf \large Supplemental Material: \it \mytitle} 
\end{center}
\renewcommand{\addcontentsline}{\oldacl}
\renewcommand{\tocname}{Appendix Contents}
\tableofcontents

\numberwithin{theorem}{section}
\numberwithin{lemma}{section}
\numberwithin{remark}{section}
\numberwithin{corollary}{section}
\numberwithin{proposition}{section}
\numberwithin{definition}{section}
\numberwithin{problem}{section}

\section{Hardness of deciding DQPTs}\label{apd:sec:hardness}
In this appendix section, we provide the full details of the hardness results of DQPT related problems:
the \gapP-hardness of estimating the rate function for global DQPT (\cref{apd:thm:rate_hardness}) and the \BQP-completeness of deciding subsystem DQPT (\cref{thm:bqp_complete_local_dqpt}).

\subsection{\gapP-\hardness{} of Estimate-DQPT}\label{apd:sec:sharp}
We first establish the computational hardness of estimating (global) DQPT by reducing from the problem of computing the gap of degree-3 polynomials over $\mathbb{F}_2$, which is known to be \nameref{def:gap_p}-hard.
For this purpose, we establish the IQP Hamiltonian construction connecting the Loschmidt amplitude to the normalized gap (\cref{apd:lem:iqp_hamiltonian}), 
and combine them with an error analysis (\cref{lem:error_amplification}) to prove the main theorem (\cref{apd:thm:rate_hardness}).
We first recall the relevant complexity-theoretic definitions about the classes $\sharpP$ and $\gapP$.

\begin{definition}[\sharpP]\label{def:sharp_p}
    A function $f:\qty{0,1}^*\to\mathbb{N}$ is in $\sharpP$ if there exists a polynomial $q$ and a polynomial-time classical Turing machine $M$ such that for every $x\in\qty{0,1}^*$,
    \begin{equation}
        f(x) = \# \qty{y\in \qty{0,1}^{q(\abs{x})}: M(x,y)=1}.
    \end{equation}
    In other words, $\sharpP$ is the class of functions that count the number of accepting witnesses for an $\NP$ verification procedure.
\end{definition}

While $\sharpP$ counts nonnegative integers, many natural computational problems involve computing \emph{differences} of counts, which may be negative.
This motivates the definition of $\gapP$, a more general class that allows signed integer outputs.

\begin{definition}[\gapP]\label{def:gap_p}
    The function class $\gapP$ consists of all functions
    $f:\qty{0,1}^*\to \mathbb{Z}$ for which there exist $g,h\in \sharpP$ such that $f=g-h$.
    A problem is $\gapP$-$\hard$ if every function in $\gapP$ can be computed via a polynomial-time Turing reduction to it.
    A problem is $\gapP$-$\complete$ if it is both $\gapP$-$\hard$ and in $\gapP$.
\end{definition}

\begin{definition}[Gap of Boolean functions]\label{def:ngap}
    For a Boolean function $f:\qty{0,1}^n\to\qty{0,1}$, the \emph{gap} is defined as
    \begin{equation}
        \gap(f) := \abs{\qty{x:f(x)=0}} - \abs{\qty{x:f(x)=1}} = N_0 - N_1,
    \end{equation}
    where $N_0 = \abs{\qty{x:f(x)=0}}$ and $N_1 = \abs{\qty{x:f(x)=1}}$.
    The \emph{normalized gap} is $\ngap(f):=\frac{1}{2^n}\gap(f)$.
\end{definition}
The relevance of these definitions to DQPT is that the Loschmidt amplitude at specific time $t=\pi$ equals the normalized gap $\ngap(f)$ (\cref{apd:lem:iqp_hamiltonian}).
Since estimating the rate function to additive precision $\bigO(1/n)$ yields $|\ngap(f)|$ to multiplicative error, the hardness of the latter implies the hardness of the former.
We now recall the complexity-theoretic status of $\ngap$.

\begin{lemma}[\gapP-hardness of approximating $|\ngap|$ \cite{bremnerAveragecaseComplexityApproximate2016}]\label{lem:ngap_gapp}
    Let $f:\{0,1\}^n \to \{0,1\}$ be a degree-3 polynomial over $\mathbb{F}_2$.
    Computing $|\ngap(f)|$ to within multiplicative error $\eps < 1/2$ is \nameref{def:gap_p}-\hard.
\end{lemma}

This is Proposition~8 in Bremner et al.~\cite{bremnerAveragecaseComplexityApproximate2016}.
Since $\ngap(f)$ is real-valued (possibly negative), an oracle for $|\ngap(f)|$ does not directly reveal the sign.
Nevertheless, Bremner et al.~\cite{bremnerAveragecaseComplexityApproximate2016} show that the sign can be recovered via an adaptive bisection: by constructing shifted polynomials $f_c$ with $\ngap(f_c) = \tfrac{1}{2}(\ngap(f) - c)$ and querying $|\ngap(f_c)|$, one iteratively narrows the estimate until $\ngap(f)$ is determined exactly in $\bigO(n)$ rounds.
This is a polynomial-time Turing reduction, establishing $\gapP$-hardness of multiplicatively approximating $|\ngap(f)|$.

\subsubsection{IQP cirucit, 3-local Hamiltonian and degree-3 polynomial gap}
We now give the IQP \cite{bremnerClassicalSimulationCommuting2011,montanaroQuantumCircuitsLowdegree2017} Hamiltonian construction stated in the main text.

\begin{lemma}[IQP Hamiltonian construction]\label{apd:lem:iqp_hamiltonian}
    Let $f:\{0,1\}^n \to \{0,1\}$ be a degree-3 polynomial over $\mathbb{F}_2$.
    There exists a 3-local Hamiltonian $H$ on $n$ qubits, constructible in
    $\poly(n)$ time from the coefficients of $f$, such that
    \begin{equation}
        \losamp(t) := \langle 0^n | e^{-\ii H t} | 0^n \rangle
        = \frac{1}{2^n} \sum_{x \in \{0,1\}^n} e^{-\ii t f(x)}.
    \end{equation}
    In particular, $\losamp(\pi) = \ngap(f)$.
\end{lemma}

\begin{proof}
First, we create a problem Hamiltonian \(H_f\) that is diagonal in the computational (Z) basis and whose eigenvalues are the values of the polynomial \(f(x)\).
We achieve this by replacing each binary variable \(x_i\) in the polynomial with the projection operator \(P_i = (I - Z_i)/2\), where \(Z_i\) is the Pauli-Z operator on qubit \(i\).
This projector has eigenvalues of 0 for state \(\ket{0}\) and 1 for state \(\ket{1}\), just like the classical bit \(x_i\).
This gives the Hamiltonian:
\begin{equation}\label{eq:Hf}
    H_f = \sum_{i,j,k} \alpha_{ijk} P_i P_j P_k + \sum_{i,j} \beta_{ij} P_i P_j + \sum_{i} \gamma_i P_i
\end{equation}
This Hamiltonian is \(k\)-local for a \(k\)-degree polynomial because each term acts on at most \(k\) qubits. For our degree-3 polynomial, it is 3-local. When acting on a computational basis state \(\ket{x} = \ket{x_1, \dots, x_n}\), its eigenvalue is \(f(x)\):
\begin{equation}
H_f\ket{x} = f(x)\ket{x}
\end{equation}

The IQP circuit applies Hadamard gates to the input and output to create and interfere superpositions.
We achieve the equivalent effect by conjugating our problem Hamiltonian \(H_f\) with Hadamard gates, which rotates the basis from Z to X:
\begin{equation}\label{eq:H_final}
    H_{\IQP} = \hdm^{\otimes n} H_f \hdm^{\otimes n}
\end{equation}
Since \(\hdm Z_i \hdm = X_i\), this is equivalent to replacing every \(Z_i\) operator in \(H_f\) with a Pauli-X operator.
The resulting Hamiltonian \(H_{\IQP}\) is also 3-local, as it is a sum of terms involving products of at most three Pauli-X operators.

Now, the amplitude $\expval{e^{-\ii H_{\IQP}t}}{0}$ can express the gap of any degree-3 polynomial.
Setting $t=\pi$ and using $e^{-\ii\pi f(x)} = (-1)^{f(x)}$:
\begin{equation}
    \expval{e^{-\ii H_{\IQP}\pi}}{0}
    = \frac{1}{2^n} \sum_x e^{-\ii \pi f(x)}
    = \frac{1}{2^n} \sum_x (-1)^{f(x)}
    = \frac{1}{2^n} (N_0-N_1)
    = \frac{1}{2^n} \gap(f)
    = \ngap(f).
\end{equation}
\end{proof}

\begin{remark}[Ising variant: a 2-local Hamiltonian with multiple hard times]\label{apd:rmk:ising_variant}
    The Boolean degree-3 polynomial input to \cref{apd:lem:iqp_hamiltonian} admits an Ising-type substitute that lowers locality and produces multiple hard times~\cite{bremnerAveragecaseComplexityApproximate2016}.
    With integer edge and vertex weights $w_{ij}, v_i$ on the complete graph, define the 2-local transverse Ising Hamiltonian
    $\ham_I = \sum_{i<j} w_{ij}\, X_i X_j + \sum_i v_i\, X_i$.
    A Hadamard sandwich gives the imaginary-time partition function identity
    \begin{equation}\label{apd:eq:ising_amplitude}
        \langle 0^n | e^{-\ii\theta\,\ham_I} | 0^n \rangle = \frac{Z(e^{-\ii\theta})}{2^n}, \qquad Z(\omega) := \sum_{z\in\{\pm 1\}^n} \omega^{H_I(z)},
    \end{equation}
    where $H_I(z) = \sum_{i<j} w_{ij}\, z_i z_j + \sum_i v_i\, z_i$ is the classical Ising energy on spins $z \in \{\pm 1\}^n$.
    This expresses the Loschmidt amplitude as the classical partition function at imaginary temperature $\omega = e^{-\ii\theta}$---the same correspondence that underpins the standard DQPT framework~\cite{heylDynamicalQuantumPhase2013,heylDynamicalQuantumPhase2018}.
    Multiplicative-constant approximation of $|Z(e^{-\ii\theta})|^2$ is $\sharpP$-hard in the worst case at $\theta = \pi/8$ and other rational multiples of $\pi$~\cite{goldbergComplexityApproximatingComplexvalued2017,fujiiCommutingQuantumCircuits2017}, and conjectured average-case hard for random weights~\cite{bremnerAveragecaseComplexityApproximate2016}, so the $\gapP$-hardness of \cref{apd:thm:rate_hardness} carries over to this 2-local Ising instance at each such $\theta$.
\end{remark}

The following lemma formalizes how additive error in $\rate(\pi)$ translates to multiplicative error in $|\ngap(f)|$.

\begin{lemma}\label{lem:error_amplification}
    Let $H$ be the IQP Hamiltonian for a degree-3 polynomial $f$
    (\cref{apd:lem:iqp_hamiltonian}), so that $\losamp(\pi) = \ngap(f)$ and
    $\rate(\pi) = -(2/n)\log|\ngap(f)|$.
    If $\tilde{r}$ approximates $\rate(\pi)$ to additive error $\eps$,
    then $\exp(-n\tilde{r}/2)$ approximates $|\ngap(f)|$ to multiplicative
    error $e^{n\eps/2}$.
    In particular, $\eps = \bigO(1/n)$ yields $\bigO(1)$ multiplicative error.
\end{lemma}
\begin{proof}
    Since $\rate(\pi) = -(2/n)\log|\ngap(f)|$, we have $|\ngap(f)| = e^{-n\rate(\pi)/2}$.
    If $|\tilde{r} - \rate(\pi)| \le \eps$, then
    \begin{equation}
        e^{-n\eps/2} \le \frac{e^{-n\tilde{r}/2}}{|\ngap(f)|} \le e^{n\eps/2}.
    \end{equation}
    For $\eps = c/n$ with $c = \bigO(1)$, 
    the multiplicative factor is $e^{c/2} = \bigO(1)$.
\end{proof}

\subsubsection{Estimation of the global rate function is \gapP-\hard}
We now combine these lemmas to establish the GapP-hardness of the rate function for DQPT.
\begin{problem}[\estimate-\dqpt]\label{apd:prm:global_dqpt}
    Given an $n$-qubit local Hamiltonian $\ham$, an initial state $\ket{\psi_0}$,
    evolution time $t\in \poly(n)$,
    the problem is to estimate the rate function $\rate(t) = -(1/n)\log\los(t)$
    to error $\eps$.
\end{problem}
\begin{theorem}[\nameref{apd:prm:global_dqpt} is \mbox{\nameref{def:gap_p}}-hard]\label{apd:thm:rate_hardness}
    Let $f: \{0,1\}^n \rightarrow \{0,1\}$ be a degree-3 polynomial over $\mathbb{F}_2$, and let $H$ be the corresponding 3-local Hamiltonian (\cref{apd:lem:iqp_hamiltonian}).
    Computing the rate function at $t=\pi$,
    \begin{equation}
        \rate(\pi) = -\frac{2}{n} \log\abs{\ngap(f)},
    \end{equation}
    to additive error $\eps = \bigO(1/n)$ is \nameref{def:gap_p}-hard under polynomial-time Turing reductions.
\end{theorem}

\begin{proof}
We reduce from the $\gapP$-hard problem of multiplicatively approximating
$|\ngap(f)|$ (\cref{lem:ngap_gapp}).
Given a degree-3 polynomial $f$, construct the 3-local Hamiltonian $H$ in $\poly(n)$
time (\cref{apd:lem:iqp_hamiltonian}), so that $\losamp(\pi) = \ngap(f)$.
By \cref{lem:error_amplification}, an additive-$\eps$ approximation to
$\rate(\pi) = -(2/n)\log|\ngap(f)|$ yields a multiplicative-$e^{n\eps/2}$
approximation to $|\ngap(f)|$.
For $\eps = \bigO(1/n)$, this gives $\bigO(1)$ multiplicative error.
Since computing $|\ngap(f)|$ to within constant multiplicative error is $\gapP$-hard (\cref{lem:ngap_gapp}),
computing $\rate(\pi)$ to additive error $\bigO(1/n)$ is $\gapP$-hard.
\end{proof}

\begin{remark}[Containment in $\mathrm{FP}^{\gapP}$]\label{apd:rmk:containment}
    The class $\mathrm{FP}$ consists of all functions $f:\{0,1\}^*\to\mathbb{Z}$ computable by a deterministic polynomial-time Turing machine, and $\mathrm{FP}^{\gapP}$ denotes the class of functions computable in polynomial time with oracle access to a $\gapP$ function.
    The estimation problem itself sits in $\mathrm{FP}^{\gapP}$: $\gap(f) = N_0 - N_1$ with $N_0, N_1 \in \sharpP$, so $\gap(f)$ is the canonical $\gapP$ function by definition~\cite{fennerGapdefinableCountingClasses1994}, and a single $\gapP$-oracle call computes it exactly, from which $\rate(\pi) = -(2/n)\log|(N_0 - N_1)/2^n|$ follows in polynomial time.
    Combined with \cref{apd:thm:rate_hardness}, this places $\eps = \bigO(1/n)$ estimation of $\rate(\pi)$ in $\mathrm{FP}^{\gapP}$ and $\gapP$-hard---the natural completeness statement under the loose convention of \cite{bremnerAveragecaseComplexityApproximate2016}.
\end{remark}

\begin{remark}[Constant-precision hardness for DQPT decision]\label{apd:rmk:constant_precision_open}
    The $\gapP$-hardness of \cref{apd:thm:rate_hardness} applies at additive precision $\eps = \bigO(1/n)$, the regime in which \cref{lem:error_amplification} converts to constant multiplicative error in $|\ngap(f)|$.
    Yet \emph{deciding} a global DQPT only requires \emph{constant} additive precision $\eps = \Theta(1)$, since the magnitude condition $\rate(t_c) \ge \xi$ and the cusp condition $-\dratem(t_c)\cdot\dratep(t_c) \ge \eta$ are gated by constants $\xi, \eta = \Theta(1)$.
    At constant $\eps$, \cref{lem:error_amplification} inflates the multiplicative window on $|\ngap(f)|$ to $e^{\Theta(n)}$, which spans essentially the full admissible range $|\ngap(f)| \in [0,1]$ and conveys no $\gapP$-information; any reduction routed through $|\ngap(f)|$ is therefore information-theoretically incapable of yielding hardness at this physically motivated precision.
    Whether constant-precision estimation of $\rate(\pi)$ is classically hard---via a fundamentally different mechanism (e.g., a $\BQP$-style argument or one exploiting the phase rather than the magnitude of $\losamp(\pi)$)---remains open with the current proof technique.
\end{remark}

The $\gapP$-hardness of \nameref{apd:prm:global_dqpt} is a barrier even for quantum computers.
A quantum computer can efficiently prepare the evolved state via Hamiltonian simulation, 
but deciding whether it is a critical time of DQPT 
(i.e. estimating the rate function to the certain precision) 
is likely intractable to quantum computers,
which may require exponential samples to resolve.
To show a tractable problem where quantum advantage is provable, 
the next section formulates the \emph{subsystem} variant of DQPT---a promise decision problem with constant gaps that is $\BQP$-\complete.

\subsection{\BQP-\hardness{} of Local-DQPT}\label{apd:sec:bqp}
In this section we prove that \nameref{apd:prm:promise_local_dqpt} is \nameref{def:bqp_complete} (\cref{thm:bqp_complete_local_dqpt}), 
which implies a practical and provable quantum advantage.
A problem is \BQP-\complete{} if it is in $\BQP$ and $\BQP$-\hard.
We establish membership in $\BQP$ in \cref{apd:thm:in_bqp} and \BQP-\textsf{hardness} in \cref{apd:thm:bqp_hard}.
We first collect the necessary complexity-theoretic definitions and the \nameref{def:bqp_complete} problem from which we reduce.
Informally, the class of computational problems that are solvable in polynomial time by quantum computers with bounded error, 
is denoted as $\BQP$.
\begin{definition}[\BQP]\label{def:bqp}
    Formally, a language $L\in \BQP$ if and only if there exists a polynomial-size quantum circuit $Q$ such that
	\begin{align}
		x\in L_{\yes} & \Longleftrightarrow \Probability[Q(x)=1]\ge 2/3 \\
		x\in L_{\no} & \Longleftrightarrow \Probability[Q(x)=1]\le 1/3
	\end{align}
    where $Q(x)$ denotes the outcome of measuring the first qubit after applying $Q$ to $\ket{x}\ket{0^m}$.
\end{definition}
The constants $1/3$ and $2/3$ are just conventional 
because the success probability can be amplified to be close to 1 by repeating the algorithm and taking the majority vote.
\begin{definition}[\BQP-\hard]\label{def:bqp_hard}
	A \BQP-\hard{} problem $\mathcal{P}$ is one with the property that any problem in $\BQP$ can be efficiently translated into an instance of $\mathcal{P}$, 
    so that the answer to the instance of $\mathcal{P}$ gives the answer to the original problem in $\BQP$.
\end{definition}
\begin{definition}[\BQP-\complete]\label{def:bqp_complete}
	A problem is $\BQP$-$\complete$ if it is both \nameref{def:bqp_hard} and contained in \nameref{def:bqp}.
\end{definition}
The following problem definition are the standard \BQP-$\complete$ problem of approximating the acceptance probability of a quantum circuit, 
which we will reduce from.
\begin{problem}[\aqcp]\label{prm:approx_qcircuit_prob}
    Given a polynomial-size quantum circuit $C$ acting on $n$ qubits and each gate acts on one or two qubits, 
    distinguish between the following two cases:   
    \begin{itemize}
        \item \yes: measuring the first qubit of the state $C\ket{0}^{\otimes n}$ yields $\ket{1}$ with probability $\ge 2/3$,
        i.e., $\expval{C^\dagger\op{1}_0 C}{\vb{0}}\ge 2/3$
        \item \no: 
        $\expval{C^\dagger\op{1}_0 C}{\vb{0}}\le 1/3$
    \end{itemize}
\end{problem}
\nameref{prm:approx_qcircuit_prob} is \nameref{def:bqp_complete} by definition.

\subsubsection{The standard BQP-complete problem: Overlap of quantum circuits}
The next problem is a variant of \nameref{prm:approx_qcircuit_prob} where the the output 1-qubit is replaced by a $k$-qubit subsystem, and the acceptance probability is amplified to $1-e^{-\Omega(k)}$.
We formally define this problem.
\begin{problem}[\promise-\koverlap]\label{prm:k_overlap}
    Given a quantum circuit $U$ on $n$ qubits with $L = \poly(n)$ gates,
    a subsystem size $k\in[n]$, and the $k$-local projector
    $\projk = \prod_{j=1}^k \op{0}_j$, 
    let $a=\bra{0^n} U^\dagger \projk\, U \ket{0^n}$ be the overlap of the output state with the $k$-qubit all-zero initial state,
    decide:
    \begin{itemize}
        \item \yes: $a
              \ge 1 - e^{-\Omega(k)}$
        \item \no:  $a
              \le e^{-\Omega(k)}$
    \end{itemize}
    promised that one of these holds.
\end{problem}
\nameref{prm:k_overlap} is also \nameref{def:bqp_complete} by circuit repetition and majority vote (gap amplification).
Later, we will use this problem to prove the \BQP-hardness of determining local DQPTs.
\begin{lemma}\label{thm:k_overlap}
    \nameref{prm:k_overlap} is \nameref{def:bqp_complete}
    for any constant $k \ge 1$.
\end{lemma}
\begin{proof}
    We reduce from \nameref{prm:approx_qcircuit_prob} 
    \cite{babbushExponentialQuantumSpeedup2023}.
    Given a circuit $C$ on $n_0$ qubits whose first output qubit
    accepts with probability $p\ge 2/3$ (\yes) or $p\le 1/3$ (\no),
    we construct a circuit $U$ on $n = \bigO(n_0)$ qubits as follows.

    \emph{Gap amplification.}
    Set $m = 2k+1$ (always odd, matching the subsystem size).
    Run $m$ independent copies of $C$ in parallel on $m \cdot n_0$
    qubits.  
    Compute the majority of the $m$ first output qubits
    $q_1,\dots,q_m$ into an ancilla bit~$a$ using a reversible
    circuit $V_{\maj}$
    (a reversible adder with threshold comparison,
    using $\bigO(m\log m)$ gates and $\bigO(\log m)$ ancilla qubits).
    By the Chernoff bound, the majority equals the correct answer
    with probability $\ge 1 - e^{-\Omega(m)} = 1 - e^{-\Omega(k)}$.

    \emph{Construction of $U$.}
    Introduce $k$ fresh ancilla qubits at positions $1,\dots,k$,
    all initialized to $\ket{0}$.
    The circuit~$U$ acts as:
    (i)~apply $C^{\otimes m}$ to the $m$ register blocks;
    (ii)~apply $V_{\maj}$ to compute the majority bit into~$a$;
    (iii)~\textsc{cnot} from $a$ to each of positions $1,\dots,k$;
    (iv)~apply $V_{\maj}^\dagger$ to uncompute $a$
    and all majority-vote ancillas;
    (v)~apply $(C^\dagger)^{\otimes m}$ to uncompute the registers.
    The total qubit count is
    $n = k + m\,n_0 + \bigO(m\log m) = \bigO(n_0)$
    (since $k = \bigO(1)$ and $m = \bigO(k) = \bigO(1)$),
    so $n_0 = \Theta(n)$
    and the Chernoff error is $e^{-\Omega(k)}$.

    \emph{Analysis.}
    After step~(iv) the global state is
    \begin{equation}
        \sum_{x\in\{0,1\}^m} \alpha_x\,
        \ket{x}_q\ket{\Phi_x}_{\mathrm{rest}}
        \,\ket{0}_A
        \,\ket{\maj(x)}^{\otimes k}_{1\dots k},
    \end{equation}
    where $q = (q_1,\dots,q_m)$ is the register of the $m$ first output qubits,
    $A$ is the majority-vote ancilla register (the bit $a$ together with the
    $\bigO(\log m)$ workspace qubits of $V_{\maj}$), restored to $\ket{0}$ by
    step~(iv),
    $\alpha_x = \prod_{i=1}^m (\sqrt{p})^{x_i}(\sqrt{1\!-\!p})^{1-x_i}$,
    and $\ket{\Phi_x} = \bigotimes_i \ket{\phi_{x_i}^{(i)}}$ collects
    the non-output qubits of each copy.
    Step~(v) applies $(C^\dagger)^{\otimes m}$ to the registers;
    since $\projk = \prod_{j=1}^k \op{0}_j$ acts only on positions
    $1,\dots,k$ (the fresh ancillas), the register state is irrelevant.  
    Therefore
    \begin{equation}
        \bra{0^n} U^\dagger \projk\, U \ket{0^n}
        = \sum_{x:\,\maj(x)=0} \abs{\alpha_x}^2
        = \Pr\qty[\maj(X)=0],
    \end{equation}
    where $X_i \sim \mathrm{Bernoulli}(p)$ independently.
    For \no{} instances ($p \le 1/3$):
    $\Pr[\maj=0] \ge 1 - e^{-\Omega(k)}$.
    For \yes{} instances ($p \ge 2/3$):
    $\Pr[\maj=0] \le e^{-\Omega(k)}$.
    The $e^{-\Omega(k)}$ gap originates from
    Chernoff concentration over $m = \Theta(k)$ copies;
    all $k$ answer positions carry copies of the single majority bit,
    so the gap is independent of how the $k$-local projector $\projk$ acts.
    The \yes{}/\no{} labels are swapped: a \yes{}
    instance of \mbox{\nameref{prm:approx_qcircuit_prob}}
    has first-qubit output $\ket{1}$ w.h.p., so the majority
    bit is~$1$ and the overlap with $\ket{0}^{\otimes k}$ is small,
    making it a \no{} instance of \mbox{\nameref{prm:k_overlap}}.
    This is valid since $\BQP = \textsf{co}\BQP$.

    \emph{Membership.}
    Run $U\ket{0^n}$ and measure positions $1,\dots,k$;
    the outcome decides the problem in $\poly(n)$ time.
\end{proof}

The proof of \cref{thm:k_overlap} constructs $U$ with a specific layout, whose properties the hardness proof uses; we record them as a structural fact.

\begin{fact}[Structure of the canonical \koverlap{} circuit]\label{apd:lem:koverlap_structure}
The circuit $U$ produced by the proof of \cref{thm:k_overlap} satisfies:
\begin{itemize}
    \item the answer positions $1,\dots,k$ are fresh ancillae initialized to $\ket{0}$, and the only gates of $U$ acting on them are the $k$ CNOTs of the fan-out block,
    (a single contiguous layer, $k = \bigO(1)$ wide);

    \item the fan-out copies a single classical bit, the majority bit of the $m = 2k+1$ circuit repetitions, 
    so that after any $i \ge 1$ fan-out CNOTs the positions hold $\ket{\maj}^{\otimes i}\ket{0}^{\otimes(k-i)}$;

    \item consequently $\bra{0^n}U^\dagger \projk U\ket{0^n} = \Pr[\maj = 0] =: a$, with $a \ge 1-\eps$ for accepting and $a \le \eps$ for rejecting instances, where $\eps = e^{-\Omega(k)}$ is the amplified overlap error.
\end{itemize}
\end{fact}

\subsubsection{Feynman--Kitaev circuit-to-Hamiltonian construction}
The key technical ingredient bridging quantum circuits and Hamiltonian dynamics is the Feynman--Kitaev construction, which encodes a circuit $U$ into a local Hamiltonian whose time evolution reproduces the computation.

\begin{lemma}[Feynman-Kitaev Circuit-to-Hamiltonian construction]\label{thm:circuit_hamiltonian}
    Given a quantum circuit $U=U_{N}\dots U_1$ composed of $N$ gates, 
    there exists a Hamiltonian $H_{U} \in \mathcal{H}_{\cc} \otimes \mathcal{H}_{\oo}$ acting on a clock register and an output register, where the clock uses $N$ qubits to unary-encode its $N+1$ legal states $\ket{0},\ldots,\ket{N}$ as $\ket{j}=\ket{1^j 0^{N-j}}$, with
    \begin{equation}\label{eq:circuit_hamiltonian}
        H_{U} = \sum_{j=1}^{N} \sqrt{j(N-j+1)} 
        \qty(\op{j}{j-1}_{\cc}\otimes U_j + \op{j-1}{j}_{\cc}\otimes U_j^\dagger),
    \end{equation}
    such that for any initial state $\ket{\psi_0}_{\oo}$, the time evolution yields a superposition over all intermediate circuit steps (a history state):
    \begin{equation}\label{eq:circuit_ham_evolution}
        e^{-\ii H_{U}t}\ket{0}_{\cc}\ket{\psi_0}_{\oo}
        = \sum_{j=0}^{N} c_j(t)\,\ket{j}_{\cc}\otimes U_j\cdots U_1\ket{\psi_0}_{\oo},
    \end{equation}
    where the clock amplitudes are
    \begin{equation}\label{eq:clock_amplitudes}
        c_j(t) = (-\ii)^j \binom{N}{j}^{1/2}\sin^{j}(t)\,\cos^{N-j}(t),
        \qquad
        \abs{c_j(t)}^2 = \binom{N}{j}\sin^{2j}(t)\,\cos^{2(N-j)}(t).
    \end{equation}
    In particular, at $t=\pi/2$ the clock deterministically reaches step $N$ and the full circuit is applied: $e^{-\ii H_{U}\pi/2}\ket{0}_{\cc}\ket{\psi_0}_{\oo}=(-\ii)^N\ket{N}_{\cc}\otimes U \ket{\psi_0}_{\oo}$, where the global phase $(-\ii)^N$ is irrelevant for the projector and echo arguments below.
\end{lemma}
\begin{proof}
    The proof strategy is to introduce an analytically solvable auxiliary system
    whose dynamics, when restricted to an invariant subspace, are spectrally
    identical to the clock sector of~$H_U$.
    Consider $N$ independent qubits with collective Pauli operator
    $\bar{X} = \sum_{j=1}^{N} X^{(j)}$.
    Since each qubit evolves independently via $e^{-\ii Xt} = \cos t\, I - \ii\sin t\, X$, the evolution from $\ket{0}^{\otimes N}$ factorizes as
    \begin{equation}\label{eq:product_evolution}
        e^{-\ii\bar{X}t}\ket{0}^{\otimes N}
        = \bigotimes_{j=1}^{N}\qty(\cos t\,\ket{0}_j - \ii\sin t\,\ket{1}_j).
    \end{equation}

    \emph{Hamming-weight decomposition.}
    Define the (normalized) Hamming-weight states
    $\ket{w_k} \coloneqq \binom{N}{k}^{-1/2}\sum_{\abs{x}=k}\ket{x}$
    for $k=0,\ldots,N$.
    Expanding \cref{eq:product_evolution} and collecting computational-basis states by Hamming weight, each weight-$k$ string picks up the coefficient $(-\ii\sin t)^k(\cos t)^{N-k}$.
    Since there are $\binom{N}{k}$ such strings,
    \begin{equation}\label{eq:hamming_projection}
        e^{-\ii\bar{X}t}\ket{0}^{\otimes N}
        = \sum_{k=0}^{N}(-\ii\sin t)^k (\cos t)^{N-k}\,\binom{N}{k}^{1/2}\,\ket{w_k}.
    \end{equation}

    \emph{Tridiagonal structure of $\bar{X}$.}
    The operator $X^{(j)}$ flips bit~$j$.
    For a string $\ket{x}$ with $\abs{x}=k$, there are $k$ positions where $x_j=1$ (flipping to weight $k-1$) and $N-k$ positions where $x_j=0$ (flipping to weight $k+1$).
    By a counting argument: each weight-$(k{-}1)$ string $y$ is produced from exactly $N-k+1$ distinct pairs $(x,j)$ with $\abs{x}=k$ and $x_j=1$, since $y$ has $N-k+1$ zero-positions that can be set to 1 to yield a weight-$k$ string.
    Similarly, each weight-$(k{+}1)$ string is produced from $k+1$ pairs.
    It follows that
    \begin{equation}\label{eq:Xbar_tridiag}
        \bar{X}\ket{w_k}
        = \sqrt{k(N{-}k{+}1)}\,\ket{w_{k-1}}
        + \sqrt{(k{+}1)(N{-}k)}\,\ket{w_{k+1}},
    \end{equation}
    where the prefactors arise from the ratios of binomial coefficients:
    $\binom{N}{k}^{-1/2}(N{-}k{+}1)\binom{N}{k{-}1}^{1/2} = \sqrt{k(N{-}k{+}1)}$, and similarly for the raising term~\cite{christandlPerfectStateTransfer2004,berryExponentialImprovementPrecision2014}.

    \emph{Spectral equivalence with $H_U$.}
    The payoff of the auxiliary construction is that $\bar{X}$
    decomposes into independent qubits (giving the closed-form
    evolution in \cref{eq:hamming_projection}), yet its restriction
    to the Hamming-weight subspace reproduces the clock dynamics
    of~$H_U$.
    The relevant invariant subspace of $H_U$ is the \emph{history subspace}
    $\mathcal{S}_{\psi_0} \coloneqq \operatorname{span}\{\ket{\alpha_k}\}_{k=0}^{N}$
    spanned by the history states
    $\ket{\alpha_k} \coloneqq \ket{k}_\cc \otimes U_k\cdots U_1\ket{\psi_0}_\oo$,
    on which the off-diagonal $U_k,U_k^\dagger$ blocks advance and rewind the
    output register in lockstep with the clock:
    \begin{equation}\label{eq:HU_history_tridiag}
        H_U\ket{\alpha_k}
        = \sqrt{(k{+}1)(N{-}k)}\,\ket{\alpha_{k+1}}
        + \sqrt{k(N{-}k{+}1)}\,\ket{\alpha_{k-1}}.
    \end{equation}
    Comparing \cref{eq:HU_history_tridiag} with \cref{eq:Xbar_tridiag}, the
    $(N{+}1)\times(N{+}1)$ tridiagonal matrix of $H_U\big|_{\mathcal{S}_{\psi_0}}$
    in the basis $\{\ket{\alpha_k}\}$ is identical to that of
    $\bar{X}\big|_{\{\ket{w_k}\}}$ under the identification
    $\ket{w_k}\leftrightarrow\ket{\alpha_k}$.
    Therefore the clock amplitudes under $H_U$ are the same as those in \cref{eq:hamming_projection}: defining $c_k(t) \coloneqq (-\ii)^k\binom{N}{k}^{1/2}\sin^k(t)\,\cos^{N-k}(t)$, we obtain
    $\abs{c_k(t)}^2 = \binom{N}{k}\sin^{2k}(t)\,\cos^{2(N-k)}(t)$,
    establishing \cref{eq:clock_amplitudes}.

    \emph{Gate application.}
    By the structure of $H_U$ (\cref{eq:circuit_hamiltonian}), each clock transition $\ket{k{-}1}_\cc\to\ket{k}_\cc$ is accompanied by $U_k$ acting on the output register.
    Hence the full evolution is $e^{-\ii H_Ut}\ket{0}_\cc\ket{\psi_0}_\oo = \sum_{k=0}^{N}c_k(t)\,\ket{k}_\cc\otimes U_k\cdots U_1\ket{\psi_0}_\oo$, establishing \cref{eq:circuit_ham_evolution}.
    At $t=\pi/2$, $\sin(\pi/2)=1$ and $\cos(\pi/2)=0$, so $\abs{c_k}^2=\delta_{k,N}$ and the full circuit $U=U_N\cdots U_1$ is applied with certainty up to global phases.
\end{proof}

\subsubsection{The construction of palindromic circuits}\label{apd:sec:simple_construction}
We now prove the $\BQP$-\hardness{} of \nameref{apd:prm:promise_local_dqpt} with the following formal definition.
It asks whether the local rate function exceeds a threshold $\xi_1$ at time $t$ (exponentially suppressed echo) while the dynamical order parameter $\dkrate$ drops by at least $J_1$ across the window $[t-\deltat,\, t+\deltat]$: an integrated susceptibility spike.

\begin{problem}[\local-\dqpt]\label{apd:prm:promise_local_dqpt}
    Given a local Hamiltonian $H$ on $n$ qubits, an initial state $\ket{\psi_0}$, an evolution time $t$,
    a window $\deltat \in \Omega(1/\poly(n))$ with $t \pm \deltat \in [0,\, \poly(n)]$, a subsystem size $k$,
    magnitude thresholds $\xi_1 \ge \xi_0 \ge 0$, and susceptibility thresholds $J_1 > J_0 \ge 0$,
    define the \emph{integrated susceptibility}
    \begin{equation}\label{apd:eq:jump_def}
        J_k(t) \;:=\; \dkrate(t - \deltat) - \dkrate(t + \deltat)
        \;=\; -\int_{t-\deltat}^{t+\deltat} \ddkrate(s)\, ds .
    \end{equation}
    The \local-\dqpt{} problem is the promise problem to determine
    \begin{itemize}
        \item \yes: $\krate(t) \ge \xi_1$ and $J_k(t) \ge J_1$
        (i.e., $t$ is a critical time: deep echo suppression and a susceptibility spike),
        \item \no: $\krate(t) \le \xi_0$ or $J_k(t) \le J_0$,
    \end{itemize}
    with the promised gaps $J_1/J_0 \ge c$ for some constant $c > 1$
    and $\xi_1 - \xi_0 \ge \Omega(1)$ with $\xi_1 \in \Omega(1)$,
    and promised that 
    $\krate(t \pm \deltat) \le R_0$ 
    for a constant $R_0$ 
    (well-conditioned probes).
\end{problem}

The quantity $J_k$ is the drop of the time-derivate of the rate function $\dkrate$ (order parameter) across the window, 
equivalently the window integral of the dynamical susceptibility $-\ddkrate$:
a cusp of $\krate$ inside the window contributes the full weight of its $\delta$-function singularity, 
so at a generic kink $J_k$ equals the slope discontinuity up to $\bigO(\deltat)$ regular corrections, independently of the window.
The admissibility promise $\deltat\in \Omega(1/\poly(n))$ plays the same role as the resolution assumption $\deltat_{\min} \in \Omega(1/\poly(n))$ of \textsc{Search}-\dqpt{}. 
The constant $R_0$ is consumed by the membership algorithm rather than by the hardness proof: 
it guarantees probe echoes $\los_k(t \pm \deltat) \ge e^{-kR_0} = \Omega(1)$, 
so the commutator-based estimation of $\dkrate(t \pm \deltat)$, and hence of $J_k$, has constant sample cost; 
the construction of \cref{apd:thm:bqp_hard} meets it with $R_0 = 1$.

To show \nameref{apd:prm:promise_local_dqpt} is \nameref{def:bqp_hard}, 
we reduce from \nameref{prm:k_overlap}.
We build a palindromic circuit whose Feynman--Kitaev clock (\cref{thm:circuit_hamiltonian}) sweeps a fresh $k$-qubit answer register past the fanned-out majority bit of the \koverlap{} instance: 
at the clock centre $t^* = \pi/4$ the echo equals the suppressed overlap, 
and a tunable block of \emph{idle} gates holds the answer live over a window in which the order parameter jumps, 
so $t^*$ is a critical time of \local-\dqpt{} exactly when the instance rejects.
The construction \emph{computes} the majority, 
\emph{holds} the fanned-out answer through the idle block, and \emph{uncomputes}; 
and centres the critical time $t^* = \pi/4$ automatically.

We build the reduction in three pieces: the answer-writing circuit $W$, the idle hold that fixes the window, and the Feynman--Kitaev clock that turns the palindrome into a local Hamiltonian; we take each in turn.
Let
\begin{equation}\label{apd:eq:W_simple}
    W \;:=\; F\, V_{\maj}\, C^{\otimes m}
\end{equation}
be the \emph{compute-and-fan-out} circuit of \cref{thm:k_overlap}: 
it runs $m = 2k+1$ copies of the base circuit, 
computes the majority bit into an ancilla with the reversible circuit $V_{\maj}$,
and fans that bit out by \textsc{cnot} to the $k$ fresh answer positions,
in $\ell := |W| = \poly(n)$ gates and with \emph{no} uncompute.
By the layout-independent content of \cref{apd:lem:koverlap_structure} (the fan-out copies the single majority bit, and no other gate touches the answer positions)---here that fan-out is the \emph{final} block of $W$, not the middle block of the doubled circuit---the answer positions are fresh ancillae touched only by the fan-out, and
\begin{equation}\label{apd:eq:W_overlap}
    a:=\;\bra{0^n} W^\dagger \projk\, W \ket{0^n} \;=\; \Pr[\maj = 0] ,
\end{equation}
the \koverlap{} value ($a \le \eps$ for accepting, $a \ge 1-\eps$ for rejecting instances).

Given a target window $\deltat$, choose an idle count $w$ (fixed by \cref{apd:lem:idle_dial} below) and form the palindrome by inserting the $w$ \emph{idle gates at the fold},
between the fan-out of $W$ and the reset fan-out of $W^\dagger$:
\begin{equation}\label{apd:eq:V_simple}
    V' \;=\; \bigl(\, W,\; \underbrace{I, \ldots, I}_{w \text{ idle gates}},\; W^\dagger \,\bigr),
    \qquad N' := 2\ell + w \text{ gates},
\end{equation}
so that $V'_{N'} \cdots V'_1 = W^\dagger W = I$.
See the circuit in \cref{apd:fig:simple_schematic}.
The idle gates are computationally inert, but they \emph{hold} the fanned-out majority bit on the answer positions for the $w$ idle steps, so $\langle\projk\rangle_j = a$ across the whole answer segment $B = \{\ell, \ldots, \ell+w\}$ rather than only near the fan-out; widening $B$ this way is exactly the dial that sets the window $\deltat$ from $w$.
Without it ($w = 0$) the answer segment collapses to the single fan-out step $B = \{\ell\}$, of width $\bigO(k)$, and the binomial clock $X_t \sim \mathrm{Bin}(N', \sin^2 t)$, whose spread $\sigma = \sqrt{N'\sin^2 t\cos^2 t} = \Theta(\sqrt{N'})$ is irreducible, can place only
$\Pr[X_t \in B] \le \bigO(k)\,\max_j \bpmf(j; N', \sin^2 t) \le \bigO(k/\sigma) = \bigO(k/\sqrt{N'})$
of its weight on $B$ at \emph{any} time (\cref{apd:fct:binomial}(ii)).
The echo formula \cref{apd:eq:exact_echo_s} then gives $\los_k(t) = 1 - (1-a)\Pr[X_t \in B] \ge 1 - \bigO(1/\sqrt{N'})$, so the rate function $\krate(t) = -\tfrac1k\log\los_k(t) = \bigO(1/\sqrt{N'})$ stays negligible \emph{uniformly} in $t$:
it never develops an $\Omega(1)$ peak or a cusp.
Even at $t^*$, where the clock mean meets the lone answer step $\ell = N'/2$, that step (a single $\bigO(k)$-wide site) captures only an $\bigO(1/\sqrt{N'})$ fraction of the $\sqrt{N'}$-wide clock, so $\krate(t^*) = \bigO(1/\sqrt{N'})$ and the magnitude promise $\krate(t^*) \ge \xi_1 = \Omega(1)$ fails: the reduction carries no signal.
The idle gates cure this: once $w \gtrsim \sigma \sim \sqrt{N'}$ the clock fits \emph{inside} $B$ at $t^*$, the echo is fully suppressed, $\los_k(t^*) \approx a$, and $\krate(t^*) \approx -\tfrac1k\log a = \Omega(1)$.
Applying \cref{thm:circuit_hamiltonian} with $N'$ gates and the unary clock encoding gives
\begin{equation}\label{apd:eq:H_simple}
    H' = \sum_{j=1}^{N'} \sqrt{j(N'\!-\!j\!+\!1)}\,
    \bigl(\op{j}{j\!-\!1}_\cc \otimes V'_j + \op{j\!-\!1}{j}_\cc \otimes V'^{\dagger}_j\bigr),
\end{equation}
with history state \cref{eq:circuit_ham_evolution} and clock index $X_t \sim \mathrm{Bin}(N', \sin^2 t)$.
The \local-\dqpt{} instance is $H'$ on $n' = n + N' = \poly(n)$ qubits, the projector $\projk = \prod_{i=1}^k \op{0}_i$, the initial state $\ket{\psi_0} = \ket{0^n}\otimes\ket{0}_\cc$, the evolution time $t^* = \pi/4$, and the window $\deltat$.
The workspace is not returned to $\ket{0}$ at the clock center---during the idle block the majority register holds the (live) majority bit---but this is immaterial: the echo $\los_k(t)$ reads only the fresh answer positions through $\projk$, and the reset fan-out of $W^\dagger$ restores them to $\ket{0^k}$.
\Cref{apd:fig:simple_schematic} shows the circuit and the clock landscape it induces.

With the instance fixed, we read its Loschmidt echo off the clock.
By \cref{thm:circuit_hamiltonian},
the history state for the palindrome $V' = \qty{V'_1,\cdots,V'_{N'}}$ at continuous time $t$ is
\begin{equation}
    \ket{\psi(t)} = \sum_{j=0}^{N'} c_j(t)\,
    \ket{j}_\cc \otimes 
    V'_j \cdots V'_1 \ket{0^n}_\oo,
\end{equation}
where the clock amplitudes satisfy (\cref{eq:clock_amplitudes}, 
with $N = N'$)
$|c_j(t)|^2 = \binom{N'}{j} \sin^{2j}(t)\,\cos^{2(N'-j)}(t)$.
At $t = \pi/2$ the clock reaches $\ket{N'}_\cc$ and $V'_{N'}\cdots V'_1\ket{0^n} = \ket{0^n}$ 
(the palindrome acts as the identity, \cref{apd:eq:V_simple}); 
at $t = 0$ the clock is at $\ket{0}_\cc$ with state $\ket{0^n}$.
The local Loschmidt echo at time $t$ is
\begin{equation}\label{eq:Lk_clock}
    \los_k(t) = \sum_{j=0}^{N'} |c_j(t)|^2\;
    \bra{0^n}_\oo V'^\dagger_1 \cdots V'^\dagger_j\, \projk\,
    V'_j \cdots V'_1 \ket{0^n}_\oo
    = \sum_{j=0}^{N'} |c_j(t)|^2\; \langle \projk \rangle_j,
\end{equation}
where the $k$-local projector expectation at clock step $j$ is
\begin{equation}\label{eq:Pk_j}
    \langle \projk \rangle_j \coloneqq \bra{0^n}_\oo 
    V'^\dagger_1 \cdots V'^\dagger_j\, \projk\, V'_j \cdots V'_1 \ket{0^n}_\oo.
\end{equation}
By the palindromic structure $V'_{N'+1-i} = V'^\dagger_i$ of \cref{apd:eq:V_simple}, the second half of $V'$ inverts the first, so $V'_{N'-j}\cdots V'_1\ket{0^n} = V'_j\cdots V'_1\ket{0^n}$ and hence $\langle\projk\rangle_{N'-j} = \langle\projk\rangle_j$.
\textit{Endpoints.}
$\langle \projk \rangle_0 = \langle \projk \rangle_{\color{blue}N'} = 1$
(state is $\ket{0^n}$), so $\los_k(0) = \los_k(\pi/2) = 1$ and
$\krate(0) = \krate(\pi/2) = 0$.

\begin{figure}[t!]
    \centering
    \resizebox{0.95\linewidth}{!}{%
    \begin{quantikz}[row sep=0.46cm, column sep=0.22cm]
        \lstick{{\scriptsize$\ket{0}$}\\[-3pt]{\tiny pos\,$1$}}
            & \qw \gategroup[6, steps=3, style={dashed, rounded corners, fill=blue!4,
                inner xsep=3pt}, background, label style={label position=above, yshift=0.1cm}]{$W$}
            & \qw & \targ{} & \qw & \qw & \qw & \qw & \qw & \qw
            & \targ{} \gategroup[6, steps=3, style={dashed, rounded corners, fill=red!4,
                inner xsep=3pt}, background, label style={label position=above, yshift=0.1cm}]{$W^\dagger$}
            & \qw & \qw & \rstick{{\scriptsize$\ket{0}$}}\qw \\[-0.28cm]
        \wave &&&&&&&&&&&&& \\[-0.28cm]
        \lstick{{\scriptsize$\ket{0}$}\\[-3pt]{\tiny pos\,$k$}}
            & \qw & \qw & \targ{} & \qw & \qw & \qw & \qw & \qw & \qw & \targ{} & \qw & \qw & \rstick{{\scriptsize$\ket{0}$}}\qw \\
        \lstick{{\scriptsize$\ket{0}^{\otimes m n_0}$}\\[-3pt]{\tiny reg.}}
            & \gate{C^{\otimes m}} & \gate[2]{V_{\maj}} & \qw & \qw & \qw & \qw & \qw & \qw & \qw
            & \qw & \gate[2]{V_{\maj}^\dagger} & \gate{(C^\dagger)^{\otimes m}} & \qw \\
        \lstick{{\scriptsize$\ket{0}^{\otimes \bigO(\log m)}$}\\[-3pt]{\tiny maj.\,anc.}}
            & \qw & & \ctrl{-4} & \qw & \qw & \qw & \qw & \qw & \qw & \ctrl{-4} & & \qw & \qw \\
        \lstick{{\scriptsize$\ket{0\cdots 0}$}\\[-3pt]{\tiny clock ($N'$)}}
            & \qwbundle{N'} & \qw
            & \qw \slice[style={thick, dashed, red}]{}
            & \qw & \qw
            & \qw \slice[style={thick, dotted}]{}
            & \qw & \qw & \qw
            & \qw \slice[style={thick, dashed, green!55!black}]{}
            & \qw & \qw & \qw
    \end{quantikz}%
    }
    \\[2pt]{\small\textsf{(a)} The palindromic circuit with idle hold in the middle $V' = (W,\,\mathrm{idle}_w,\,W^\dagger)$} \\
    [10pt]
    \includegraphics[width=0.95\linewidth]{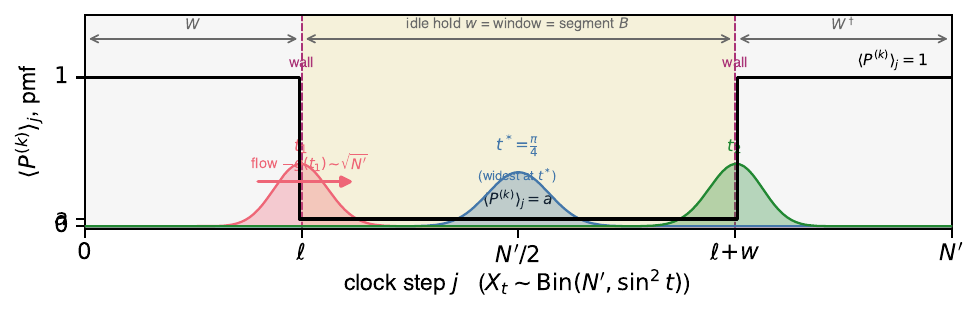}
    \\[1pt]{\small\textsf{(b)} The clock amplitude landscape and two-value of $\langle\projk\rangle_j$}
    \caption{The compute--hold--uncompute palindromic circuit construction (\cref{apd:sec:simple_construction}).
    \textbf{(a)}~The circuit $V' = (W, \mathrm{idle}_w, W^\dagger)$: $W$ (blue) runs the $m$ copies $C^{\otimes m}$, the majority $V_{\maj}$, and fans the majority bit onto the $k$ fresh answer positions; 
    $w$ idle gates \emph{hold} it; $W^\dagger$ (red) resets the positions by a second fan-out and uncomputes.
    The two fan-outs sit at clock steps $\ell$ and $\ell+w$ (the probes $t_1, t_2$, dashed), the idle center at $N'/2$ (the critical time $t^*$, dotted).
    \textbf{(b)}~\emph{The clock landscape} on the clock index $j\in[N']$. 
    The fresh answer positions give the exact two-value profile $\langle\projk\rangle_j = 1$ off $B$, $= a$ on the answer segment $B = \{\ell,\ldots,\ell+w\}$ (black step; $B$ shaded yellow, bounded by the two \textcolor{purple}{purple} dashed \emph{walls} at $\ell$ and $\ell+w$) (\cref{apd:lem:idle_dial}).
    The binomial clock $|c_j(t)|^2 \sim \mathrm{Bin}(N', \sin^2 t)$ is shown at $t^*$ (\textcolor{blue}{blue} curve: mean $N'/2$, deep inside $B$, widest since $p=\tfrac12$: \emph{Center}, $g(t^*)\le\eps$) and at the two probes $t^*\mp\deltat$ (\textcolor{red}{red} at $t_1$, \textcolor{green!55!black}{green} at $t_2$: means $\ell$, $\ell+w$, on the walls);
    there the clock straddles a wall, so the current $-\dot g(t_1)\sim\sqrt{N'}$ (\textcolor{red}{red} arrow) pours into $B$ (\emph{Probe flow}).
    }
    \label{apd:fig:simple_schematic}
\end{figure}

\subsubsection{Window dial and echo formula}\label{apd:sec:simple_echo}

We now turn to the window \emph{dial}:
how the idle count $w$ sets the window $\deltat$ and where the answer segment $B$ sits.
Geometrically, the dial pins the probe mean $N'\sin^2(t^*-\deltat)$ onto the wall $\ell$ of $B$, so the probe clock \emph{straddles} the wall (half its mass inside $B$, half outside).
The boundary PMF then sits at the binomial mode, where the population crossing into $B$ is largest, giving the order-maximal $\Theta(\sqrt{N'})$ current that later drives the susceptibility jump.
\begin{lemma}[Window dial via idle hold]\label{apd:lem:idle_dial}
    For every idle count $0 \le w \le N'$, the answer segment of $V'$ is $B = \{\ell, \ldots, \ell+w\}$ (up to the $\bigO(k)$ fan-out width on each side),
    centered at the clock midpoint $N'/2$, and the probe times $t^* \mp \deltat$ land within $\bigO(k)$ of the two walls of $B$ when
    \begin{equation}\label{apd:eq:idle_dial}
        \sin(2\deltat) = \frac{w}{N'},
        \qquad\text{equivalently}\qquad
        N' \sin^2\!\bigl(t^* - \deltat\bigr) = \ell .
    \end{equation}
    The canonical window $\deltat = \pi/12$ is the half-idle case $w = N'/2$ (so $w = 2\ell$); $w \to 0$ gives $\deltat \to 0$ and $w \to N'$ gives $\deltat \to \pi/4$.
\end{lemma}
\begin{proof}
    The fresh positions are written by the fan-out of $W$ at clock step $\ell$ and reset by the fan-out of $W^\dagger$ at step $\ell + w$, and are untouched in between (idle) and outside (the compute/uncompute blocks $A, A^\dagger$ act only on registers and majority ancillae), so $B = \{\ell, \ldots, \ell+w\}$ with midpoint $\ell + w/2 = N'/2$.
    Each fan-out is $k$ \textsc{cnot}s, one per answer position, not a single gate, so each wall is sharp only up to an $\bigO(k)$ smear in clock step; with $k$ constant this is $\bigO(1)$, absorbed throughout.
    Under \cref{thm:circuit_hamiltonian} the unary clock sits at step $X_t \sim \bpf(N', \sin^2 t)$ (\cref{apd:eq:H_simple}), with mean $N'\sin^2 t$, so the probe $t_1 := t^* - \deltat$ lands on the left wall $\ell$ exactly when $N'\sin^2 t_1 = \ell$.
    Since $\sin^2 t_1 = \sin^2(\pi/4 - \deltat) = \bigl(1 - \sin 2\deltat\bigr)/2$, this reads $\bigl(1 - \sin 2\deltat\bigr)/2 = \ell/N'$, i.e.\ 
    \begin{equation}
        \sin 2\deltat = 1 - 2\ell/N' = w/N'.
    \end{equation}
The dial realizes every $w \in [0, N']$, hence every $\deltat \in [0, \pi/4]$; only a sub-range is admissible for the reduction, set later by the proof's estimates (\cref{apd:rmk:floor_tradeoff}).
\end{proof}

The next lemma reduces every quantity in the \nameref{apd:prm:promise_local_dqpt} to the leakage probability outside the answer segment $B = \{\ell, \ldots, \ell+w\}$.

\begin{lemma}[Echo formula]\label{apd:lem:exact_echo_s}
    With $X_t \sim \bpf(N', \sin^2 t)$ and the answer segment $B = \{\ell, \ldots, \ell+w\}$, 
    the step expectations of $V'$ are exactly two-valued, 
    $\langle\projk\rangle_j = 1$ for $j \notin B$ and $\langle\projk\rangle_j = a$ for $j \in B$, 
    so the Loschmidt echo is a mixture of $a$ and the leakage $g(t)$,
    \begin{equation}\label{apd:eq:exact_echo_s}
        \los_k(t) = a + (1-a)\,g(t),
    \end{equation}
    where the leakage $g(t)$ is the probability of the clock being outside $B$,
    \begin{equation}\label{apd:eq:leakage}
        \qquad g(t) := \Pr[X_t \notin B] = \Pr[X_t < \ell] + \Pr[X_t > \ell + w] .
    \end{equation}
    The palindromic circuit $V'$ gives the mirror symmetry $\los_k(\pi/2 - t) = \los_k(t)$, hence $\dkrate(\tfrac\pi4 + s) = -\dkrate(\tfrac\pi4 - s)$ and, with $t_1 := t^* - \deltat$, the integrated susceptibility collapses to a single probe slope,
    \begin{equation}\label{apd:eq:jump_probe_s}
        J_k(t^*) = 2\,\dkrate(t_1),
        \qquad
        \dkrate(t) = \frac{(1-a)\,\bigl(-\dot g(t)\bigr)}{k\, \los_k(t)} .
    \end{equation}
\end{lemma}
\begin{proof}
    Two-valuedness and \cref{apd:eq:exact_echo_s} are immediate from the placement in \cref{apd:lem:idle_dial}: by the history state \cref{eq:circuit_ham_evolution} the echo is the clock-weighted average $\los_k(t) = \sum_j |c_j(t)|^2 \langle\projk\rangle_j$ of the step expectations, and off $B$ the answer positions are exactly $\ket{0^k}$ (untouched) while on $B$ they carry the fanned-out majority bit, which equals $0$ with probability $a$.
    The word $V' = (W, I^{w}, W^\dagger)$ satisfies $V'_{N'+1-j} = V'^{\dagger}_j$, 
    so $\langle\projk\rangle_{N'-j} = \langle\projk\rangle_j$ and $B$ is symmetric under $j \mapsto N'-j$; 
    the binomial reflection $X_t \leftrightarrow N' - X_{\pi/2-t}$ then gives $g(\pi/2-t) = g(t)$ and the stated symmetry, and differentiating $\krate = -\frac1k\log\los_k$ yields the slope formula and \cref{apd:eq:jump_probe_s}.
\end{proof}

The remaining analysis uses four standard estimates on the binomial distribution, collected here for reference.

\begin{fact}[Binomial toolbox]\label{apd:fct:binomial}
Write $\bpmf(j; N, p) := \binom{N}{j}p^j(1-p)^{N-j}$ for the probability mass function (PMF) of the binomial distribution $\bpf(N, p)$,
write $\sigma^2 := Np(1-p)$, 
and let $X \sim \bpf(N, p)$.
\begin{enumerate}
    \item[(i)] \emph{Tail derivative:} for $p = \sin^2 t$,
    \begin{equation}\label{eq:cdf_deriv}
        \frac{d}{dt}\Pr[X \ge j_0]
        = \sin(2t)\cdot N\, \bpmf(j_0 - 1;\, N-1,\, p) .
    \end{equation}
    \item[(ii)] \emph{Pointwise size:} if $\sigma^2 \ge 1$ then $\max_j \bpmf(j; N, p) \le C_0/\sigma$, and $\bpmf(j; N, p) \ge c_0/\sigma$ for every $j$ with $|j - Np| \le C$, with $c_0$ depending only on $C$.
    \item[(iii)] \emph{Tail bound (Hoeffding):} $\Pr[|X - Np| \ge s] \le 2e^{-2s^2/N}$.
    \item[(iv)] \emph{Median:} $|\Pr[X < m] - \tfrac12| \le C_1/\sigma$ whenever $|m - Np| = \bigO(1)$ and $\sigma^2 \ge 1$.
\end{enumerate}
\end{fact}
\begin{proof}

	\emph{(i) Tail derivative.}
	Differentiating one PMF term in $p$ and regrouping with the binomial identities $j\binom{N}{j} = N\binom{N-1}{j-1}$ and $(N-j)\binom{N}{j} = N\binom{N-1}{j}$ turns the derivative into a first difference of one-step-shorter PMFs,
	\begin{equation}
		\frac{d}{dp}\bpmf(j; N, p) = N\bigl[\bpmf(j-1; N-1, p) - \bpmf(j; N-1, p)\bigr] .
	\end{equation}
	Summing this over the upper tail $j \ge j_0$ telescopes, 
    leaving only the boundary term, 
    and the chain rule with $dp/dt = \sin(2t)$ for $p = \sin^2 t$ converts the $p$-derivative into the time derivative:
	\begin{equation}
		\begin{aligned}
			\frac{d}{dt}\Pr[X \ge j_0]
			&= \sin(2t)\sum_{j \ge j_0} N\bigl[\bpmf(j-1; N-1, p) - \bpmf(j; N-1, p)\bigr] \\
			&= \sin(2t)\, N\, \bpmf(j_0 - 1; N-1, p) ,
		\end{aligned}
	\end{equation}
    which is \cref{eq:cdf_deriv}: 
    the tail CDF moves in time at a rate set by the boundary PMF $\bpmf(j_0-1; N-1, p)$, 
    the population density crossing the cut $j_0$.

	\emph{(ii) Pointwise size.}
	By the Stirling's formula applied to $\binom{N}{j}$, 
    whenever $\sigma^2 = Np(1-p) \ge 1$ the PMF is controlled, uniformly in $j$, by a Gaussian envelope of width $\sigma$,
	\begin{equation}
		\bpmf(j; N, p) = \frac{1}{\sqrt{2\pi}\,\sigma}\,\exp\!\biggl(-\frac{(j - Np)^2}{2\sigma^2}\biggr)\,\bigl(1 + o(1)\bigr) .
	\end{equation}
    The exponential never exceeds $1$, so the envelope peaks at $1/(\sqrt{2\pi}\,\sigma)$ and $\max_j \bpmf(j; N, p) \le C_0/\sigma$ for an absolute constant $C_0$.
	On the bulk $|j - Np| \le C$ the exponent stays above $-C^2/(2\sigma^2) \ge -C^2/2$ (using $\sigma^2 \ge 1$), so the envelope is bounded below by $\bpmf(j; N, p) \ge c_0/\sigma$ with $c_0 = \Theta(e^{-C^2/2})$ depending only on $C$.
	Thus every bulk site carries an $\Theta(1/\sigma)$ share of the mass, and none carries more.

	\emph{(iii) Tail bound (Hoeffding).}
	Write $X = \sum_{i=1}^{N} X_i$ as a sum of independent $X_i \sim \mathrm{Bernoulli}(p)$,
    each supported on $[0,1]$ with $\mathbb{E}X = Np$, 
    the binomial's defining representation.
	Hoeffding's inequality for bounded independent summands 
    (unit ranges, so $\sum_i (b_i-a_i)^2 = N$) gives the one-sided tail $\Pr[X - Np \ge s] \le e^{-2s^2/N}$;
    the same bound applies to the lower tail, and a union bound over the two doubles the constant:
	\begin{equation}
		\Pr\bigl[\,|X - Np| \ge s\,\bigr]
		\le \Pr[X - Np \ge s] + \Pr[Np - X \ge s]
		\le 2e^{-2s^2/N} .
	\end{equation}
	This pins the clock to within $\bigO(\sqrt{N\log(1/\delta)})$ of its mean except with probability $\delta$, 
    the concentration used to push the far wall's tail below the probe split.

	\emph{(iv) Median.}
	The median $\nu$ of $\bpf(N, p)$ satisfies $|\nu - Np| \le 1$, so at the median $\Pr[X < \nu] \le \tfrac12 \le \Pr[X \le \nu] = \Pr[X < \nu] + \bpmf(\nu; N, p)$, whence $|\Pr[X < \nu] - \tfrac12| \le \bpmf(\nu; N, p) \le C_0/\sigma$ by (ii).
	Any cut $m$ with $|m - Np| = \bigO(1)$ lies $|m - \nu| = \bigO(1)$ sites from the median, so sliding the threshold from $\nu$ to $m$ adds the mass of $\bigO(1)$ integer sites, each at most $C_0/\sigma$:
	\begin{equation}
		\Bigl|\Pr[X < m] - \tfrac12\Bigr|
		\le \underbrace{\bigl|\Pr[X < m] - \Pr[X < \nu]\bigr|}_{\le\, \bigO(1)\cdot C_0/\sigma}
		+ \underbrace{\bigl|\Pr[X < \nu] - \tfrac12\bigr|}_{\le\, C_0/\sigma}
		\le \frac{C_1}{\sigma} .
	\end{equation}
    So a binomial is balanced to within $\bigO(1/\sigma)$ about any near-mean threshold, the estimate behind the half-and-half probe split.
\end{proof}

With the echo formula and the toolbox in place, we can now prove the theorem.

\subsubsection{BQP-hardness proof}\label{apd:sec:simple_proof}

\begin{figure}[t!]
    \centering
    \includegraphics[width=0.72\linewidth]{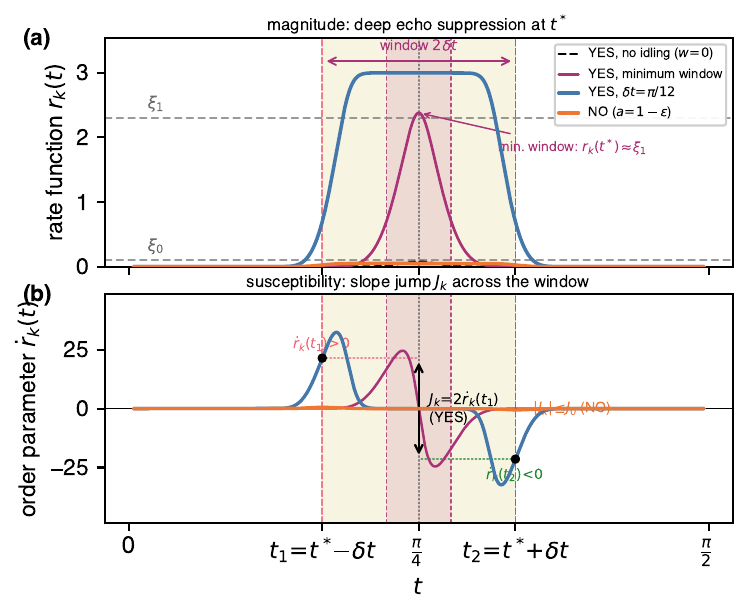}
    \caption{The two promise conditions on the rate function, for the \yes{} case (rejecting \koverlap, $a = \eps$) at the (canonical) window $\deltat = \pi/12$ (blue) and at the minimum window 
    (purple), 
    against the \no{} case (accepting, $a = 1-\eps$, orange).
    As an illustrative example, let $\ell = 50$, $k = 1$, $\eps = 0.05$ ($w = 100$, $N' = 200$ canonical; $w = 21$, $N' = 121$ minimum window).
    Both panels share the time axis. 
    Throughout, the \emph{probes} $t_{1,2} = t^*\mp\deltat$ are the evaluation times and the \emph{walls} $\ell, \ell+w$ are the clock steps bounding $B$ that the probes straddle; each case's window $[t_1, t_2]$ is shaded between its probes --- the canonical window (yellow, red/green probes) and the narrower minimum window nested inside it (purple).
    \textbf{(a)}~\emph{Magnitude.} 
    The rate function $\krate(t)$ develops a deep suppressed-echo plateau at $t^* = \pi/4$ in the \yes{} case, $\krate(t^*) \ge \xi_1$ (\cref{apd:eq:xi_yes_s}), 
    whereas the \no{} curve (orange) stays uniformly below $\xi_0$ (\cref{apd:eq:xi_no_s}), hugging the $t$-axis at $\krate \approx 0$ (since $\los_k \approx 1$), so that it is easily missed.
    Padding deepens and widens the plateau: 
    at the canonical window it sits well above $\xi_1$, 
    while at the minimum window it is a narrow peak that just yields $\xi_1$ --- 
    this is precisely the window where $\krate(t^*) \to \xi_1$, 
    since $\xi_1 = \frac1k\log\frac1{2\eps}$ is window-independent while $\los_k(t^*) \to 2\eps$ there (\cref{apd:rmk:floor_tradeoff}).
    The black dashed curve is the \emph{same} rejecting instance with \emph{no idling} ($w = 0$): the answer segment shrinks to a single clock step, so $\krate$ collapses to $\bigO(1/\sqrt{N'})$ and hugs the axis like the accepting case, confirming that the idle hold is what lifts the peak to $\Omega(1)$ (\cref{apd:sec:simple_echo}).
    \textbf{(b)}~\emph{Susceptibility.}
    The order parameter $\dkrate(t)$ jumps across the window in the \yes{} case, $\dkrate(t_1) > 0 > \dkrate(t_2)$, 
    giving the integrated susceptibility $J_k = \dkrate(t_1) - \dkrate(t_2) = 2\dkrate(t_1) \ge J_1$ (\cref{apd:eq:J_yes_s}), 
    while the \no{} case has $|J_k| \le J_0$ (\cref{apd:eq:J_no_s}); 
    the gap is $J_1/J_0 = \Theta(1/\eps)$ (\cref{apd:eq:gap_s}).
    The minimum-window jump (purple) is narrower and shorter, its probes closing in on $t^*$.
    }
    \label{apd:fig:simple_rate}
\end{figure}

\begin{theorem}
    [\BQP-\hardness{} of \local-\dqpt{}]
    \label{apd:thm:bqp_hard}
    \nameref{apd:prm:promise_local_dqpt} is $\BQP$-\hard{} for any constant $k \ge 1$.
    Specifically, for any \koverlap{} instance $W$ on $n$ qubits with $\ell = \poly(n)$ gates and overlap error $\eps = e^{-\Omega(k)}$, and any window 
    $\deltat \in \bigl[\Omega(\sqrt{(k+\log\ell)/\ell}),\, \pi/4 - \Omega(1)\bigr]$,
    the instance 
    (Hamiltonian $H'$, initial state $\ket{0^{n'}}$, time point $t^* = \pi/4$, offset $\deltat$) 
    with $w = \bigl\lfloor 2\ell\sin(2\deltat)/(1-\sin(2\deltat))\bigr\rfloor$ idle gates 
    ($N' = 2\ell + w$)
    satisfies the \yes{}/\no{} conditions of \mbox{\nameref{apd:prm:promise_local_dqpt}}: 
    \begin{itemize}
        \item 
    rejecting \koverlap{} ($a \le \eps$) gives $\krate(t^*) \ge \xi_1$ and $J_k(t^*) \ge J_1 = \Theta(\sqrt{N'}/k)$, 
        \item 
    accepting \koverlap{} ($a \ge 1-\eps$) gives $\krate(t) \le \xi_0$ uniformly and $|J_k(t^*)| \le J_0 = \bigO(\eps\sqrt{N'}/k)$, 
    \end{itemize}
    with gap $J_1/J_0 = \Theta(1/\eps) = e^{\Omega(k)}$, magnitude gap $\xi_1 - \xi_0 = \Omega(1)$ with $\xi_1 = \Omega(1)$, 
    and probe conditioning $\krate(t^*\pm\deltat) \le 1$.
\end{theorem}

\begin{proof}
By the exact echo formula \cref{apd:eq:exact_echo_s},
every quantity is a function of the \emph{leakage} $g(t) := \Pr[X_t \notin B]$---the weight the clock $X_t \sim \mathrm{Bin}(N', \sin^2 t)$ places outside the answer segment $B = \{\ell, \ldots, \ell+w\}$,
which is centered at $N'/2$ with its two walls at $\ell$ and $\ell+w$ (distance $w/2 + \bigO(k)$ from the center, $w + \bigO(k)$ apart, the $\bigO(k)$ fan-out width absorbed throughout).
The minimum-window bound $\deltat \ge \Omega(\sqrt{(k+\log\ell)/\ell})$ gives $w^2/N' \ge \Omega(k + \log\ell)$ 
(via the dial $\sin(2\deltat) = w/N'$), 
invoked at each step below.
The two terms have two origins (both derived in \cref{apd:rmk:floor_tradeoff}):
the $\Omega(k)$ enforces the center suppression $g(t^*)\le\eps=e^{-\Omega(k)}$, i.e.\ Hoeffding exponent $w^2/(2N')\ge\Omega(k)$;
the $\Omega(\log\ell)$ pushes the far-wall tail $2e^{-2w^2/N'}$ below the $\eps^2/N'$ that the $\Theta(1/\sqrt\ell)$ probe split needs.

We establish each promise condition together with the one binomial estimate it needs, handling at once the two cases of the underlying \koverlap{} instance, \emph{rejecting} ($a \le \eps$) and \emph{accepting} ($a \ge 1-\eps$).
\cref{apd:fig:simple_rate} contrasts the resulting \yes{} and \no{} rate functions and their derivatives.

\textbf{Magnitude of rate function.}
At the center $t^* = \pi/4$ the clock mean is $N'/2$, the center of $B$, and both walls are at least $w/2$ away; so Hoeffding's inequality (\cref{apd:fct:binomial}(iii)) bounds the leakage,
\begin{equation}\label{apd:eq:center_s}
    g(t^*) \le 2e^{-2(w/2)^2/N'} = 2e^{-w^2/(2N')} \le \eps ,
\end{equation}
the last step because the minimum-window bound makes the exponent $w^2/(2N') \ge \Omega(k) \ge \log(2/\eps)$, with $\eps = e^{-\Omega(k)}$.
Set the magnitude thresholds $\xi_1 := \frac1k\log\frac1{2\eps}$ (the \yes{} floor) and $\xi_0 := \frac{2\eps}k$ (the \no{} ceiling).
\begin{itemize}
    \item 
    In the rejecting case ($a \le \eps$), the echo formula gives 
    \begin{equation}
        \los_k(t^*) = a + (1-a)\,g(t^*) \le a + g(t^*) \le 2\eps
    \end{equation}
    by \cref{apd:eq:center_s}, 
    so the rate function clears the \yes{} floor:
    \begin{equation}\label{apd:eq:xi_yes_s}
        \krate(t^*) = -\tfrac1k\log\los_k(t^*) \ge \tfrac1k\log\tfrac1{2\eps} = \xi_1 = \Omega(1) .
    \end{equation}
    \item 
    In the accepting case ($a \ge 1-\eps$), the echo formula gives 
    \begin{equation}
        \los_k(t) = a + (1-a)\,g(t) \ge a \ge 1-\eps 
    \end{equation}
    for \emph{all} $t$ (as $g \ge 0$), so with $\log\frac1{1-x} \le 2x$ on $[0,\tfrac12]$ the rate function stays under the \no{} ceiling,
    \begin{equation}\label{apd:eq:xi_no_s}
        \krate(t) \le \tfrac{2\eps}k = \xi_0 \qquad \text{uniformly in } t ,
    \end{equation}
    and the magnitude gap is $\xi_1 - \xi_0 = \Omega(1)$.
\end{itemize}

\textbf{Probe conditioning.}
The membership algorithm for \mbox{\nameref{apd:prm:promise_local_dqpt}} estimates the probe slopes $\dkrate(t^*\pm\deltat)$ by the commutator method, 
whose sample cost grows as $1/\los_k(t^*\pm\deltat)^2$; 
the promise $\krate(t^*\pm\deltat) \le 1$, 
keeps these probe echoes $\Omega(1)$ and hence the cost polynomial.

At the probe $t_1 := t^* - \deltat$ the clock mean is $\ell$, within $\bigO(k)$ of the left wall, with spread $\sigma_1 := \sqrt{N' p_1(1-p_1)} = \Theta(\sqrt\ell)$ where $p_1 := \sin^2 t_1 = \ell/N'$ (\cref{apd:eq:idle_dial}).
The median estimate (\cref{apd:fct:binomial}(iv)) gives $\Pr[X_{t_1} < \ell] = \tfrac12 + \bigO(1/\sqrt\ell)$ (the probe clock straddling the left wall in \cref{apd:fig:simple_schematic}(b)),
while the far wall, $w + \bigO(k)$ away, contributes only $\Pr[X_{t_1} > \ell+w] \le 2e^{-2w^2/N'} \le \eps^2/N'$ (the minimum-window bound again).
So the leakage and echo at the probe are
\begin{equation}\label{apd:eq:probe_g_s}
    g(t_1) = \tfrac12 + \bigO(1/\sqrt\ell),
    \qquad
    \los_k(t_1) = \tfrac{1+a}2 + \bigO(1/\sqrt\ell) = \Theta(1) ,
\end{equation}
where $\krate(t_1) \le 1$. 
By the mirror symmetry of \cref{apd:lem:exact_echo_s}, 
the same holds at $t^*+\deltat$, meeting the promise.

\textbf{Jump of time derivative.}
The slope formula \cref{apd:eq:jump_probe_s} turns the population current $-\dot g$ into the order parameter, 
\begin{equation}
    \dkrate = (1-a)(-\dot g)/(k\,\los_k) .
\end{equation}
Writing $g = 1 - \Pr[X \ge \ell] + \Pr[X \ge \ell+w+1]$, we differentiate with the tail-derivative identity (\cref{apd:fct:binomial}(i)),
\begin{equation}\label{apd:eq:tail_deriv_s}
    \frac{d}{dt}\Pr[X_t \ge j] = \sin(2t)\, N'\, \bpmf(j-1;\, N'\!-\!1,\, \sin^2 t) ;
\end{equation}
applied to the two tails $j = \ell$ and $j = \ell+w+1$ it gives
\begin{equation}\label{apd:eq:gdot_s}
    \dot g(t) = \sin(2t)\, N'\bigl[\bpmf(\ell+w;\, N'\!-\!1,\, p) - \bpmf(\ell-1;\, N'\!-\!1,\, p)\bigr], \qquad p:= \sin^2 t .
\end{equation}
The spread is $\sigma := \sqrt{N'p(1-p)}$, and the prefactor cancels it, since $\sin(2t) = 2\sin t\cos t = 2\sqrt{p(1-p)}$,
\begin{equation}\label{apd:eq:cancel_s}
    \frac{\sin(2t)\,N'}{\sigma} = \frac{2\sqrt{p(1-p)}\;N'}{\sqrt{N'p(1-p)}} = 2\sqrt{N'} .
\end{equation}
Hence the current is bounded uniformly, $|\dot g(t)| \le C\sqrt{N'}$ at \emph{every} $t$, with $C := 2\max\{C_0, 1\}$,
where $C_0 \ge c_0 > 0$ are the peak and mode PMF bounds of \cref{apd:fct:binomial}(ii) 
(the probe constant $c_1 := 2c_0$ enters below).
For $\sigma \ge 1$,
\begin{align*}
    |\dot g(t)|
    &= \sin(2t)\, N'\,\bigl|\bpmf(\ell+w) - \bpmf(\ell-1)\bigr|
        \tag{\cref{apd:eq:gdot_s}} \\
    &\le \sin(2t)\, N'\,\frac{C_0}{\sigma}
        \tag{\cref{apd:fct:binomial}(ii)} \\
    &= 2C_0\sqrt{N'} \le C\sqrt{N'} ,
        \tag{\cref{apd:eq:cancel_s}}
\end{align*}
where the middle step is the only subtlety: 
each PMF lies in $[0, C_0/\sigma]$ 
(\cref{apd:fct:binomial}(ii)), 
so their difference lies in $[-C_0/\sigma,\, C_0/\sigma]$, 
hence $\abs{\bpmf(\ell+w) - \bpmf(\ell-1)}\le C_0/\sigma$.
Where $\sigma < 1$ (i.e.\ $t$ near $0$ or $\pi/2$), \cref{apd:fct:binomial}(ii) does not apply; 
instead each $\bpmf \le 1$, and the cancellation gives $\sin(2t)\,N' = 2\sigma\sqrt{N'} < 2\sqrt{N'}$, 
so $|\dot g(t)| \le 2\sqrt{N'} \le C\sqrt{N'}$ directly.

\begin{itemize}
    \item
    In the accepting case of the $k$-Overlap instance,
    $1 - a \le \eps$, the uniform bound $|\dot g| \le C\sqrt{N'}$ above, 
    and $\los_k(t) = a + (1-a)g(t) \ge a \ge \tfrac12$ (\cref{apd:eq:exact_echo_s}) give the uniform per-point bound 
    \begin{equation}
        |\dkrate(t)| = (1-a)|\dot g(t)|/(k\los_k(t)) \le \eps\, C\sqrt{N'}/(k/2) = 2C\eps\sqrt{N'}/k,
    \end{equation}
    so the mirror identity \cref{apd:eq:jump_probe_s} keeps the integrated susceptibility under the \no{} ceiling $J_0 := 4C\eps\sqrt{N'}/k$,
    \begin{align}
        |J_k(t^*)|
        &= 2\,|\dkrate(t_1)|
            \tag{\cref{apd:eq:jump_probe_s}} \\
        &\le 2\cdot\frac{2C\eps\sqrt{N'}}{k}
            = \frac{4C\eps\sqrt{N'}}{k} =: J_0 , \label{apd:eq:J_no_s}
    \end{align}
    indicating that the rate function is essentially flat across the window.

    \item
    In the rejecting case of the $k$-Overlap instance, evaluate \cref{apd:eq:gdot_s} at $t_1$, where the dial centers the probe clock on the lower wall (mean $N'p_1 = \ell$).
    The lower-wall term $\bpmf(\ell-1; N'\!-\!1, p_1)$ then sits at the binomial \emph{mode}, $\ge c_0/\sigma_1$ (a constant fraction of the peak height $\sim 1/\sigma_1$),
    while the upper wall $\ell+w$ is $w$ away, deep in the tail, so $\bpmf(\ell+w; N'\!-\!1, p_1) \le \eps^2/N'$ by the minimum-window bound (the far tail of \cref{apd:eq:probe_g_s}).
    The population current is then positive and $\Theta(\sqrt{N'})$:
    \begin{align*}
        -\dot g(t_1)
        &= \sin(2t_1)\,N'\bigl[\bpmf(\ell-1; N'\!-\!1, p_1) - \bpmf(\ell+w; N'\!-\!1, p_1)\bigr]
            \tag{\cref{apd:eq:gdot_s}} \\
        &\ge \sin(2t_1)\,N'\,\frac{c_0}{\sigma_1}\,(1-o(1))
            \tag{mode bound; far tail $\eps^2/N' = o(c_0/\sigma_1)$} \\
        &= 2c_0\sqrt{N'}\,(1-o(1)) =: c_1\sqrt{N'} .
            \tag{\cref{apd:eq:cancel_s}}
    \end{align*}
    Here $c_1 \le C$, since the mode value $c_0$ is at most the PMF bound $C_0$ of the uniform estimate.
    As $a \le \eps$ and $\los_k(t_1) \le \tfrac34$ (by \cref{apd:eq:probe_g_s} and $\eps \le \tfrac14$),
    \begin{equation}\label{apd:eq:rdot_probe_s}
        \dkrate(t_1) = (1-a)\frac{-\dot{g}(t_1)}{k\,\los_k(t_1)}
         \ge \frac{(1-\eps)\,c_1\sqrt{N'}}{(3/4)\,k} \ge \frac{c_1\sqrt{N'}}k > 0 .
    \end{equation}
    So the mirror identity \cref{apd:eq:jump_probe_s} clears the \yes{} floor $J_1 := 2c_1\sqrt{N'}/k$,
    \begin{equation}\label{apd:eq:J_yes_s}
        J_k(t^*) = 2\,\dkrate(t_1) \ge \frac{2c_1\sqrt{N'}}k =: J_1.
    \end{equation}
    Moreover $\dkrate(t^*-\deltat) > 0 > \dkrate(t^*+\deltat)$, 
    so the rate function genuinely peaks at $t^*$.
\end{itemize}

\textbf{Susceptibility gap.}
The common $\sqrt{N'}/k$ factor cancels in the ratio of \cref{apd:eq:J_yes_s,apd:eq:J_no_s},
\begin{equation}\label{apd:eq:gap_s}
    \frac{J_1}{J_0} = \frac{c_1}{2C\eps} = \Theta(1/\eps) = e^{\Omega(k)} \ge c > 1 ,
\end{equation}
so, with the magnitude gap \cref{apd:eq:xi_yes_s,apd:eq:xi_no_s}, both directions of the reduction hold. 
The \yes/\no{} labels between \mbox{\nameref{prm:k_overlap}} and \mbox{\nameref{apd:prm:promise_local_dqpt}} are swapped, valid since $\BQP = \mathsf{co}\BQP$.

    \textbf{Polynomial construction.}
    Write $s := \sin(2\deltat)$.
    Substituting the idle count $w = \bigl\lfloor 2\ell s/(1-s)\bigr\rfloor$ into $N' = 2\ell + w$ gives the closed form $N' = 2\ell/(1-s) - \theta$ with $\theta := \{2\ell s/(1-s)\} \in [0,1)$, the dial $w/N' = \sin(2\deltat)$ of \cref{apd:eq:idle_dial} inverted (exact up to the $\bigO(1/N')$ floor rounding, a $\bigO(1)$ clock-step shift absorbed in the $\bigO(k)$ fan-out smear).
    The upper window cutoff $\deltat \le \pi/4 - \Omega(1)$ keeps $1-s = \Omega(1)$, so
    \begin{equation}\label{apd:eq:Nprime_poly}
        2\ell \;\le\; N' \;=\; \frac{2\ell}{1-s} - \theta \;\le\; \frac{2\ell}{1-s} \;=\; \bigO(\ell) \;=\; \poly(n) .
    \end{equation}
    As $\deltat \to \pi/4$ the denominator $1-s \to 0$ and $N'$ would diverge superpolynomially in $1/(\pi/4-\deltat)$, so the cutoff is what keeps the construction polynomial (\cref{apd:rmk:floor_tradeoff} relaxes it to $\pi/4 - 1/\poly(n)$, still giving $N' = \poly(n)$).
    Each of the $N'$ terms of $H'$ is $\bigO(1)$-local (an $\bigO(1)$-local clock transition tensored with an at most $2$-local gate), $n' = n + N' = \poly(n)$, and the thresholds are computable in polynomial time from $(n, \ell, k, \deltat)$, so the reduction is polynomial.
\end{proof}

\begin{remark}[The admissible window range]\label{apd:rmk:floor_tradeoff}
Here the answer segment $B$ \emph{is} the window---its width $w$ and the window $\deltat$ are locked by \cref{apd:eq:idle_dial}---so a narrow window forces a narrow $B$, and the proof needs $w^2/N' = \Omega(k + \log\ell)$ (the $\Omega(k)$ from the center bound \cref{apd:eq:center_s}, the $\Omega(\log\ell)$ from suppressing the far-wall tail \cref{apd:eq:probe_g_s} below the $\Theta(1/\sqrt\ell)$ probe split), i.e.\ $w = \Omega(\sqrt{N'(k+\log\ell)})$ and hence a minimum window
\begin{equation}\label{apd:eq:floor_s}
    \deltat \;\gtrsim\; \sqrt{\tfrac{k+\log\ell}{\ell}} = 1/\poly(n) .
\end{equation}
At the upper end the exact dial pins $\sigma_1 = \sqrt{N'p_1(1-p_1)} = \sqrt{\ell(1-\ell/N')} = \Theta(\sqrt\ell)$ for every $\deltat < \pi/4$, so no clock estimate degrades as $\deltat \to \pi/4$ and the only obstruction is polynomiality $N' = \poly(n)$; the construction therefore realizes every window in $\bigl[\,\Omega(\sqrt{(k+\log\ell)/\ell}),\; \pi/4 - 1/\poly(n)\,\bigr]$, which contains the admissible range $\deltat \ge 1/\poly(n)$ of \mbox{\nameref{apd:prm:promise_local_dqpt}} as well as the canonical $\deltat = \pi/12$ (here $w = N'/2$).
\Cref{apd:thm:bqp_hard} states the slightly weaker upper bound $\pi/4 - \Omega(1)$, which already covers every window the problem admits.
\end{remark}

\begin{proposition}\label{apd:thm:in_bqp}
    \mbox{\nameref{apd:prm:promise_local_dqpt}} is contained in $\BQP$
    for any constant $k \ge 1$.
\end{proposition}
\begin{proof}
    The decision criterion involves two checks.
    \emph{Magnitude check:} estimate $\los_k(t^*)$ to constant
    multiplicative error by preparing $e^{-\ii Ht^*}\ket{\psi_0}$
    and measuring $\projk$; since $k = \bigO(1)$ and $\xi_1 = \Omega(1)$,
    the echo $\los_k = e^{-k\,\krate} = e^{-\Theta(k)} = \Theta(1)$,
    so $\bigO(1)$ repetitions suffice.
    \emph{Susceptibility check:} estimate the probe slopes $\dkrate(t^* \pm \deltat)$
    via the commutator method (\cref{apd:lem:comm_complexity}) and form the
    integrated susceptibility $J_k(t^*) = \dkrate(t^*-\deltat) - \dkrate(t^*+\deltat)$ (\cref{apd:eq:jump_def}).
    Under the \textsc{or} structure of the \no{} case it suffices to detect
    either $J_k(t^*) \le J_0$ or $\krate(t^*) \le \xi_0$;
    if either holds, the algorithm outputs \no.
    The well-conditioning promise $\krate(t^* \pm \deltat) \le R$ keeps the probe echoes $\los_k(t^* \pm \deltat) = \Omega(1)$, so each $\dkrate$ estimate has constant sample cost, and the promise gaps $J_1/J_0 \ge c > 1$ and $\xi_1 - \xi_0 \ge \Omega(1)$ then resolve the decision at constant relative precision.
    Each evaluation requires Hamiltonian simulation for time
    $\bigO(t^*) = \bigO(1)$, achievable in $\poly(n)$ gates via product formulas.
    See more details in \cref{apd:sec:algorithm}.
\end{proof}

\begin{corollary}\label{thm:bqp_complete_local_dqpt}
    \nameref{apd:prm:promise_local_dqpt} is
    \nameref{def:bqp_complete} for any constant $k \ge 1$.
\end{corollary}
\begin{proof}
    Membership in $\BQP$ follows from \cref{apd:thm:in_bqp} and
    \nameref{def:bqp}-hardness from \cref{apd:thm:bqp_hard}.
\end{proof}

Having established the complexity-theoretic landscape,
we now turn to the algorithmic side.

\section{Algorithm analysis for searching DQPTs}\label{apd:sec:algorithm}

This appendix provides complete proofs for the algorithmic results stated in the main text.
We consider the following search problem for finding the critical times of \nameref{apd:prm:promise_local_dqpt}.

\begin{problem}[\search-\dqpt]\label{apd:prm:search_dqpt}
    Given an $n$-qubit local Hamiltonian $H$, an initial state $\ket{\psi_0}$,
    a time range $T\in \bigO(\poly(n))$, and precision $\eps > 0$,
    the \search-$\dqpt$ problem is to find all critical times
    $t_c \in [0,T]$ of \nameref{apd:prm:promise_local_dqpt} to additive precision~$\eps$
    (i.e., output $\tilde{t}_c$ with $|\tilde{t}_c - t_c| \le \eps$ for each critical time),
    promised that consecutive critical times are separated by at least $\deltat_{\min} \ge 1/\poly(n)$.
\end{problem}
We begin with the near-term baseline protocol in \cref{apd:sec:analog}, 
which achieves shot-noise-limited $\bigO(M/\eps^2)$ scaling via independent projective measurements.
We then collect the Hamiltonian simulation and derivative estimation primitives 
for searching local DQPTs. 
The full analysis of the coherent quantum algorithm, including the oracle construction, and multi-time observable estimation subroutine in \cref{apd:sec:ft_proof}.

\subsection{Near-term baseline protocol and primitives}\label{apd:sec:analog}

The near-term protocol estimates the subsystem Loschmidt echo at $M$ time points independently via projective measurement.
This approach has been demonstrated experimentally by Karch et al.~\cite{karchProbingQuantumManybody2025}, 
who measured $\los_k(t)$ on a cesium quantum-gas microscope with $\sim 10^3$ samples per time point for sub-system sizes up to $k=7$.
We now demonstrate the cost of the near-term protocol.

\begin{lemma}[Sample complexity of single-time estimation]\label{lem:analog_chernoff}
    Given a time point $t$,
    estimating Loschmidt echo $\los_k(t)$ to additive precision $\eps$ with probability $\ge 1-\delta$
    requires $\bigO(\eps^{-2}\log(1/\delta))$ prepare--evolve--measure samples.
\end{lemma}
\begin{proof}
    Each measurement yields a Bernoulli sample with success probability $\los_k(t)$.
    By the Chernoff--Hoeffding bound, the empirical mean $\hat{\los}_k$ over $N$ samples satisfies $\Pr[|\hat{\los}_k - \los_k| > \eps] \le 2e^{-2N\eps^2}$, which drops below $\delta$ once $N = \bigO(\eps^{-2}\log(1/\delta))$.
\end{proof}

\begin{proposition}[Near-term multi-time observable estimation]\label{apd:thm:near_term_baseline}
    Given an $n$-qubit Hamiltonian $\ham$, an initial state $\ket{\psi_0}$,
    $M$ time points $0 < t_1 < \cdots < t_M \le \tmax$,
    and evolution oracles for $e^{-\ii\ham t}$,
    there exists a near-term protocol to estimate $\los_k(t_j)$ at all $M$ points
    to additive precision $\eps$ using
    $\bigO(M/\eps^2)$ samples
    and total Hamiltonian simulation time $\bigO(M\,\tmax/\eps^2)$,
    with no ancilla qubits.
\end{proposition}
\begin{proof}
    By \cref{lem:analog_chernoff}, each time point requires $\bigO(1/\eps^2)$ samples (setting $\delta = 1/3$).
    Summing over $M$ independent time points gives $\bigO(M/\eps^2)$ total samples.
    Each sample at time $t_j$ requires one application of $e^{-\ii\ham t_j}$,
    implemented by an analog simulator or a $p$th-order product formula (\cref{apd:thm:trotter}).
    Since $t_j \le \tmax$, the total Hamiltonian simulation time is at most $\bigO(M\,\tmax/\eps^2)$.
    No ancilla qubits are needed: the echo is estimated by direct projective measurement of the $k$-local projector $P$.
\end{proof}

\subsubsection{Hamiltonian dynamics simulation}\label{apd:sec:ham_sim}
All protocols for \nameref{apd:prm:search_dqpt} require implementing the time evolution $e^{-\ii \ham t}$ as a subroutine.
Recall the Hamiltonian simulation problem:
given a local Hamiltonian $H$ on $n$ qubits, initial state $\ket{\psi_0}$, evolution time $t$, and precision $\eps$,
the task is to prepare a state $\ket{\psi}$ such that $\norm{\ket{\psi} - e^{-\ii Ht}\ket{\psi_0}} \le \eps$.
The near-term algorithms for this problem are based on Trotter--Suzuki product formulas \cite{lloydUniversalQuantumSimulators1996,childsTheoryTrotterError2021,zhaoEntanglementAcceleratesQuantum2025}.
\begin{theorem}[Trotter-Suzuki $p$th-order product formula {\cite{childsTheoryTrotterError2021}}]\label{apd:thm:trotter}
    Let $H=\sum_{l=1}^L H_l$ be a Hermitian operator consisting with $L$ summands, and let $t \ge 0$.
    Let  $\pfU_p(t)$ be a $p$th-order product formula.
    Define $\alpha_{p} :=\sum\limits_{l_1,\dots,l_{p+1}=1}^{L} \norm{[H_{l_1},[H_{l_2},\dots,[H_{l_p},H_{l_{p+1}}]]]}$.
    Then, the Trotter error $\norm{\pfU_p(t)-e^{-\ii Ht}}$ can be asymptotically bounded as
    $\bigO(\alpha_{p} t^{p+1})$.
    In other words, we need Trotter steps
    $\trsteps=\bigO\qty(\alpha_{p}^{1/p} t^{1+1/p}/\epsilon^{1/p})$
    for precision $\norm{\pfU^R_p(t/r)-e^{-\ii Ht}}\le \epsilon$.
\end{theorem}
The best known algorithms for this problem are based on quantum signal processing (QSP) and quantum singular value transformation (QSVT) \cite{lowOptimalHamiltonianSimulation2017, gilyenQuantumSingularValue2019}, which achieve near-optimal query complexity scaling as $\bigO(t + \log(1/\eps))$.
\begin{lemma}[No-fast-forwarding theorem, lower bound in time {\cite{berryEfficientQuantumAlgorithms2007,berryHamiltonianSimulationNearly2015}}]\label{apd:thm:no_fast_forwarding}
    For any $d$-sparse Hamiltonian $\ham$ with $\|\ham\|_{\max} \le 1$ and evolution
    time $t > 0$, any quantum algorithm that simulates $e^{-\ii \ham t}$ to constant
    precision requires $\Omega(dt)$ queries to the Hamiltonian oracle.
\end{lemma}

The no-fast-forwarding theorem implies a lower bound in $t$ for \nameref{apd:prm:search_dqpt}:
any quantum algorithm to search for a critical time in the range $[0, \tmax]$ requires $\Omega(\tmax)$ time.

\subsubsection{Commutator estimator for time-derivative of observable}\label{apd:sec:comm_estimator}
Detecting a DQPT requires estimating the time-derivative $\dot{r}_k(t)$, which in turn reduces to estimating $\dot{\mathcal{L}}_k(t)$.
Rather than finite-differencing two nearby echo values, one can compute the derivative exactly from a single time point via the Heisenberg equation of motion.
For any time-independent observable $O$, the expectation value in the Schrödinger-picture state $\ket{\psi(t)} = e^{-\ii \ham t}\ket{\psi_0}$ satisfies
\begin{equation}\label{eq:heisenberg_eom}
    \dv{t}\expval{O}_t = \ii\bra{\psi(t)}[\ham, O]\ket{\psi(t)}.
\end{equation}
Specializing to $O = \projk$, the commutator method uses the identity
\begin{equation}\label{eq:comm_identity}
    \dot{\los}_k(t) = \bra{\psi(t)} Q \ket{\psi(t)}, \qquad Q := \ii[\ham, \projk],
\end{equation}
which reduces derivative estimation to a single expectation value
on the same state $\ket{\psi(t)}$ used for echo estimation.
The observable $Q$ is Hermitian ($Q^\dagger = -\ii[\projk, \ham] = Q$)
with operator norm $\norm{Q} \le 2\norm{\ham}$.
For geometrically local $\ham = \sum_{l=1}^{L} a_l H_l$ and $k$-local $\projk$,
most commutators $[H_l, \projk]$ vanish---only the $L_\partial$ terms
whose support overlaps that of $\projk$ contribute.
Decomposing $Q$ into Pauli strings $Q = \sum_{j=1}^{L_\partial} b_j Q_j$,
the Pauli $1$-norm satisfies
\begin{equation}\label{eq:comm_1norm}
    \norm{Q}_1 := \sum_{j} \abs{b_j} \le 2\norm{\ham}_{1,\partial P}
    := 2\!\sum_{\substack{l:\,\supp(H_l)\\\cap\,\supp(\projk)\neq\emptyset}} \abs{a_l},
\end{equation}
which depends only on boundary terms.
For a geometrically local Hamiltonian on a $d$-dimensional lattice with $k$-local $\projk$,
$L_\partial = \bigO(k)$ and $\norm{Q}_1 = \bigO(k\,J_{\max})$
where $J_{\max}$ is the maximum coupling strength,
independent of system size $n$ and the total number of Hamiltonian terms $L$.

\emph{Protocol.}
(1)~Prepare $\ket{\psi(t)} = e^{-\ii\ham t}\ket{\psi_0}$ via a $p$th-order product formula
with $\trsteps$ Trotter steps, each of gate depth $\bigO(L)$.
(2)~Decompose $Q = \ii[\ham, \projk]$ into $L_\partial$ Pauli strings $\{Q_j\}$
and estimate each $\expval{Q_j}$ from $\mathcal{N}_j$ samples,
allocating $\mathcal{N}_j \propto \abs{b_j}$ for optimal variance reduction.
All samples use the \emph{same} Trotter circuit with different Pauli measurement bases.
(3)~Compute $\hdlos[k,\cmm] = \sum_j b_j \hat{q}_j$
and form $\hdkrate(t) = -(1/k)\,\hdlos[k,\cmm]/\hat{\los}_k(t)$.
\begin{lemma}
    [Sample complexity of estimating $\dkrate(t)$]
    \label{apd:lem:comm_complexity}
    Let $\ham = \sum_{l=1}^{L} a_l H_l$ act on $n$ qubits,
    and suppose $\los_k(t) \ge \los_{\min} > 0$.
    The commutator estimator achieves additive error $\eps$ in $\dkrate(t)$ using
    $\mathcal{N} = \bigO\!\left(\frac{\norm{Q}_1^2}{(k\,\los_{\min}\,\eps)^2}\right)$
    samples.
    For geometrically local $\ham$ and $k$-local $\projk$,
    $\norm{Q}_1 \le 2\norm{\ham}_{1,\partial P}$ (\cref{eq:comm_1norm}).
\end{lemma}

\begin{proof}

Each Pauli observable $Q_j$ satisfies $Q_j^2 = I$,
so measuring $Q_j$ on state $\ket{\psi(t)}$ yields outcomes $\pm 1$
with variance $\le 1$.
With $\mathcal{N}_j$ samples allocated to $Q_j$,
$\mathrm{Var}(\hat{q}_j) \le 1/\mathcal{N}_j$.
For the weighted sum $\hdlos[k,\cmm] = \sum_j b_j \hat{q}_j$,
\begin{equation}
    \mathrm{Var}(\hdlos[k,\cmm]) = \sum_j b_j^2 / \mathcal{N}_j.
\end{equation}
Optimal allocation $\mathcal{N}_j = \mathcal{N} \cdot \abs{b_j}/\norm{Q}_1$ (where $\mathcal{N} = \sum_j \mathcal{N}_j$) yields
$\mathrm{Var}(\hdlos[k,\cmm]) = \norm{Q}_1^2/\mathcal{N}$.

Since $\dkrate = -(1/k)\,\dot{\los}_k/\los_k$,
the statistical error is amplified by $1/(k\,\los_{\min})$.
Requiring $\norm{Q}_1/(\sqrt{\mathcal{N}}\,k\,\los_{\min}) \le \eps$
gives $\mathcal{N} = \bigO\!\left(\norm{Q}_1^2/(k\,\los_{\min}\,\eps)^2\right)$.
\end{proof}

\subsubsection{Near-term search protocol}\label{apd:sec:near_term_search}

The echo estimator (\cref{apd:thm:near_term_baseline}), 
the commutator derivative estimator (\cref{apd:lem:comm_complexity}), 
and the product-formula simulation (\cref{apd:thm:trotter}) combine into a two-phase near-term protocol for \nameref{apd:prm:search_dqpt} without ancilla.
A coarse \emph{screening} pass tests the magnitude and susceptibility conditions of \nameref{apd:prm:promise_local_dqpt} on a uniform grid to bracket every critical time, and a subsequent \emph{bisection} pass localizes each bracketed time to the target precision $\eps$.

\begin{proposition}[Near-term Search-DQPT cost]\label{apd:prop:uniform_grid}
    Let $\ham$ be an $n$-qubit local Hamiltonian with bounded couplings, $\ket{\psi_0}$ an initial state, $\tmax\in\poly(n)$ a time range, $k=\bigO(1)$ the subsystem size, and $\eps>0$ a precision.
    Assume the critical times of \nameref{apd:prm:promise_local_dqpt} in $[0,\tmax]$ are separated by at least $\deltat_{\min}\in\Omega(1/\poly(n))$, with window offset $\deltat=\Theta(\deltat_{\min})$ and a uniform echo floor $\los_k(t)\ge\los_{\min}=\Omega(1)$ on $[0,\tmax]$.
    Then the near-term protocol solves \nameref{apd:prm:search_dqpt} using total Hamiltonian simulation time
    \begin{equation}\label{apd:eq:near_term_search_cost}
        \bigO\!\left(\frac{\tmax^2}{\deltat_{\min}}\right) + \bigO\!\left(N_c\,\tmax\log\frac{\deltat_{\min}}{\eps}\right),
    \end{equation}
    where $N_c$ is the number of critical times in $[0,\tmax]$, with $\bigO(n)$ qubits and no ancilla.
\end{proposition}
\begin{proof}
    \emph{Screening.}
    Place a uniform grid of $M_0 = \Theta(\tmax/\deltat_{\min})$ points on $[0,\tmax]$.
    At each $t_j$, estimate the magnitude $\krate(t_j)$ from the subsystem echo (\cref{apd:thm:near_term_baseline}) and the susceptibility $J_k(t_j) = \dkrate(t_j-\deltat) - \dkrate(t_j+\deltat)$ (\cref{apd:eq:jump_def}) from the commutator derivative estimator (\cref{apd:lem:comm_complexity}), flagging the cell when $\krate(t_j)\ge\xi_1$ and $J_k(t_j)\ge J_1$.
    The commutator estimator reaches additive precision $\eps'$ in $\bigO(\norm{Q}_1^2/(k\los_{\min}\eps')^2)$ samples with $\norm{Q}_1 = \bigO(kJ_{\max})$ (\cref{eq:comm_1norm});
    since the promise gaps $\xi_1-\xi_0 = \Omega(1)$ and $J_1/J_0\ge c>1$ are constant for constant $k$ and the echo floor $\los_{\min}=\Omega(1)$, a constant target $\eps'=\Theta(1)$ separates \yes{} from \no{}, so $\bigO(1)$ samples per point suffice.
    Each sample evolves to a time $t_j\le\tmax$ at $\bigO(\tmax)$ simulation time, so screening costs $\bigO(M_0\,\tmax) = \bigO(\tmax^2/\deltat_{\min})$.

    \emph{Bisection.}
    Each flagged cell brackets a single critical time, at which $\dkrate$ changes sign.
    A binary search performs $\ceil{\log_2(\deltat_{\min}/\eps)} = \bigO(\log(\deltat_{\min}/\eps))$ sign tests of $\dkrate$ at the midpoint, each a constant-margin decision at $\bigO(1)$ samples and $\bigO(\tmax)$ simulation time, valid because the uniform floor $\los_{\min}=\Omega(1)$ keeps the derivative estimable up to the kink.
    Over the $N_c$ flagged cells this costs $\bigO(N_c\,\tmax\log(\deltat_{\min}/\eps))$.

    \emph{Simulation.}
    Each evolution $e^{-\ii\ham t}$ with $t\le\tmax$ is implemented to constant precision by a $p$th-order product formula, requiring $\bigO(\alpha_p^{1/p}\,\tmax^{1+1/p})$ Trotter steps of gate depth $\bigO(L)$ (\cref{apd:thm:trotter}), and the protocol measures the $k$-local projector and the boundary Pauli strings of $Q=\ii[\ham,\projk]$ directly, with no ancilla.
    Summing the two phases gives \cref{apd:eq:near_term_search_cost}.
\end{proof}

The coherent quantum algorithm of \cref{apd:cor:search_dqpt} screens the same grid by gradient estimation algorithm in $\widetilde{\bigO}(\tmax^{3/2}/\deltat_{\min}^{1/2})$, 
a quadratic $\sqrt{M_0} = \sqrt{\tmax/\deltat_{\min}}$ speedup over the near-term screening of \cref{apd:prop:uniform_grid}.

\subsection{Hamiltonian snapshots estimation}\label{apd:sec:ft_proof}

The simulation and estimation primitives of the preceding subsection yield the near-term baseline, 
which remains shot-noise limited at $\bigO(1/\eps^2)$ per evaluation point and needs $\bigO(M)$ rounds of prepare-evolve-measure for $M$ time points.
We show that a coherent quantum algorithm overcomes both, enabling a faster search for local critical times.
We first treat the general task of taking snapshots of Hamiltonian dynamics: measuring observables on the evolved state at multiple time points.

\begin{problem}[Hamiltonian snapshots estimation]\label{apd:prm:multi_time_obs}
    Given an $n$-qubit Hamiltonian $\ham$, an initial state $\ket{\psi_0}$ prepared by $U_\psi$,
    $M$ time points $\{t_j\}_{j=1}^M$ in a range $[0,T]$, 
    bounded Hermitian observables $\{O_j\}_{j=1}^M$ with $\|O_j\| \le 1$,
    and precision $\eps > 0$,
    output estimates $\tilde{o}_j$ satisfying
    $|\tilde{o}_j - \bra{\psi_0} e^{\ii \ham t_j} O_j e^{-\ii \ham t_j} \ket{\psi_0}| \le \eps$ for all $j \in [M]$.
    The goal is to minimize total queries to $U_\psi$ and total Hamiltonian simulation time.
\end{problem}

\subsubsection{Quantum gradient estimation algorithm}
We now present this coherent algorithm, 
whose $\bigO(\sqrt{M}/\eps)$ query complexity (\cref{apd:cor:adaptive_snapshots}) reflects two distinct quantum speedups over the baseline (\cref{apd:thm:near_term_baseline}).
The first is \emph{sublinear scaling in $M$}: the $M$ target expectation values are encoded as the gradient of a single function and read out together by quantum gradient estimation, so all $M$ are obtained in $\bigO(\sqrt{M})$ queries rather than the $\bigO(M)$ of estimating them one by one.
The second is \emph{Heisenberg-limited precision}: encoding each target as a measurement probability and extracting it by coherent amplitude estimation costs $\bigO(1/\eps)$ oracle calls, beating the $\bigO(1/\eps^2)$ shot-noise limit of prepare-and-measure sampling.
Both speedups are optimal: the lower bound $\Omega(\sqrt{M}/\eps)$ (\cref{apd:thm:multi_lower_bound}) establishes $\sqrt{M}/\eps$ as the fundamental price of estimating $M$ correlated quantities through a single oracle interface.

We obtain the $M$ Heisenberg-picture expectation values $\langle O_j(t_j)\rangle$ in three steps.
First, we build a parameterized unitary $U(\vbx)$ whose gradient at $\vbx = \vbzero$ encodes all $M$ of them, satisfying the derivative growth condition the gradient algorithm requires (\cref{apd:lem:oracle_construction}).
Second, we realize a probability oracle for the associated function $f$ through the Hadamard test (\cref{apd:cor:oracle_circuit,def:probability_oracle}).
Third, we feed this oracle into the gradient estimation algorithm (\cref{thm:gradient_expectation}), yielding the multi-time snapshot estimates (\cref{apd:cor:adaptive_snapshots}).

\begin{lemma}[Interleaved parameterized unitary construction]\label{apd:lem:oracle_construction}
    Given an $n$-qubit Hamiltonian $\ham$, an initial state $\ket{\psi_0}$,
    $M$ time points $0 < t_1 < \cdots < t_M \le \tmax$, and
    Hermitian operators $O_1, \ldots, O_M$ with $\|O_j\| \le 1$,
    let $U(\vbx)$ be the parameterized unitary with parameters $\vbx = (x_1, \ldots, x_M)$
    \begin{align}\label{apd:eq:param_unitary}
        U(\vbx) 
        &:= \qty(\prod_{j=1}^{M} U(t_{j-1}, t_j)\,e^{-2\ii x_j O_j}) U(t_M, t_0) \\
        &= U(t_0, t_1)\,e^{-2\ii x_1 O_1}\,U(t_1, t_2)\,e^{-2\ii x_2 O_2}\,
        \cdots\, U(t_{M-1}, t_M)\,e^{-2\ii x_M O_M}\,U(t_M, t_0),
    \end{align}
    where $U(t_a, t_b) := e^{-\ii \ham (t_a - t_b)}$ is the time-evolution operator,
    and let
    \begin{equation}\label{apd:eq:f_gradient}
        f(\vbx) := -\tfrac{1}{2}\Im\bigl[\bra{\psi_0}U(\vbx)\ket{\psi_0}\bigr] + \tfrac{1}{2}.
    \end{equation}
    Then the gradient of $f$ at $\vbzero$ encodes Heisenberg-picture expectation values:
    $\partial_{x_\ell} f|_{\vbx=\vbzero} = \bra{\psi_0} O_\ell(t_\ell) \ket{\psi_0}$ for all $\ell \in [M]$,
    where $O_\ell(t) := e^{\ii \ham t}\,O_\ell\,e^{-\ii \ham t}$.
    Moreover, $f$ satisfies the derivative growth condition
    $|\partial_\alpha f(\vbzero)| \le c^m m^{m/2}$ for all multi-indices $\alpha$ with $|\alpha| = m$, with constant $c = 2$.
\end{lemma}
\begin{proof}
    At $\vbx = \vbzero$, each rotation $e^{-2\ii \cdot 0 \cdot O_j} = \I$, 
    so the product in \cref{apd:eq:param_unitary} telescopes:
    $$U(\vbzero) = U(t_0, t_1)\,U(t_1, t_2)\cdots U(t_{M-1}, t_M)\,U(t_M, t_0) = U(t_0, t_0) = \I.$$
    Thus $f(\vbzero) = 1/2 \in [0,1]$.

    Differentiating \cref{apd:eq:param_unitary} with respect to $x_\ell$ at $\vbx = \vbzero$:
    \begin{align}\label{eq:dU_dxl}
        \frac{\partial U}{\partial x_\ell}\bigg|_{\vbx=\vbzero}
        &= U(t_0, t_1)\cdots U(t_{\ell-1}, t_\ell)\,(-2\ii O_\ell)\,U(t_\ell, t_{\ell+1})\cdots U(t_M, t_0) \\
        &= U(t_0, t_\ell)\,(-2\ii O_\ell)\,U(t_\ell, t_0) \\
        &= -2\ii\, e^{\ii \ham t_\ell}\,O_\ell\,e^{-\ii \ham t_\ell},
    \end{align}
    where the second line telescopes the products at $\vbx = \vbzero$ and the third uses $U(t_0, t_\ell) = e^{\ii \ham t_\ell}$.
    Hence
    \begin{align}\label{apd:eq:gradient_expectation}
        \partial_{x_\ell} f\big|_{\vbzero}
        &= -\tfrac{1}{2}\Im\bigl[\bra{\psi_0}\bigl(-2\ii\, e^{\ii \ham t_\ell}\,O_\ell\,e^{-\ii \ham t_\ell}\bigr)\ket{\psi_0}\bigr] \\
        &= -\tfrac{1}{2}\Im\bigl(-2\ii\,\bra{\psi_0} O_\ell(t_\ell) \ket{\psi_0}\bigr) \\
        &= \bra{\psi_0} O_\ell(t_\ell) \ket{\psi_0},
    \end{align}
    where the last step uses $\bra{\psi_0} O_\ell(t_\ell) \ket{\psi_0} \in \mathbb{R}$.

    \emph{Derivative growth condition.}
    Consider a multi-index $\alpha$ with $|\alpha| = m$.
    The $m$th-order partial derivative $\partial_\alpha U(\vbx)$ involves $m$ insertions of $(-2\ii O_j)$ into the product \cref{apd:eq:param_unitary}, with each insertion bounded in operator norm by $2\|O_j\| \le 2$ (since $\|O_j\| \le 1$ by hypothesis).
    All time-evolution operators are unitary and contribute norm at most $1$.
    Therefore $\|\partial_\alpha U(\vbzero)\| \le 2^m$ and
    \begin{equation}
        |\partial_\alpha f(\vbzero)| \le \tfrac{1}{2}\|\partial_\alpha U(\vbzero)\| \le 2^{m-1}.
    \end{equation}
    Since $m^{m/2} \ge 1$ for $m \ge 1$, we have $2^{m-1} \le 2^m m^{m/2}$, confirming the derivative growth condition required by the gradient estimation algorithm~\cite{hugginsNearlyOptimalQuantum2022} (\cref{thm:gradient_expectation}) with constant $c = 2$.
The analyticity that the gradient estimation algorithm (\cref{thm:gradient_expectation}) requires is in the auxiliary variable $\vbx$, not the physical time $t$: $f$ is entire in $\vbx$ because each factor $e^{-2\ii x_j O_j}$ is entire for bounded $O_j$.
\end{proof}

To state the corollary precisely, we first recall the notion of a probability oracle
and the Hadamard test that implements it.

\begin{definition}[probability oracle]\label{def:probability_oracle}
    Consider a function $f:\mathbb{R}^M\to [0,1]$.
    A \emph{probability oracle} $U_f$ for $f$ is a unitary operator that acts as 
    \begin{equation}\label{eq:probability_oracle}
        U_f: \ket{\vbx} \ket{\vb{0}} \mapsto
        \ket{\vbx} \qty(
        \sqrt{f(\vbx)} \ket{1} \ket{\phi_1(\vbx)} + 
        \sqrt{1-f(\vbx)} \ket{0} \ket{\phi_0(\vbx)}  
        )
    \end{equation}
    where $\ket{\vbx}$ denotes a discretization of the variable $\vbx$ encoded into a register of qubits,
    $\ket{\vb{0}}$ denotes the all-zeros state of a register of ancilla qubits,
    and $\ket{\phi_0(\vbx)}$ and $\ket{\phi_1(\vbx)}$ are arbitrary quantum states.
\end{definition}

Combining \cref{apd:lem:oracle_construction} with the Hadamard test, we obtain a probability oracle for $f$.

\begin{figure}[!t]
    \centering
    \begin{quantikz}[row sep=0.4cm, column sep=0.4cm]
        \lstick{$\ket{0}$\\[-2pt]{\scriptsize anc}}
            & \gate{S^\dagger\! H} & \ctrl{1} & \ctrl{1} & \ctrl{1} & \ldots
            & \ctrl{1} & \ctrl{1} & \gate{H} & \meter{} \\
        \lstick{$\ket{0}^{\otimes n}$\\[-2pt]{\scriptsize sys}}
            & \gate{U_\psi}
            & \gate[style={fill=orange!15}]{e^{-\ii\ham t_M}}
            & \gate[style={fill=violet!15}]{e^{-2\ii x_M O_M}}
            & \gate{e^{\ii\ham\Delta t_M}}
            & \ldots
            & \gate[style={fill=violet!15}]{e^{-2\ii x_1 O_1}}
            & \gate{e^{\ii\ham\Delta t_1}}
            & \qw & \qw
    \end{quantikz}
    \caption{Probability oracle $F(\vbx)$ for $f(\vbx)$ (\cref{apd:eq:hadamard_circuit}).
    The ancilla is prepared in $S^\dagger \hdm\ket{0} = \frac{1}{\sqrt{2}}(\ket{0} - \ii\ket{1})$,
    then controls each gate in the parameterized unitary $U(\vbx)$ (\cref{apd:eq:param_unitary}).
    The sequence alternates Hamiltonian evolution segments $e^{\pm\ii\ham\Delta t_j}$
    with rotations $e^{-2\ii x_j O_j}$ encoding the parameters $\vbx$.
    The initial evolution $e^{-\ii\ham t_M}$ is shaded orange; the parameter-encoding rotations $e^{-2\ii x_j O_j}$ are shaded violet.
    A final Hadamard yields $\Pr[\text{anc}=1] = f(\vbx)$ (\cref{apd:eq:f_gradient}).
    This is the building block called by the gradient estimation protocol (\cref{apd:fig:circuit_protocol}). 
    }
    \label{apd:fig:circuit_oracle}
\end{figure}
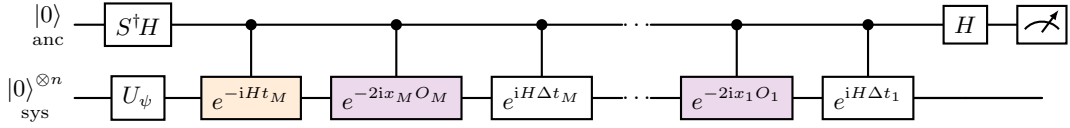

\begin{corollary}[Probability oracle construction via Hadamard test]\label{apd:cor:oracle_circuit}
    Let $U_\psi$ be the state-preparation unitary ($U_\psi\ket{\vb{0}} = \ket{\psi_0}$).
    The Hadamard test circuit
    \begin{equation}\label{apd:eq:hadamard_circuit}
        F(\vbx) := (\hdm \otimes \I)\,(c\text{-}U(\vbx))\,(S^\dagger \hdm \otimes U_\psi),
    \end{equation}
    where $\hdm$ is the Hadamard gate, $S = \operatorname{diag}(1,i)$ the phase gate,
    and $c\text{-}U(\vbx)$ the controlled-$U(\vbx)$ from \cref{apd:eq:param_unitary},
    is a \nameref{def:probability_oracle} for the function $f$ of \cref{apd:eq:f_gradient}.
    Each query requires total Hamiltonian simulation time $2\tmax$ and $M$ controlled rotations $e^{-2\ii x_j O_j}$.
\end{corollary}
\begin{proof}
    We trace through the three stages of $F(\vbx) = (\hdm \otimes \I)\,(c\text{-}U(\vbx))\,(S^\dagger \hdm \otimes U_\psi)$
    acting on $\ket{0}\ket{\vb{0}}$.

    \emph{State preparation.}
    Since $\hdm\ket{0} = \frac{1}{\sqrt{2}}(\ket{0}+\ket{1})$ and
    $S^\dagger = \operatorname{diag}(1,-\ii)$, the ancilla is prepared as
    $S^\dagger\hdm\ket{0} = \frac{1}{\sqrt{2}}(\ket{0} - \ii\ket{1})$.
    The system register becomes $U_\psi\ket{\vb{0}} = \ket{\psi_0}$, giving
    \begin{equation}
        (S^\dagger \hdm \otimes U_\psi)\ket{0}\ket{\vb{0}}
        = \frac{1}{\sqrt{2}}\bigl(\ket{0} - \ii\ket{1}\bigr)\ket{\psi_0}.
    \end{equation}

    \emph{Controlled unitary.}
    The gate $c\text{-}U(\vbx)$ applies $U(\vbx)$ to the system register
    conditioned on the ancilla being $\ket{1}$:
    \begin{equation}
        \frac{1}{\sqrt{2}}\bigl(\ket{0}\ket{\psi_0} - \ii\ket{1}\,U(\vbx)\ket{\psi_0}\bigr).
    \end{equation}

    \emph{Final Hadamard.}
    Applying $\hdm \otimes \I$ and using
    $\hdm\ket{0} = \frac{1}{\sqrt{2}}(\ket{0}+\ket{1})$,
    $\hdm\ket{1} = \frac{1}{\sqrt{2}}(\ket{0}-\ket{1})$:
    \begin{equation}\label{apd:eq:F_expanded}
        F(\vbx)\ket{0}\ket{\vb{0}}
        = \frac{1}{2}\ket{0}\bigl(\ket{\psi_0} - \ii\,U(\vbx)\ket{\psi_0}\bigr)
        + \frac{1}{2}\ket{1}\bigl(\ket{\psi_0} + \ii\,U(\vbx)\ket{\psi_0}\bigr).
    \end{equation}

    \emph{Probability of outcome $\ket{1}$.}
    The squared norm of the $\ket{1}$ component is
    \begin{align}
        \frac{1}{4}\bigl\|\ket{\psi_0} + \ii\,U(\vbx)\ket{\psi_0}\bigr\|^2
        &= \frac{1}{4}\bigl(2 + \ii\bra{\psi_0}U(\vbx)\ket{\psi_0}
            - \ii\,\overline{\bra{\psi_0}U(\vbx)\ket{\psi_0}}\bigr) \notag\\
        &= \frac{1}{2} - \frac{1}{2}\Im\bigl[\bra{\psi_0}U(\vbx)\ket{\psi_0}\bigr]
        = f(\vbx),
    \end{align}
    where the second equality uses $\ii z - \ii\bar{z} = -2\Im(z)$.
    Setting
    $\ket{\phi_1(\vbx)} := \frac{\ket{\psi_0}+\ii\,U(\vbx)\ket{\psi_0}}{2\sqrt{f(\vbx)}}$
    and
    $\ket{\phi_0(\vbx)} := \frac{\ket{\psi_0}-\ii\,U(\vbx)\ket{\psi_0}}{2\sqrt{1-f(\vbx)}}$,
    \cref{apd:eq:F_expanded} becomes
    \begin{equation}
        F(\vbx)\ket{0}\ket{\vb{0}}
        = \sqrt{f(\vbx)}\,\ket{1}\ket{\phi_1(\vbx)}
        + \sqrt{1-f(\vbx)}\,\ket{0}\ket{\phi_0(\vbx)},
    \end{equation}
    which matches the \nameref{def:probability_oracle}.
    Note that this construction holds for arbitrary $\vbx \in \mathbb{R}^M$;
    the gradient estimation protocol (\cref{apd:fig:circuit_protocol,thm:gradient_expectation})
    queries $F$ at points near $\vbx = \vbzero$, where the expectation values
    $\bra{\psi_0} O_j(t_j) \ket{\psi_0}$ are recovered as the gradient components
    $\partial_{x_j} f|_{\vbzero}$ (\cref{apd:lem:oracle_construction}).

    Each application of $F(\vbx)$ requires
    one application of $U_\psi$ (state preparation),
    one application of the controlled-$U(\vbx)$ (which traverses all $M$ time segments forward from $t_0$ to $t_M$ and then backward from $t_M$ to $t_0$, for total Hamiltonian simulation time $2\tmax$),
    and $M$ controlled rotations $e^{-2\ii x_j O_j}$, whose gate cost depends on the structure of $O_j$.
\end{proof}

\begin{remark}
    \label{apd:rmk:inverse_oracle}
    A single forward query of $F$ uses only $U_\psi$.
    Although \cref{thm:gradient_expectation} is stated in terms of probability-oracle queries to $F$ (as in Huggins et al.), Gilyén's underlying gradient construction (after Jordan) operates on a \emph{phase} oracle $O_f:\ket{\vbx}\mapsto e^{\ii f(\vbx)}\ket{\vbx}$;
    the probability-to-phase conversion this entails, folded into that query count, is where $F^\dagger$ enters.
    Following Gilyén et al.~\cite{gilyenOptimizingQuantumOptimization2019}, this conversion is built on the Grover-like iterate $G = (2\Pi_1 - I)\,F^\dagger\,(2\Pi_2 - I)\,F$, where $\Pi_2$ projects onto the Hadamard-test ``$1$'' outcome and $\Pi_1$ onto the zeroed ancilla:
    $G$ rotates by $2\arcsin\sqrt{f(\vbx)}$ in the two-dimensional subspace attached to each $\ket{\vbx}$, and a linear combination of unitaries (LCU) in powers of $G$ imprints the phase $e^{\ii f(\vbx)}$ to precision $\eps'$ using $\bigO(\log(1/\eps'))$ applications of each of $F$ and $F^\dagger$.
    The inverse is thus intrinsic, exactly as in amplitude amplification: the construction alternates $F$ and $F^\dagger$ between the two reflections, with $F^\dagger$ uncomputing the Hadamard-test ancilla so that the encoded value $f(\vbx)$ acts coherently on the index register.
    The total query complexity in \cref{thm:multi_time_estimation} is accordingly $\widetilde{\bigO}(\sqrt{M}/\eps)$ calls to $U_\psi$ and the same number to $U_\psi^\dagger$, matching the convention of Refs.~\cite{hugginsNearlyOptimalQuantum2022,wadaHeisenbergLimitedAdaptiveGradient2025}, and is independent of the internal structure of $U(\vbx)$, whether or not it contains Hamiltonian-evolution segments.
\end{remark}

\begin{figure}[t!]
    \centering
    \begin{quantikz}[row sep=0.3cm, column sep=0.5cm]
        \lstick{$\ket{0}^{\otimes 3}$\\[-2pt]{\scriptsize $x_1$}}
            & \gate{H^{\otimes 3}} & \gate[style={fill=green!15}]{R(\tilde{u}_1^{(q)})} & \gate[wires=4][1.8cm]{F(\vbx)} & \gate{\mathrm{QFT}^{-1}} & \meter{} \\
        \setwiretype{n}
            & \vdots & \vdots & & \vdots & \\
        \lstick{$\ket{0}^{\otimes 3}$\\[-2pt]{\scriptsize $x_M$}}
            & \gate{H^{\otimes 3}} & \gate[style={fill=green!15}]{R(\tilde{u}_M^{(q)})} & & \gate{\mathrm{QFT}^{-1}} & \meter{} \\
        \lstick{$\ket{0}^{\otimes (n+1)}$\\[-2pt]{\scriptsize anc $+$ sys}}
            & \qw & \qw & & \qw & \trash{\text{disc.}}
    \end{quantikz}
    \caption{One round of the adaptive gradient estimation protocol for $\los_k$ snapshots estimation
    (Alg.1 in the main text).
    The circuit is repeated for rounds $q = 0, 1, \ldots, \ceil{\log_2(1/\eps)}$.
    Each index register $x_j$ ($p = 3$ qubits, independent of~$\eps$) is initialized in uniform superposition,
    then shifted by classical phase corrections $R(\tilde{u}_j^{(q)})$ (shaded green)
    based on previous-round estimates.
    The probability oracle $F(\vbx)$ (\cref{apd:fig:circuit_oracle}), 
    unchanged across rounds,
    acts on all registers, encoding $f(\vbx)$ into the ancilla measurement probability.
    Inverse QFT on the index registers followed by measurement extracts
    the gradient $\nabla f|_{\vbzero}$, 
    whose components equal
    $\los_k(t_j)$.
    The estimates are then updated classically as
    $\tilde{u}_j^{(q+1)} = \tilde{u}_j^{(q)} + \pi\,2^{-q} g_j^{(q)}$.
    The ancilla and system qubits are discarded.
    The two quantum speedups originate from distinct circuit elements:
    the superposition over index registers combined with the inverse QFT
    enables $\bigO(\sqrt{M})$ scaling by probing all $M$ gradient directions in parallel,
    while the probability encoding in $F(\vbx)$ enables coherent amplitude estimation
    (a Grover-like procedure using repeated calls to $F$ and $F^\dagger$),
    achieving Heisenberg-limited $\bigO(1/\eps)$ precision.
    See the main text for the fully expanded circuit.
    }
    \label{apd:fig:circuit_protocol}
\end{figure}
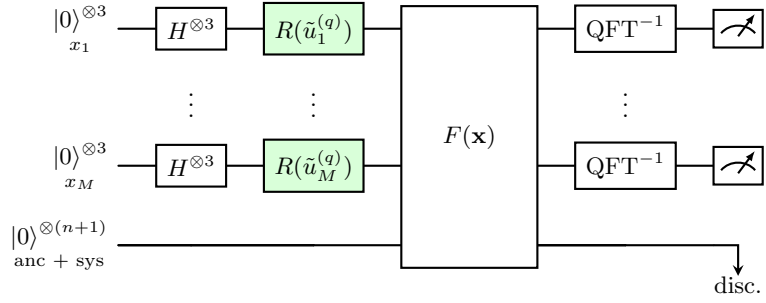

The oracle $F(\vbx)$ satisfies all the prerequisites of the following gradient estimation algorithm, 
which we apply in the proof of \cref{apd:cor:adaptive_snapshots}.

\begin{theorem}[Gradient-based estimation algorithm, rephrase \cite{hugginsNearlyOptimalQuantum2022,gilyenOptimizingQuantumOptimization2019}]\label{thm:gradient_expectation}
    Let $\eps$, $c\in\mathbb{R}_+$ be fixed constants, with $\eps\le c$.
    Let $M\in\mathbb{Z}_+$ and $\vbx\in\mathbb{R}^M$.
    Suppose that $f:\mathbb{R}^M\to \mathbb{R}$ is an analytic function such that for every $k\in\mathbb{Z}_+$,
    the following bound holds for all $k$th order partial derivatives of $f$ at $\vbx$
    (denoted by $\partial_{\alpha}f(\vbx)$): $\abs{\partial_{\alpha}f(\vbx)}\le c^k k^{(k/2)}$.
    Then, there is a quantum algorithm that outputs an estimate $\tilde{\vb{g}}\in\mathbb{R}^M$
    such that $\norm{\nabla f(\vbx) - \tilde{\vb{g}}}_{\infty} \le \eps$,
    with probability at least $1-\delta$.
    This algorithm makes $\bigO(c\sqrt{M}\log(M/\delta)/\eps)$ queries to a probability oracle for $f$.
\end{theorem}

    The sublinear $\bigO(\sqrt{M})$ scaling rests on \emph{quantum gradient estimation}: the parameterized unitary \cref{apd:eq:param_unitary} encodes the $M$ target expectation values as the components of the gradient $\nabla f|_{\vbzero}$ of a single function, and the optimal gradient algorithm of Gilyén et al.~\cite{gilyenOptimizingQuantumOptimization2019}---a refinement of Jordan's original single-query gradient method~\cite{jordanFastQuantumAlgorithm2005}---reads out all $M$ components together.
    Preparing $M$ index registers in uniform superposition and applying the probability oracle $F$ controlled on the index (\cref{apd:fig:circuit_protocol}) probes $f$ along all $M$ directions in one coherent query; an inverse QFT then demultiplexes the accumulated phases into the individual partial derivatives.
    Because the uniform superposition shares each query's amplitude across the $M$ directions, every component carries only $\sim 1/\sqrt{M}$ of the signal, so amplitude amplification needs $\bigO(\sqrt{M})$ rounds to resolve it.
    This $\sqrt{M}$ factor is optimal: estimating $M$ expectation values to additive error $\eps$ through a black-box state-preparation oracle requires $\Omega(\sqrt{M}/\eps)$ queries~\cite{hugginsNearlyOptimalQuantum2022} (\cref{apd:thm:multi_lower_bound}).

    The $\bigO(1/\eps)$ precision, in place of the $\bigO(1/\eps^2)$ standard quantum limit, is the Heisenberg scaling, and it comes from coherent access to the signal.
    The Hadamard test encodes $f(\vbx)$ as a measurement probability, an amplitude squared (\cref{apd:fig:circuit_oracle}), which the gradient algorithm queries coherently so that the signal phase accumulates linearly in the number of oracle calls; classical sampling, by contrast, lowers the statistical error only as $1/\sqrt{N_{\mathrm{samples}}}$, giving $\bigO(1/\eps^2)$.
    The higher-order central-difference formulas of Gilyén et al.~\cite{gilyenOptimizingQuantumOptimization2019} are what secure this $\bigO(1/\eps)$ scaling for the full gradient, a quadratic improvement in $\eps$ over Jordan's original method.
    Realizing it requires coherent access to both $F$ and $F^\dagger$, which is why the protocol is fault-tolerant rather than near-term.

\subsubsection{The algorithm for Hamiltonian snapshot estimation}
In our setting, the expectation values $\bra{\psi_0} O_j(t_j) \ket{\psi_0}$ for $j = 1, \ldots, M$
are precisely the $M$ components of $\nabla f|_{\vbzero}$ \cref{eq:dU_dxl}.
Applying \cref{thm:gradient_expectation} to the probability oracle $F(\vbx)$ of \cref{apd:cor:oracle_circuit}
therefore extracts all $M$ as a single gradient estimation.
Thus, we have the following quantum algorithm for \nameref{apd:prm:multi_time_obs}.
For \DQPT, setting $O_j = P$ for all $j$ recovers the Loschmidt echo values $\los_k(t_j)$.

\begin{theorem}[Non-adaptive snapshots estimation]\label{thm:non_adaptive_snapshots}
    Given an $n$-qubit Hamiltonian $\ham$, an initial state $\ket{\psi_0}$ prepared by a unitary $U_\psi$,
    $M$ time points $0< t_1 < \cdots < t_M \le \tmax$, and bounded Hermitian observables $\qty{O_j}_{j=1}^M$ with $\|O_j\| \le 1$,
    there exists a quantum algorithm that outputs estimates $\tilde{o}_j$ satisfying
    $|\tilde{o}_j - \bra{\psi_0} O_j(t_j) \ket{\psi_0}| \le \eps$ for all $j \in [M]$
    with probability $\ge 1 - \delta$, using:
    \begin{itemize}
        \item $\bigO(\eps^{-1}\sqrt{M}\log(M/\delta))$ queries to the probability oracle $F$,
        \item $\bigO(\eps^{-1}\sqrt{M}\log(M/\delta) \cdot \tmax)$ total Hamiltonian simulation time,
        \item $\bigO(M\log(1/\eps) + n)$ qubits.
    \end{itemize}
    Each oracle query further requires $M$ controlled rotations $e^{-2\ii x_j O_j}$.
\end{theorem}
\begin{proof}
By \cref{apd:lem:oracle_construction,apd:cor:oracle_circuit},
the circuit $F(\vbx)$ is a probability oracle for $f(\vbx)$ satisfying the derivative growth condition with constant $c = 2$, and $\partial_{x_\ell} f|_{\vbzero} = \bra{\psi_0} O_\ell(t_\ell) \ket{\psi_0}$ for all $\ell \in [M]$.

\emph{Index registers and phase accumulation.}
The gradient estimation algorithm~\cite{gilyenOptimizingQuantumOptimization2019,hugginsNearlyOptimalQuantum2022} introduces $M$ index registers, each consisting of $p = \bigO(\log(1/\eps))$ qubits, initialized in uniform superposition over the grid $G_p = \{(\mu - 2^{p-1} + \tfrac{1}{2})/2^p : \mu = 0, \ldots, 2^p-1\}$.
The joint state $\ket{\vbx} = \ket{x_1}\cdots\ket{x_M}$ ranges over $G_p^M$.
The probability oracle $F(\vbx)$ is then applied $\bigO(2^p \sqrt{M})$ times in a controlled fashion, accumulating phases proportional to $f(\vbx)$ in the index registers.
Near $\vbx = \vbzero$, $f(\vbx) \approx \sum_j x_j \bra{\psi_0} O_j(t_j) \ket{\psi_0}$ (\cref{apd:eq:gradient_expectation}), 
so the phase in the $j$th register encodes the target expectation value $\bra{\psi_0} O_j(t_j) \ket{\psi_0}$.

\emph{Linearization.}
The function $f$ is not exactly linear: higher-order terms $\partial_\alpha f$ contribute to the phase.
The derivative growth condition $|\partial_\alpha f(\vbzero)| \le c^k k^{k/2}$ (\cref{apd:lem:oracle_construction}) ensures that these contributions can be suppressed by the central-difference structure of the grid $G_p$, combined with uniform singular value amplification~\cite{gilyenOptimizingQuantumOptimization2019}.
This linearization step accounts for the $\bigO(\sqrt{M})$ factor in the query complexity: the uniform superposition over index registers gives each gradient direction amplitude $1/\sqrt{M}$, and amplitude amplification requires $\bigO(\sqrt{M})$ rounds to extract the signal.

\emph{Readout.}
An inverse quantum Fourier transform on each $p$-qubit index register, followed by computational-basis measurement, extracts the gradient components $\partial_{x_j} f|_{\vbzero} = \bra{\psi_0} O_j(t_j) \ket{\psi_0}$ to precision $\eps$ (\cref{apd:fig:circuit_protocol}).
Coordinate-wise median amplification over $\bigO(\log(M/\delta))$ independent copies boosts the joint success probability from a per-coordinate constant to $1 - \delta$ for all $j$ simultaneously.

\emph{Query complexity.}
Applying \cref{thm:gradient_expectation} with $c = 2$ yields
$\bigO(c\sqrt{M}\log(M/\delta)/\eps) = \bigO(\eps^{-1}\sqrt{M}\log(M/\delta))$ queries to $F$.
For constant failure probability ($\delta = 1/3$), this is $\bigO(\eps^{-1}\sqrt{M}\log M)$.

\emph{Per-query cost.}
By \cref{apd:cor:oracle_circuit}, the telescoping structure of the parameterized unitary \cref{apd:eq:param_unitary} ensures that each query to $F$ uses Hamiltonian simulation time $2\tmax$ and one call each to $U_\psi$ and $U_\psi^\dagger$.
The total Hamiltonian simulation time is therefore $\bigO(\eps^{-1}\sqrt{M}\log(M/\delta) \cdot \tmax)$.

\emph{Space.}
The circuit uses $n$ system qubits, $Mp = \bigO(M\log(1/\eps))$ ancilla qubits for the $M$ index registers of $p = \bigO(\log(1/\eps))$ qubits each, and $1$ Hadamard-test ancilla, totalling $\bigO(M\log(1/\eps) + n)$.
\end{proof}

A guarantee stronger than the additive-error bound of \cref{thm:non_adaptive_snapshots} is the mean-squared error $\mathrm{MSE}[\hat{o}_j] = \mathbb{E}[(\hat{o}_j - \bra{\psi_0} O_j(t_j) \ket{\psi_0})^2]$,
which averages over \emph{all} measurement outcomes rather than tolerating an arbitrarily large error on a failure branch of probability $\le 1/3$.
A failure of probability $\delta$ that incurs $\bigO(1)$ error contributes $\Omega(\delta)$ to the MSE, so $\mathrm{MSE} \le \eps^2$ forces the failure probability to satisfy $\delta \lesssim \eps^2$.
Meeting it costs the non-adaptive algorithm on two fronts.
In \emph{queries}, the requirement $\delta \lesssim \eps^2$ enters the median-amplification factor $\log(M/\delta)$, raising the count from $\bigO(\eps^{-1}\sqrt{M}\log M)$ at constant success to $\bigO(\eps^{-1}\sqrt{M}\log(M/\eps))$.
The \emph{space} cost is more consequential: resolving each gradient component to precision $\eps$ in a single readout demands $p = \bigO(\log(1/\eps))$ qubits per index register, so the $M$ registers consume $Mp = \bigO(M\log(1/\eps))$ ancilla qubits, a precision-driven overhead independent of the failure probability $\delta$.

\subsubsection{The adaptive version algorithm}
The adaptive gradient estimation algorithm of Ref.~\cite{wadaHeisenbergLimitedAdaptiveGradient2025} removes both overheads at once.
Determining each gradient component one binary digit at a time over $\ceil{\log_2(1/\eps)}$ iterative rounds, with only $p = 3$ index qubits per coordinate, brings the space down to $\bigO(M)$ ancilla qubits, while a geometric schedule of per-round failure probabilities yields the clean Heisenberg query count $\bigO(\eps^{-1}\sqrt{M}\log M)$ with no $\log(1/\eps)$ factor (\cref{apd:cor:adaptive_snapshots}).
The algorithm maintains temporal estimates $\tilde{u}_j^{(q)} \in [-1,1]$ (initialized to $\tilde{u}_j^{(0)} = 0$) and refines them at each round $q$: the residuals $\bra{\psi_0} O_j(t_j) \ket{\psi_0} - \tilde{u}_j^{(q)}$, which are $\bigO(1/2^q)$ by induction, are amplified by a factor of $2^q$ via $\bigO(2^q \sqrt{M})$ queries to the probability oracle, bringing them back to $\bigO(1)$ where the $3$-qubit index registers can resolve them.
After measurement and an inverse QFT, the estimates are updated as $\tilde{u}_j^{(q+1)} = \tilde{u}_j^{(q)} + \pi\, 2^{-q} g_j^{(q)}$.

\begin{corollary}[Adaptive snapshots estimation]\label{apd:cor:adaptive_snapshots}
    Given an $n$-qubit Hamiltonian $\ham$, an initial state $\ket{\psi_0}$ prepared by a unitary $U_\psi$,
    $M$ time points $0 < t_1 < \cdots < t_M \le \tmax$, and bounded observables $O_1, \ldots, O_M$ with $\|O_j\| \le 1$,
    there exists a quantum algorithm that outputs estimators $\hat{o}_j$ satisfying the Heisenberg-limited root mean squared error bound
    \begin{equation}\label{eq:adaptive_mse}
        \max_{j \in [M]}\; \mathbb{E}\bigl[(\hat{o}_j - \bra{\psi_0} O_j(t_j) \ket{\psi_0})^2\bigr] \le \eps^2,
    \end{equation}
    using:
    \begin{itemize}
        \item $\bigO(\eps^{-1}\sqrt{M}\log M)$ queries to the probability oracle $F$ (\cref{thm:gradient_expectation}),
        \item $\bigO(\eps^{-1}\sqrt{M}\log M \cdot \tmax)$ total Hamiltonian simulation time ,
        \item $\bigO(M + n)$ qubits ($3M$ ancilla qubits for $M$ index registers of $p = 3$ qubits each, independent of~$\eps$).
    \end{itemize}
\end{corollary}
\begin{proof}

    \emph{Query complexity.}
    In the $q$th round, gradient estimation to precision $2^{-q}$ uses $\bigO(2^q \sqrt{M} \log(M/\delta^{(q)}))$ queries to $F$, where the $\log(M/\delta^{(q)})$ factor arises from coordinate-wise median amplification over the $M$ gradient components~\cite{wadaHeisenbergLimitedAdaptiveGradient2025}.
    The failure probability is scheduled geometrically as $\delta^{(q)} = c_0 / 8^{q_{\max} - q}$ for a constant $c_0 \in (0, 3/8(1+\pi)^2]$, so that $\log(M/\delta^{(q)}) = \bigO(\log M + (q_{\max} - q))$.
    The total query count is
    \begin{equation}
        \sum_{q=0}^{q_{\max}} \bigO\!\left(2^q \sqrt{M}\,(\log M + (q_{\max} - q))\right)
        = \bigO\!\left(\frac{\sqrt{M}}{\eps}\,\log M\right),
    \end{equation}
    where the geometric series $\sum_{q} 2^q$ is dominated by its last term $\bigO(1/\eps)$, and $\sum_q (q_{\max} - q)\,2^q = \bigO(1/\eps)$ converges (its weight concentrates at large $q$, where $q_{\max} - q = \bigO(1)$), so no $\log(1/\eps)$ correction arises, recovering the clean Heisenberg scaling of Ref.~\cite[Theorem~1]{wadaHeisenbergLimitedAdaptiveGradient2025}.

    \emph{Space.}
    Each of the $M$ index registers uses $p = 3$ qubits (independent of $\eps$), giving $3M$ ancilla qubits plus $n$ system qubits and $1$ Hadamard-test ancilla, for a total of $\bigO(M + n)$.

    \emph{MSE guarantee.}
    Because the MSE averages over all measurement outcomes, the geometric failure schedule of the query-complexity analysis is what secures the bound:
    early rounds, which fix the most significant bits and hence dominate the MSE, receive exponentially smaller failure budgets, while later rounds, whose errors contribute only $\bigO(2^{-2q})$, tolerate higher failure rates.
    By the probability-tree analysis of Ref.~\cite[Theorem~1]{wadaHeisenbergLimitedAdaptiveGradient2025}, the scheduled $\delta^{(q)}$ ensure that the root MSE of each estimator $\hat{o}_j$ is at most $\eps$, yielding \cref{eq:adaptive_mse} within the same $\bigO(\eps^{-1}\sqrt{M}\log M)$ query budget.
\end{proof}

\subsubsection{The algorithm for Search-DQPT}

We now assemble the snapshot estimator (\cref{apd:cor:adaptive_snapshots}) and the commutator derivative estimator (\cref{apd:lem:comm_complexity}) into a two-phase algorithm for \nameref{apd:prm:search_dqpt}, establishing the complexity claimed in the main text.
Throughout we take the subsystem size $k=\bigO(1)$, the regime in which \nameref{apd:prm:promise_local_dqpt} is \nameref{def:bqp_complete}.
Two specializations of the snapshot engine drive the search.
Setting every observable to the $k$-local projector, $O_j = \projk$, the estimator returns the subsystem echoes $\los_k(t_j) = \bra{\psi_0}\projk(t_j)\ket{\psi_0}$,
hence the magnitude $\krate(t_j) = -\tfrac1k\log\los_k(t_j)$ controlling the condition $\krate(t)\ge\xi_1$.
Setting instead $O_j = Q/\norm{Q}$ with $Q = \ii[\ham,\projk]$ (\cref{eq:comm_identity}), the same engine returns the echo derivatives $\dot{\los}_k(t_j) = \bra{\psi_0}Q(t_j)\ket{\psi_0}$,
and hence $\dkrate(t_j) = -\dot{\los}_k(t_j)/(k\,\los_k(t_j))$ wherever the echo is bounded below.
The rescaling that meets the engine's normalization $\norm{O_j}\le1$ is by $\norm{Q}\le\norm{Q}_1 = \bigO(kJ_{\max})$ (\cref{eq:comm_1norm}), an $n$-independent constant for constant $k$ and bounded couplings $J_{\max}$---not the extensive $2\norm{\ham}=\poly(n)$---so the derivative sweep costs the same as the echo sweep up to a constant factor.
Both specializations leave the oracle construction of \cref{apd:cor:adaptive_snapshots} intact, adding only $\bigO(kM)$ controlled rotations per query.
The screen tests the integrated susceptibility $J_k(t) = \dkrate(t-\deltat) - \dkrate(t+\deltat)$ of \nameref{apd:prm:promise_local_dqpt} (\cref{apd:eq:jump_def}), assembled from two derivative estimates at the fixed window offset $\deltat$.

\begin{corollary}[Complexity of \nameref{apd:prm:search_dqpt}]\label{apd:cor:search_dqpt}
    Let $\ham$ be an $n$-qubit local Hamiltonian with bounded local couplings, $\ket{\psi_0}$ an initial state, $\tmax\in\poly(n)$ a time range, $k=\bigO(1)$ the subsystem size, and $\eps>0$ a precision.
    Assume the critical times of \nameref{apd:prm:promise_local_dqpt} in $[0,\tmax]$ are separated by at least $\deltat_{\min}\in\Omega(1/\poly(n))$, with window offset $\deltat=\Theta(\deltat_{\min})$,
    and the rate function uniformly bounded, $\krate(t)\le R_0$ on $[0,\tmax]$ for a constant $R_0$ (equivalently a subsystem echo floor $\los_k(t)\ge e^{-kR_0}=\Omega(1)$, extending the well-conditioned-probe clause of \nameref{apd:prm:promise_local_dqpt} from the offsets to the cell interior).
    Then there is a quantum algorithm that outputs estimates $\tilde{t}_c$ with $\abs{\tilde{t}_c - t_c}\le\eps$ for every critical time $t_c\in[0,\tmax]$, using total Hamiltonian simulation time
    \begin{equation}\label{apd:eq:search_cost}
        \widetilde{\bigO}\!\left(\frac{\tmax^{3/2}}{\deltat_{\min}^{1/2}}\right) + \widetilde{\bigO}\!\left(N_c\,\tmax\log\frac{\deltat_{\min}}{\eps}\right)
    \end{equation}
    and $\bigO(\tmax/\deltat_{\min} + n)$ qubits, where $N_c$ denotes the number of critical times in $[0,\tmax]$.
    For $N_c = \bigO(1)$ the screening term dominates and the cost reduces to the form stated in the main text.
\end{corollary}
\begin{proof}
    The algorithm screens a uniform grid for susceptibility spikes at constant precision, then bisects each flagged cell to precision $\eps$.

    \emph{Grid and promise structure.}
    Place a uniform grid $\qty{t_j = j h}_{j=0}^{M_0}$ on $[0,\tmax]$ with spacing $h := \deltat_{\min}/3$, so $M_0 = 3\tmax/\deltat_{\min} = \Theta(\tmax/\deltat_{\min})$.
    By the separation promise, consecutive critical times differ by at least $\deltat_{\min} = 3h$, so each critical time occupies a distinct cell with at least one empty cell between any two.
    Choosing the window offset $\deltat$ to be an integer multiple of $h$ (consistent with $\deltat=\Theta(\deltat_{\min})$) aligns the probe points $t_j\pm\deltat$ with the grid, so one derivative sweep supplies every $\dkrate(t_j)$ and the susceptibility $J_k(t_j) = \dkrate(t_j-\deltat) - \dkrate(t_j+\deltat)$ (\cref{apd:eq:jump_def}) is read off by combining grid values.
    At a critical time $t$ the \yes{} conditions $\krate(t)\ge\xi_1$ and $J_k(t)\ge J_1$ hold (deep echo suppression and a susceptibility spike);
    away from any critical time the \no{} disjunction $\krate(t)\le\xi_0$ or $J_k(t)\le J_0$ holds, with gaps $\xi_1-\xi_0=\Omega(1)$ and $J_1/J_0\ge c>1$.

    \emph{Step 1. Screening.}
    Apply \cref{apd:cor:adaptive_snapshots} once with $O_j=\projk$ to estimate every $\krate(t_j)$, and once with $O_j = Q/\norm{Q}$ to estimate every $\dkrate(t_j)$, each to a fixed constant precision $\eps_0=\Theta(1)$ set below the promise gaps.
    The conversion $\dkrate=-\dot{\los}_k/(k\los_k)$ at the probes is well-conditioned: the well-conditioned-probe promise $\krate(t\pm\deltat)\le R_0$ of \nameref{apd:prm:promise_local_dqpt} gives an echo floor $\los_k(t\pm\deltat)\ge e^{-kR_0}=\Omega(1)$ for constant $k$, so constant additive precision on $\dot{\los}_k$ yields constant additive precision on each $\dkrate(t_j\pm\deltat)$, and hence on $J_k(t_j)$.
    Flag a cell when its magnitude estimate clears $\xi_1$ and its susceptibility estimate clears $J_1$.
    Because the gaps are $\Theta(1)$, constant precision $\eps_0$ with $\bigO(\log M_0)$ median boosting separates \yes{} from \no{} at every grid point with high probability by a union bound over the $M_0=\poly(n)$ cells, so the flagged set is exactly the $N_c$ cells containing the critical times.
    By \cref{apd:cor:adaptive_snapshots} at $M=M_0$ and constant precision, each sweep uses $\widetilde{\bigO}(\sqrt{M_0})$ oracle queries and total simulation time $\widetilde{\bigO}(\sqrt{M_0}\,\tmax) = \widetilde{\bigO}(\tmax^{3/2}/\deltat_{\min}^{1/2})$, with $\bigO(M_0+n)=\bigO(\tmax/\deltat_{\min}+n)$ qubits; the boosting and union bound are absorbed into $\widetilde{\bigO}$.

    \emph{Step 2: bisection.}
    Each flagged cell brackets a single critical time, at which $\krate$ has a kink and $\dkrate$ jumps from positive (rising suppression) to negative.
    Bisecting on the sign of $\dkrate$ at the bracket midpoint localizes this kink to precision $\eps$ in $\ceil{\log_2(h/\eps)} = \bigO(\log(\deltat_{\min}/\eps))$ rounds, keeping the half on which the sign persists.
    Each round is a single-point estimate of $\dkrate$, costing $\bigO(\eps_0^{-1}\tmax)=\bigO(\tmax)$ simulation time and $\bigO(1+n)$ qubits.
    The sign test stays well-conditioned because the slope jump $J_1$ keeps $\abs{\dkrate}=\Omega(1)$ with a definite sign up to the kink, while the uniform echo floor $\los_k\ge e^{-kR_0}=\Omega(1)$ lets a single-point estimate of $\dot{\los}_k = \bra{\psi_0}Q(t)\ket{\psi_0}$ resolve that sign at constant precision, even as the midpoints approach the suppressed-echo centre $t_c$.
    Over all $N_c$ cells, with the per-round success boosted to clear a union bound over the $\bigO(N_c\log(\deltat_{\min}/\eps))$ sign tests, this costs $\widetilde{\bigO}(N_c\,\tmax\log(\deltat_{\min}/\eps))$.

    \emph{Total cost.}
    Summing the two phases gives \cref{apd:eq:search_cost}.
    The screening term dominates whenever $N_c\log(\deltat_{\min}/\eps) = \widetilde{\bigO}(\sqrt{\tmax/\deltat_{\min}})$, in particular for the generic case of $\bigO(1)$ critical times, recovering the main-text bound; the qubit count is fixed by the screening phase.
    Finally, each probability-oracle query embeds a time evolution to time at most $\tmax$, which by the no-fast-forwarding theorem (\cref{apd:thm:no_fast_forwarding}, with locality degree $d=\bigO(1)$) costs $\Omega(\tmax)$ Hamiltonian-simulation time, so the linear $\tmax$ factor in each phase is unavoidable.
\end{proof}

\section{Optimality and quantum speedups}\label{apd:lower_bounds}
Given the quantum algorithm (\cref{apd:cor:adaptive_snapshots}) for \nameref{apd:prm:multi_time_obs} presented in the previous section,
we show its optimality in \cref{apd:sec:quantum_lower_bounds} and its speedup over classical algorithms in \cref{apd:sec:classical_lower_bound} by proving the corresponding lower bounds.

\subsection{Optimality: quantum lower bounds}\label{apd:sec:quantum_lower_bounds}

We first show the optimality of our quantum algorithm (the tightness of upper bound \cref{apd:cor:adaptive_snapshots}).
We first recall the relevant lower bounds for multi-dimensional amplitude estimation (\cref{thm:query_lower_bound_multi_amp_est}) and multi-observable estimation (\cref{thm:query_lower_bound}), 
then derive the constructive $\bigO(1)$-local-Hamiltonian, single-observable lower bound via Feynman--Kitaev reduction (\cref{apd:thm:local_lower_bound}).

\subsubsection{Lower bound for multi-observable estimation}
\begin{lemma}[$\ell_\infty$ lower bound for multidimensional amplitude estimation {\cite{vanapeldoornQuantumProbabilityOracles2021}}]\label{thm:query_lower_bound_multi_amp_est}
    Let $\eps\in(0,1/(3\sqrt{M}))$ and let $M$ be a positive integer power of two.
    There exists a matrix $A\in\qty{-1,1}^{M\times M}$
    such that any algorithm that for every $\vbp\in\Delta^M$ (with success probability at least $2/3$) outputs a $\tilde{\vbq}\in\Delta^M$
    for which $\norm{A\vbp-\tilde{\vbq}}_{\infty}\le \eps$, uses at least $\Omega(\sqrt{M}/\eps)$ queries to a quantum probability oracle for $\vbp$.
\end{lemma}

Here $\Delta^M \coloneqq \qty{\vbp \in \mathbb{R}^M : p_j \ge 0,\ \norm{\vbp}_1 = 1}$ denotes the set of all probability distributions on $M$ outcomes.
The bound counts coherent queries to $U_p$ (including $U_p^\dagger$ and controlled-$U_p$), so it applies to any quantum algorithm using this access; the classical analogue, drawing i.i.d.\ samples from $p$, admits only a trivial $\tilde{O}(M/\eps^2)$ upper bound via Chernoff and a union bound~\cite{vanapeldoornQuantumProbabilityOracles2021}.
The $l_\infty$ condition $\norm{A\vbp - \tilde{\vbq}}_\infty \le \eps$ requires \emph{every} component $(A\vbp)_j$ to be estimated within~$\eps$, 
precisely the per-observable guarantee demanded by multi-observable estimation, implying the following lower bound.

\begin{lemma}[Lower bound for multi-observable estimation {\cite{hugginsNearlyOptimalQuantum2022}}]\label{thm:query_lower_bound}
    Let $M$ be a positive integer power of $2$ and let $\eps\in(0,1/(3\sqrt{M}))$.
    Let $\mathcal{A}$ be any algorithm that takes as input an arbitrary set of $M$ observables $\qty{O_j}$.
    Suppose that, for every quantum state $\ket{\psi}$, accessed via state preparation oracle $U_{\psi}$,
    $\mathcal{A}$ outputs estimates of each $\bra{\psi} O_j \ket{\psi}$ to within additive error $\eps$ (with probability at least $2/3$).
    Then, there exists a set of observables $\qty{O_j}$ such that $\mathcal{A}$ applied to $\qty{O_j}$ must use $\Omega(\sqrt{M}/\eps)$ queries to $U_{\psi}$.
\end{lemma}
\begin{proof}[Proof of {\cite[Theorem~5]{hugginsNearlyOptimalQuantum2022}}]
    Suppose for contradiction that for every set $\{O_j\}$ and every $U_\psi$,
    algorithm $\mathcal{A}$ uses $o(\sqrt{M}/\eps)$ queries.
    Recall that in \cref{thm:query_lower_bound_multi_amp_est},
    the quantum probability oracle for $\vbp \in \Delta^M$ is a unitary $U_p$
    satisfying $U_p\ket{0} = \sum_{j=1}^{M} \sqrt{p_j}\,\ket{j}\ket{\phi_j}$
    for some states $\ket{\phi_j}$.
    Let $A \in \{-1,+1\}^{M \times M}$ be the hard matrix from
    \cref{thm:query_lower_bound_multi_amp_est};
    it serves as a fixed, known reduction gadget that converts amplitude-estimation hardness into hardness for the observable set $\qty{O_j = Z_j}$ with state-preparation oracle $U_\psi$.
    For any probability oracle $U_p$, define the state
    \begin{equation}\label{eq:huggins_state}
        \ket{\psi(U_p)}
        := U_A\,(I \otimes U_p)\ket{\vb{0}},
    \end{equation}
    where $U_A = \sum_{j=1}^{M}
    \bigl(\bigotimes_{i=1}^{M} X_i^{\,\delta_{A_{ij},\,-1}}\bigr)
    \otimes \ketbra{j}{j} \otimes I$
    flips qubit~$i$ whenever $A_{ij} = -1$.
    A direct computation gives
    \begin{equation}
        \bra{\psi(U_p)} Z_i \ket{\psi(U_p)}
        = \sum_{j} p_j\, A_{ij} = (Ap)_i.
    \end{equation}
    Setting $O_j = Z_j$ and $U_\psi = U_A\,(I \otimes U_p)$,
    algorithm $\mathcal{A}$ estimates every $(A\vbp)_j$ to error $\eps$
    using $o(\sqrt{M}/\eps)$ queries to $U_\psi$,
    hence to $U_p$ (since $U_A$ is known).
    This contradicts \cref{thm:query_lower_bound_multi_amp_est}.
    Since the argument quantifies over any $\mathcal{A}$ with coherent $U_\psi$-access, the bound also applies to the non-coherent (classical prepare-and-measure) algorithms, cf. \cref{def:pm_model}.
\end{proof}

If we allow the observables in \nameref{apd:prm:multi_time_obs} to be different at each time point, the lower bound of \cref{thm:query_lower_bound} applies directly, since estimating $\bra{\psi_0} O_j(t_j) \ket{\psi_0}$ for $M$ different observables $O_j$ is a special case of multi-observable estimation.
We now establish the lower bound for estimating a single \emph{fixed observable} at multiple time points of Hamiltonian dynamics, 
which is the common scenario in quantum experiments. 

\begin{lemma}[Lower bound for single-observable multi-time snapshots estimation]
    \label{apd:lem:multi_quantum_lower_bound}
    \label{apd:thm:multi_lower_bound}
    For $\eps \in (0, \eps_0/\sqrt{M})$ with $\eps_0 > 0$ a sufficiently small absolute constant, there exist an $n$-qubit
    Hamiltonian $\ham$, a single $1$-local projector $P$, and $M$ time points
    $\{t_j\}$ such that any quantum algorithm estimating the snapshot expectations
    $\bra{\psi_0}\, e^{\ii \ham t_j}\, P\, e^{-\ii \ham t_j}\, \ket{\psi_0}$
    at all $M$ time points to additive error $\eps$ requires
    $\Omega(\sqrt{M}/\eps)$ queries to $U_\psi$ in the worst case
    over the (oracle-prepared) initial state $\ket{\psi_0}$.
\end{lemma}
\begin{proof}
    We reduce from \cref{thm:query_lower_bound}.
    By that lemma, for $O_j^* = Z_j$ (Pauli-$Z$ on qubit~$j$, $j \in [M]$),
    estimating $\bra{\psi}Z_j\ket{\psi}$ for all $j$ to additive error $\eps$
    requires $\Omega(\sqrt{M}/\eps)$ queries to $U_\psi$ in the worst case.
    Since $Z_j = 2\ketbra{0}{0}_j \otimes I - I$, this is equivalent to estimating
    the projector expectation values $\bra{\psi}\ketbra{0}{0}_j\otimes I\ket{\psi}$
    to precision $\eps/2$.

    \emph{Construction.}
    Take $n = M+1$ qubits and set $P = \ketbra{0}{0}_1 \otimes I$.
    Define the cyclic left-shift operator $\Pi$ on $n$ qubits by
    \begin{equation}\label{eq:cyclic_shift_def}
        \Pi\ket{x_1, x_2, \ldots, x_n} = \ket{x_2, x_3, \ldots, x_n, x_1},
    \end{equation}
    i.e.\ the content of qubit~$j$ moves to position~$j{-}1$ (cyclically).
    Concretely, 
    \begin{equation}
        \Pi:= \mathrm{SWAP}_{n-1,n} \cdots \mathrm{SWAP}_{2,3}\, \mathrm{SWAP}_{1,2}.
    \end{equation}
    Since $\Pi^n = I$, the eigenvalues of $\Pi$ are the $n$-th roots of unity
    $\omega^k = e^{2\pi\ii k/n}$ ($k = 0, \ldots, n{-}1$),
    and the Hermitian generator $\ham = H_\Pi$ satisfying $e^{-\ii H_\Pi} = \Pi$ is
    \begin{equation}\label{eq:cyclic_hamiltonian}
        H_\Pi = -\sum_{k=0}^{n-1} \frac{2\pi k}{n}\, \mathcal{P}_k,
    \end{equation}
    where $\mathcal{P}_k$ projects onto the $\omega^k$-eigenspace of $\Pi$.

    At integer times $t_j = j$ for $j = 1,\ldots,M$,
    the evolution $e^{-\ii H_\Pi j} = \Pi^j$ shifts each qubit label
    by $j$ positions, bringing the content of qubit~$j+1$ to position~$1$.
    Conjugating the projector gives
    \begin{equation}\label{eq:cyclic_shift_projector}
        e^{\ii H_\Pi j}\, P\, e^{-\ii H_\Pi j}
        = \Pi^{-j}\, P\, \Pi^j
        = \ketbra{0}{0}_{j+1} \otimes I,
    \end{equation}
    since $P$ acts on position~$1$, which after the shift $\Pi^j$ holds the original content of qubit~$j+1$.
    Write
    \begin{equation}
        o(t_j) \coloneqq \bra{\psi_0}\, e^{\ii H_\Pi t_j}\, P\, e^{-\ii H_\Pi t_j}\, \ket{\psi_0} = \bra{\psi_0}\, \ketbra{0}{0}_{j+1}\otimes I\, \ket{\psi_0}
    \end{equation}
    for the snapshot expectation of $P$ at time $t_j$.

    \emph{Reduction.}
    Suppose an algorithm estimates $o(t_j)$ for all $j \in [M]$ to additive error $\eps/2$.
    Since $Z_{j+1} = 2\ketbra{0}{0}_{j+1}\otimes I - I$, this yields estimates of $\bra{\psi_0}Z_{j+1}\ket{\psi_0}$ to precision $\eps$.
    Each snapshot evaluation uses one query to $U_\psi$
    (\cref{apd:cor:oracle_circuit}),
    so the snapshot-estimation algorithm directly yields
    a multi-observable estimation algorithm with the same query count.
    By \cref{thm:query_lower_bound}, the number of queries is $\Omega(\sqrt{M}/\eps)$
    (the constant $\eps_0$ in the hypothesis absorbs the factor-$2$ loss from $Z = 2P - I$, keeping the amplified precision within the regime of that lemma).
\end{proof}

\subsubsection{Local Hamiltonian bound via Feynman--Kitaev construction}
The hard instance used above relies on $H_\Pi$, the generator of the cyclic shift on $M+1$ qubits. Although the unitary $\Pi$ itself is locally implementable as a depth-$M$ circuit of nearest-neighbor SWAPs, its continuous-time generator 
$H_{\Pi}$ \cref{eq:cyclic_hamiltonian}
is neither $k$-local nor geometrically local: 
the spectral projectors $\mathcal{P}_k$ act on global momentum eigenspaces. The lower bound of \cref{apd:thm:multi_lower_bound} is therefore an existential statement for the class of general Hamiltonians and does not, on its own, specialize to local ones. We now upgrade it to a constructive $\bigO(1)$-local hard instance via the Feynman--Kitaev construction.

\begin{theorem}
    [Local Hamiltonian version of \cref{apd:lem:multi_quantum_lower_bound}]
    \label{apd:thm:local_lower_bound}
    Let $M$ be a positive integer power of $2$ and let $\eps \in (0, \eps_0/\sqrt{M})$ for a sufficiently small absolute constant $\eps_0 > 0$.
    There exist an $\bigO(1)$-local Hamiltonian $\ham$ on $n = \bigO(M^2)$ qubits, a single observable $O$ (a $1$-local projector with $\norm{O} = 1$), and $M$ time points $\{t_m\}_{m=1}^{M}$ such that any quantum algorithm estimating
    $\bra{\psi_0}\, e^{\ii\ham t_m}\, O\, e^{-\ii\ham t_m}\, \ket{\psi_0}$
    at all $M$ time points to additive error $\eps$ requires $\Omega(\sqrt{M}/\eps)$ queries to $U_\psi$ in the worst case over the (oracle-prepared) initial state $\ket{\psi_0}$.
\end{theorem}
\begin{proof}
    We reduce from \cref{thm:query_lower_bound} via the
    Feynman--Kitaev local Hamiltonian construction (\cref{thm:circuit_hamiltonian}).

    \emph{Circuit construction.}
    Consider $n_o = M+1$ output qubits.
    Recall that the cyclic left-shift
    $\Pi = \mathrm{SWAP}_{M,M{+}1}\cdots\mathrm{SWAP}_{2,3}\,
           \mathrm{SWAP}_{1,2}$
    consists of $M$ nearest-neighbour gates
    (\cref{eq:cyclic_shift_def}).
    Define the circuit $V = \Pi^M$---the $M$-fold composition of
    $\Pi$---consisting of $N = M^2$ nearest-neighbour SWAP gates, 
    labelled $V_1,\dots,V_N$ in application order so that $V = V_N\cdots V_1$; 
    the partial product $V_j\cdots V_1$ denotes the first $j$ gates, and $(V_j\cdots V_1)^\dagger P\,(V_j\cdots V_1)$ is the projector $P$ in the Heisenberg picture after gate step $j$.

    \emph{Feynman--Kitaev embedding.}
    Apply \cref{thm:circuit_hamiltonian} to $V$, obtaining a
    Hamiltonian $H_{\mathrm{FK}}$ on
    $n_o + N = M^2 + M + 1 = \bigO(M^2)$ qubits
    (output register $\oo$ plus clock register $\cc$ in unary
    encoding).
    Each term in $H_{\mathrm{FK}}$ couples two adjacent clock qubits
    to one two-qubit SWAP on the output register, giving
    $\bigO(1)$-local interactions
    (the same locality as in the BQP-hardness proof,
    \cref{apd:thm:bqp_hard}).

    \emph{Initial state and projector.}
    Set $P = I_\cc \otimes \op{0}_{\oo,1} \otimes I$
    (a $1$-local projector on the first output qubit).
    Choose the initial state
    $\ket{\psi_0} = \ket{0}_\cc \otimes \ket{0}_{\oo,1}
    \otimes \ket{\psi'}_{\oo,2:M{+}1}$,
    where $\ket{\psi'}$ is an arbitrary $M$-qubit state
    on qubits $2,\ldots,M{+}1$ prepared by the oracle $U_\psi$.
    Since qubit~$1$ is always $\ket{0}$, we have $p_1 \coloneqq
    \bra{\psi_0}\op{0}_{\oo,1}\ket{\psi_0} = 1$.

    \emph{Block structure.}
    The $N = M^2$ gates of $V = \Pi^M$ partition into $M$ blocks of
    $M$ gates each: block~$q$ ($q = 0,\ldots,M{-}1$) consists of
    gate steps $qM{+}1$ through $(q{+}1)M$.
    At gate step $qM{+}1$, the first gate
    $\mathrm{SWAP}_{1,2}$ of the $(q{+}1)$-th copy of $\Pi$
    brings the content of output qubit~$q{+}2$
    (after $q$ complete shifts) to position~$1$;
    the remaining $M{-}1$ gates
    $\mathrm{SWAP}_{2,3},\ldots,\mathrm{SWAP}_{M,M{+}1}$
    do not touch position~$1$.
    Therefore, for every gate step $j$ in block~$q$,
    \begin{equation}
        (V_j\cdots V_1)^\dagger\, P\, (V_j\cdots V_1) = I_\cc \otimes
        \op{0}_{\oo,q{+}2} \otimes I.
    \end{equation}

    \emph{Snapshot decomposition.}
    Write $o(t) \coloneqq \bra{\psi_0}\, e^{\ii H_{\mathrm{FK}} t}\, P\, e^{-\ii H_{\mathrm{FK}} t}\, \ket{\psi_0}$ for the snapshot expectation of $P$ at time~$t$.
    By the clock-amplitude formula (\cref{eq:Lk_clock,eq:clock_amplitudes}, specialized to the initial state $\ket{\psi_0}$ above),
    \begin{equation}\label{eq:local_echo_FK}
        o(t) = \abs{c_0(t)}^2\, p_1
        + \sum_{q=0}^{M-1} w_q(t)\, p_{q+2},
    \end{equation}
    where $w_q(t) \coloneqq \sum_{j=qM+1}^{(q+1)M} \abs{c_j(t)}^2$ is the total clock weight on block~$q$, and $p_{q+2} = \bra{\psi'}\op{0}_{\oo,q+2}\ket{\psi'}$ is the projector expectation on output qubit~$q{+}2$.
    Since $p_1 = 1$ and $\abs{c_0(t)}^2$ depends only on $M$ and $t$ (not on $\ket{\psi'}$), it can be subtracted:
    \begin{equation}\label{eq:local_echo_FK_reduced}
        \tilde{o}(t) \coloneqq o(t) - \abs{c_0(t)}^2 = \sum_{q=0}^{M-1} w_q(t)\, p_{q+2}.
    \end{equation}

    \emph{Diagonal dominance.}
    Choose time points
    $\sin^2(t_m) = (2m{-}1)/(2M)$ for $m = 1,\ldots,M$,
    so that the binomial mean
    $\mu_m = N\sin^2(t_m) = (m{-}\tfrac{1}{2})M$
    lies at the centre of block~$m{-}1$.
    The standard deviation satisfies
    $\sigma_m = \sqrt{N\sin^2(t_m)\cos^2(t_m)} \le M/2$,
    so block~$m{-}1$ (width~$M$) spans at least
    $M/\sigma_m \ge 2$ standard deviations.
    Hence the diagonal entry of the $M\times M$ weight matrix $W_{mq} \coloneqq w_q(t_m)$,
    the mass a $\bpf(N, \sin^2 t_m)$ variable places on block~$m{-}1$,
    is controlled by the Gaussian approximation.
    The binding (most concentrated) case is $\sin^2 t_m = \tfrac{1}{2}$ near $m = M/2$,
    where $\sigma_m = M/2$ and the block is exactly $[\mu_m - \sigma_m,\, \mu_m + \sigma_m]$;
    the edge points $m = 1, M$ have $\sigma_m = \Theta(\sqrt{M})$ and span $\Theta(\sqrt{M})$ standard deviations.
    By the Berry--Esseen theorem the binomial cumulative distribution deviates from the Gaussian by $\bigO(1/\sigma_m)$,
    which at this case is $\bigO(1/M)$, so
    \begin{equation}\label{eq:block_weight_lower}
        W_{m,m-1} = w_{m-1}(t_m) \ge \Pr\!\qty(\abs{\mathcal{N}(0,1)} \le 1) - \bigO(1/M) = 0.682\ldots - o(1) > \tfrac{1}{2}
    \end{equation}
    for all $m \in [M]$ once $M$ exceeds an absolute constant,
    the finitely many smaller $M$ verified by direct evaluation of the exact binomial sum $W_{m,m-1} = \sum_{k=(m-1)M}^{mM-1} \binom{N}{k}(\sin^2 t_m)^k(\cos^2 t_m)^{N-k} > \tfrac12$, with no Gaussian approximation.
    In particular $W_{m,m-1}$ is bounded away from $\tfrac{1}{2}$ uniformly in $M$.
    Since $\sum_{q} W_{mq} = 1 - \abs{c_0(t_m)}^2 \le 1$,
    the off-diagonal row sum is
    $\sum_{q \ne m-1} W_{mq} < 1 - W_{m,m-1} < \tfrac{1}{2}$,
    establishing strict diagonal dominance.
    By Varah's bound for strictly row-diagonally-dominant matrices~\cite{varahLowerBoundSmallest1975},
    which gives $\norm{A^{-1}}_\infty \le 1/\min_i\bigl(\abs{a_{ii}} - \sum_{j\ne i}\abs{a_{ij}}\bigr)$ for any such $A$,
    the row margins above yield
    \begin{equation}
        \norm{W^{-1}}_\infty \le
        1/\min_m(2W_{m,m-1} - 1 + \abs{c_0(t_m)}^2) = \bigO(1),
    \end{equation}
    where the denominator is $\Omega(1)$ by the uniform margin just established.

    \emph{Reduction.}
    An algorithm estimating $o(t_m)$ for all $m$ to precision~$\eps$ also estimates $\tilde{o}(t_m)$ to precision~$\eps$ (subtracting the known $\abs{c_0(t_m)}^2$).
    Stacking \cref{eq:local_echo_FK_reduced} over the $M$ time points gives the linear system $\tilde{\vb{o}} = W \vb{p}$ with $\vb{p} = (p_2, \ldots, p_{M+1})$,
    so $W^{-1}$ recovers the projector expectations and the $\ell_\infty$ error amplification is exactly $\norm{W^{-1}}_\infty$;
    applying it yields estimates of $p_2, \ldots, p_{M+1}$ to precision $\norm{W^{-1}}_\infty \eps = \bigO(\eps)$.
    Since $Z_{q+2} = 2\op{0}_{\oo,q+2} - I$, this gives
    estimates of $\bra{\psi'}Z_j\ket{\psi'}$ for $j = 1,\ldots,M$
    to precision $\bigO(\eps)$.
    The reduction is query-preserving:
    $U_\psi$ is built from the oracle $U_{\psi'}$ for $\ket{\psi'}$ with $\bigO(1)$ fixed-state preparations,
    and the subtraction and $W^{-1}$ are classical post-processing,
    so any algorithm estimating the $M$ snapshots yields an $M$-observable estimation algorithm with the same $U_\psi$-query count.
    Choosing $\eps_0$ small enough that the amplified precision $2\norm{W^{-1}}_\infty \eps$ lies in $(0, 1/(3\sqrt{M}))$ (possible since $\norm{W^{-1}}_\infty = \bigO(1)$),
    \cref{thm:query_lower_bound} gives $\Omega(\sqrt{M}/\eps)$.

\end{proof}

\begin{remark}[Gaps to lower bound of \nameref{apd:prm:search_dqpt}]\label{rem:search_lower_bound}
    The lower bounds of this subsection target snapshot estimation, not \nameref{apd:prm:search_dqpt} directly. Translating to a Search-DQPT bound leaves three distinct gaps:
    \begin{itemize}
        \item \emph{Search vs.\ estimation.} The bound constrains Search-DQPT algorithms that explicitly estimate $\los_k$ at a grid of $M_0 = \tmax/\deltat_{\min}$ time points (the screening strategy of the optimal algorithm). 
        Algorithms that bypass full grid screening---adaptive strategies, gradient search on $r_k$, or other tricks---are not directly constrained. 

        \item \emph{Observable.} The hard instance uses a fixed $1$-local projector, not the local Loschmidt echo $\los_k$; the echo's target $\op{\psi_0^{(k)}}$ depends on the unknown $\ket{\psi_0}$ and is not directly measurable under oracle access to the initial state.

        \item \emph{Precision regime.} The rigorous bound $\Omega(\sqrt{M}/\eps)$ requires $\eps < 1/(3\sqrt{M})$ (the high-precision regime of Apeldoorn--Huggins~\cite{vanapeldoornQuantumProbabilityOracles2021,hugginsNearlyOptimalQuantum2022}). 
        The screening regime of Search-DQPT operates at constant $\eps$, outside this regime.
    \end{itemize}
\end{remark}

\subsection{Speedups: classical lower bounds}\label{apd:sec:classical_lower_bound}
The previous subsection established the (nearly) optimality of our quantum algorithm for multi-time snapshot estimation.
A natural question is whether it brings a genuine speedup over \emph{classical} algorithms.
We answer affirmatively in the prepare-and-measure model (\cref{def:pm_model}), the natural non-coherent classical access model:
any generic classical algorithm estimating the $M$ snapshots requires $\Omega(M/\eps^2)$ rounds (\cref{thm:joint_classical_lower_bound}),
quadratically more, in \emph{both} $M$ and $1/\eps$, than the $\widetilde{\bigO}(\sqrt{M}/\eps)$ rounds of the quantum algorithm (\cref{apd:cor:adaptive_snapshots}). 
Both bounds count the same resource, one call to the state-preparation oracle $U_\psi$ per round.
The bound reflects two compounding costs: the $\Omega(1/\eps^2)$ shot-noise cost of resolving a single near-balanced snapshot, and an extra factor $M$ for the $M$ distinct time points.
We capture both with a multi-parameter estimation lemma (\cref{lem:classical_multi_lower_bound}) applied to a continuous near-$\tfrac12$ encoding embedded in an $\bigO(1)$-local Hamiltonian.
Structurally, both the quantum and classical lower bounds couple the $M$ unknown marginals through a single $M\times M$ mixing matrix (the Hadamard matrix $A$ of \cref{thm:query_lower_bound_multi_amp_est} in the quantum black-box model, the clock-weight matrix $W$ below in the classical one), so the quadratic separation is the price of resolving that matrix's action coherently, $\widetilde{\bigO}(\sqrt{M}/\eps)$, versus one classical bit at a time, $\Omega(M/\eps^2)$.

\subsubsection{Classical prepare-and-measure model and Bernoulli lower bound}

\begin{definition}[Classical prepare-and-measure model]\label{def:pm_model}
    A \emph{classical prepare-and-measure} (PM) algorithm, with oracle access to a state-preparation unitary $U_\psi$ ($U_\psi\ket{\vb{0}} = \ket{\psi_0}$) and to evolution under a Hamiltonian $\ham$, proceeds in rounds.
    Each round
    (i) prepares one copy of $\ket{\psi_0}$ with a single call to $U_\psi$,
    (ii) evolves it under $e^{-\ii\ham t}$ for a time $t \ge 0$ of the algorithm's choosing, and
    (iii) measures the fixed projector $P$, recording the one outcome bit.
    The choice of $t$ and the decision to halt may depend adaptively on previous outcomes and on internal randomness.
    The cost is the number of rounds, equal to the number of $U_\psi$ calls.
\end{definition}

A PM round at time $t$ thus returns a single sample from $\Ber\big(\bra{\psi_0} e^{\ii\ham t}\, P\, e^{-\ii\ham t}\ket{\psi_0}\big)$.
The PM class is a strict subclass of the coherent quantum algorithms of the previous subsection: it forbids coherence across rounds and restricts each round to one measurement of the fixed observable $P$, the operationally natural protocol for tracking $\langle P\rangle$ along the dynamics.
This single-projector restriction is what the lower bound below exploits.

\begin{lemma}[Classical Bernoulli sample-complexity lower bound {\cite{dagumOptimalAlgorithmMonte2000}}]\label{lem:bernoulli_lower_bound}
    Let $p \in (0,1)$ and $\eps \in (0, p(1-p))$. Any algorithm that estimates the mean of a $\mathrm{Bernoulli}(p)$ distribution to additive precision $\eps$ with success probability at least $2/3$ from classical i.i.d.\ samples requires $\Omega(p(1-p)/\eps^2)$ samples; for $p$ bounded away from $\qty{0,1}$, this is $\Omega(1/\eps^2)$, matching the Chernoff--Hoeffding upper bound of \cref{lem:analog_chernoff}.
\end{lemma}

\subsubsection{Reduction from multi-parameter Bernoulli estimation}
This scalar bound is a warm-up: 
it fixes the per-coordinate (time point) cost but is not invoked directly below, 
since a single PM round does not measure one snapshot in isolation.
The multi-time task compounds this scalar cost with that of resolving $M$ distinct time points.
The obstruction is statistical and has two parts.
First, each projector measurement returns one bit of a \emph{mixture} of the $M$ snapshot values and cannot resolve them in parallel; 
this is the source of the factor $M$.
Second, a \emph{continuous} near-$\tfrac12$ encoding keeps every snapshot genuinely fractional, with competing hypotheses separated by only $\Theta(\eps)$, so each coordinate must be resolved to full precision $\eps$; this is the source of the $1/\eps^2$.
Estimating $M$ such near-$\tfrac12$ coordinates through single-bit PM rounds is an \emph{$M$-parameter testing problem}.
\emph{Assouad's lemma reduces it to $M$ simultaneous two-point tests and lower-bounds the total error by their difficulty, measured by the total variation between half-cube mixtures, converting per-coordinate hardness into a clean $M\times(\text{per-coordinate})$ bound}; 
we use the following averaged form.

\begin{lemma}[Assouad's lemma, averaged form~{\cite[Lemma~2.12]{tsybakovIntroductionNonparametricEstimation2008}}]\label{lem:assouad}
    Let $\{\mathbb{P}_v\}_{v \in \{-1,+1\}^M}$ be probability measures on a common space, 
    and for $i \in [M]$ let $\mathbb{P}_{\pm i} \coloneqq 2^{-(M-1)}\!\sum_{v:\, v_i = \pm 1}\mathbb{P}_v$ denote the two half-cube mixtures.
    Every estimator $\hat v$ of $v$ from a single draw under $\mathbb{P}_v$ obeys, 
    \begin{equation}\label{eq:assouad}
        2^{-M}\sum_{v}\mathbb{E}\,d_{\mathrm H}(\hat v, v) \;\ge\; \frac{M}{2} - \frac12\sum_{i=1}^{M}\TV(\mathbb{P}_{+i}, \mathbb{P}_{-i}).
    \end{equation}
    where $d_{\mathrm H}(\hat v, v) = \sum_{i=1}^M \mathbf{1}[\hat v_i \neq v_i]$ is the Hamming distance and $\TV(P, Q) = \sup_{A}\abs{P(A) - Q(A)}$ is the total variation distance (equivalently $\tfrac12\sum_\omega\abs{P(\omega) - Q(\omega)}$ for discrete outcomes).
\end{lemma}
\begin{proof}
    Each $\mathbb{P}_v$ is a distribution on the same outcome space---in the application below, the space of query--outcome transcripts produced by the algorithm---so every total variation $\TV(\mathbb{P}_{+i}, \mathbb{P}_{-i})$ compares two distributions over identical outcomes and is well defined.
    Recovering $v$ means deciding each of the $M$ signs $v_i = \pm1$,
    and the Hamming distance $d_{\mathrm H}(\hat v, v)$ counts the wrong decisions,
    so the left-hand side of \cref{eq:assouad} is the expected number of mistaken coordinates under a uniform prior over the hypercube.
    The half-cube mixture $\mathbb{P}_{\pm i}$ averages over all configurations of the other $M-1$ coordinates,
    isolating coordinate $i$ and reducing the $M$-way estimation to $M$ two-point tests, one per sign.
    The left-hand side equals $\sum_{i=1}^M 2^{-M}\sum_v \Pr(\hat v_i \neq v_i)$;
    for each $i$, the inner average is 
    the error probability (under the uniform prior) of the test $\hat v_i$ distinguishing $\mathbb{P}_{+i}$ from $\mathbb{P}_{-i}$, 
    hence at least the Bayes error $\tfrac12\bigl(1 - \TV(\mathbb{P}_{+i}, \mathbb{P}_{-i})\bigr)$.
    Summing over $i$ gives \cref{eq:assouad}.
    The bound quantifies the intuition:
    one beats random guessing on coordinate $i$ only insofar as its two mixtures are statistically distinguishable,
    so the average risk falls below $M/2$ by exactly $\tfrac12\sum_i \TV(\mathbb{P}_{+i}, \mathbb{P}_{-i})$;
    forcing the risk to zero requires every $\TV(\mathbb{P}_{+i}, \mathbb{P}_{-i})$ close to $1$, a cost the query budget must pay coordinate by coordinate.
\end{proof}

Applying it to a hypercube of near-$\tfrac12$ instances yields the $M$-coordinate generalization of \cref{lem:bernoulli_lower_bound}.
In the Hamiltonian reduction below, 
the weight vector $\vb{a}$ will be the clock weights $w_q(t)$ (selected by the algorithm's choice of evolution time $t$) and $\vb{p}$ the projector marginals to be estimated, 
while the offset $1 - \norm{\vb{a}}_1$ is the fixed, 
$\vb{p}$-independent leakage $\abs{c_0(t)}^2$. 
We state the lemma abstractly for any such single-bit oracle.

\begin{lemma}[Classical lower bound for multi-parameter Bernoulli estimation]\label{lem:classical_multi_lower_bound}
    Let $M$ be a positive integer and $\alpha \in (0, 1/8)$.
    An unknown vector $\vb{p} \in [\tfrac12 - 2\alpha,\, \tfrac12 + 2\alpha]^M$ is to be estimated from a \emph{single-bit oracle} 
    whose every query chooses a weight vector $\vb{a} \in \mathbb{R}_{\ge 0}^M$ with $\norm{\vb{a}}_1 \le 1$ and returns one bit $b \sim \Ber\big(1 - \vb{a}^\top(\vb{1} - \vb{p})\big)$, 
    equivalently $\Ber\big(\vb{a}^\top \vb{p} + (1 - \norm{\vb{a}}_1)\big)$.
    The choice of $\vb{a}$ may depend adaptively on previous outcomes and on internal randomness.
    Any algorithm that, for every such $\vb{p}$, outputs $\hat{\vb{p}}$ with $\norm{\hat{\vb{p}} - \vb{p}}_\infty \le \alpha$ with probability at least $2/3$ requires $\Omega(M/\alpha^2)$ queries.
\end{lemma}
\begin{proof}
    We apply Assouad's lemma to a hypercube of hard instances.
    It suffices to rule out algorithms that use at most $T$ queries: 
    if such an algorithm halts earlier, pad its transcript with deterministic dummy queries $\vb{a}_t=\vb{0}$, which return the uninformative bit $b_t=1$ and do not change its output, 
    so that every transcript has the fixed length $T$ the divergence accounting below requires.
    A dummy round has bias $1 - \vb{0}^\top(\vb{1} - \vb{p}) = 1$ under every hypothesis $\vb{p}$, 
    so it contributes $\KL(\Ber(1)\Vert\Ber(1)) = 0$ to the divergence sum below; padding is thus free. 
    This also disposes of the boundary value $1$, reached only by such dummy rounds: there the two hypotheses coincide, 
    so the contribution vanishes rather than blowing up, while the informative rounds $\vb{a}_t \neq \vb{0}$ stay off the boundary through the variance floor below.
    For $v \in \{-1,+1\}^M$ set $\vb{p}^{(v)} \coloneqq \tfrac12\vb{1} + 2\alpha\, v \in [\tfrac12 - 2\alpha,\, \tfrac12 + 2\alpha]^M$.
    The two hypotheses for each coordinate, $\tfrac12 \pm 2\alpha$, are $4\alpha$ apart, so any $\hat{\vb{p}}$ with $\norm{\hat{\vb{p}} - \vb{p}^{(v)}}_\infty \le \alpha$ has $\operatorname{sgn}(\hat p_i - \tfrac12) = v_i$ for all $i$;
    an $\alpha$-accurate estimator thus recovers $v$ exactly on success, 
    so $d_{\mathrm H}(\hat v, v) = 0$ on success and at most $M$ always. 
    Therefore, success probability $\ge 2/3$ forces the average Hamming risk 
    \begin{equation}
        2^{-M}\sum_v \mathbb{E}\, d_{\mathrm H}(\hat v, v) \le M\cdot\tfrac13 = M/3.
    \end{equation}

    Let $\mathbb{P}_v$ be the distribution of the transcript $(\vb{a}_t, b_t)_{t=1}^T$ under $\vb{p}^{(v)}$, and let $v^{\oplus i}$ flip coordinate $i$.
    Write $s_t^{(v)} \coloneqq 1 - \vb{a}_t^\top(\vb{1} - \vb{p}^{(v)})$ for the bias of query $t$ under hypothesis $v$, 
    and abbreviate the per-query, per-coordinate Kullback--Leibler divergence (relative entropy) by $D_{t,i} \coloneqq \KL\!\big(\Ber(s_t^{(v)}) \,\Vert\, \Ber(s_t^{(v^{\oplus i})})\big)$, where $\KL(P\Vert Q) \coloneqq \sum_\omega P(\omega)\log\tfrac{P(\omega)}{Q(\omega)}$.

    Because each query may depend on past outcomes, the transcript distributions are not products;
    still, the KL chain rule applies round by round.
    The query policy is identical under $\vb{p}^{(v)}$ and $\vb{p}^{(v^{\oplus i})}$ (the algorithm never sees $v$),
    so only the outcome bit of each round contributes, and for every \emph{fixed} coordinate $i$,
    \begin{align}
        \KL\!\left(\mathbb{P}_v \,\Vert\, \mathbb{P}_{v^{\oplus i}}\right)
        = \sum_{t=1}^T \mathbb{E}_{\mathbb{P}_v}\!\left[ D_{t,i} \right].
        \label{apd:eq:kl_chain_fixed_i}
    \end{align}
    The quantity entering Assouad's bound below is the \emph{total} divergence $\sum_{i=1}^M \KL(\mathbb{P}_v \Vert \mathbb{P}_{v^{\oplus i}})$ over all $M$ coordinate flips, which we bound by first controlling the coordinate sum $\sum_{i=1}^M D_{t,i}$ for one fixed query $t$ and then summing over the $T$ queries.
    The two biases differ only through coordinate $i$, by $\abs{s_t^{(v)} - s_t^{(v^{\oplus i})}} = 4\alpha\, a_{t,i}$.
    The denominator of the Pinsker/KL bound below is the Bernoulli variance $s(1-s)$, which the window $\alpha < 1/8$ keeps bounded away from $0$.
    For any $\vb{p} \in [\tfrac12 - 2\alpha,\, \tfrac12 + 2\alpha]^M$ the bias $s = 1 - \vb{a}_t^\top(\vb{1} - \vb{p})$ satisfies $1 - s = \vb{a}_t^\top(\vb{1} - \vb{p}) \ge (\tfrac12 - 2\alpha)\norm{\vb{a}_t}_1$ and $s \ge \tfrac12 - 2\alpha \ge \tfrac14$, so that
    \begin{equation}\label{apd:eq:variance_floor}
        s(1-s) \ge c_1\norm{\vb{a}_t}_1, \qquad c_1 \coloneqq \tfrac14\big(\tfrac12 - 2\alpha\big) = \Theta(1).
    \end{equation}
    This holds for both $s_t^{(v)}$ and $s_t^{(v^{\oplus i})}$, as $\vb{p}^{(v)}, \vb{p}^{(v^{\oplus i})}$ lie in the window.
    Each Bernoulli relative entropy is bounded by the corresponding chi-square divergence, 
    $\KL(\Ber(s)\Vert\Ber(s')) \le \chi^2(\Ber(s)\Vert\Ber(s')) = (s-s')^2/[s'(1-s')]$ \mbox{\cite[Lemma~2.7]{tsybakovIntroductionNonparametricEstimation2008}}, 
    the squared bias gap over the Bernoulli variance. 
    Inserting the bias gap $4\alpha\, a_{t,i}$ and the variance floor \cref{apd:eq:variance_floor} (queries with $\vb{a}_t = \vb{0}$ are deterministic and contribute nothing),
    \begin{align}
        \sum_{i=1}^M D_{t,i}
        &\;\le\; \sum_{i=1}^M \frac{\big(s_t^{(v)} - s_t^{(v^{\oplus i})}\big)^2}{s_t^{(v^{\oplus i})}\big(1 - s_t^{(v^{\oplus i})}\big)}
        \tag{$\KL \le \chi^2$} \\
        &\;\le\; \sum_{i=1}^M \frac{(4\alpha\, a_{t,i})^2}{c_1 \norm{\vb{a}_t}_1} \tag{variance $\ge c_1\norm{\vb{a}_t}_1$} \\
        &\;=\; \frac{16\alpha^2}{c_1}\,\frac{\norm{\vb{a}_t}_2^2}{\norm{\vb{a}_t}_1}
        \;\le\; \frac{16\alpha^2}{c_1}.
    \end{align}
    where the last step uses $\norm{\vb{a}_t}_2^2 \le \norm{\vb{a}_t}_\infty\norm{\vb{a}_t}_1 \le \norm{\vb{a}_t}_1$ and $\norm{\vb{a}_t}_\infty \le 1$.

    Summing the chain rule \cref{apd:eq:kl_chain_fixed_i} over the $M$ coordinates, exchanging the finite $i$- and $t$-sums, and inserting the per-query bound 
    (which holds for every realization of the adaptive query $\vb{a}_t$, so its constant right-hand side $16\alpha^2/c_1$ bounds the expectation $\mathbb{E}_{\mathbb{P}_v}$ by monotonicity),
    \begin{align}
        \sum_{i=1}^M \KL\!\left(\mathbb{P}_v \,\Vert\, \mathbb{P}_{v^{\oplus i}}\right)
        = \sum_{t=1}^T \mathbb{E}_{\mathbb{P}_v}\!\left[ \sum_{i=1}^M D_{t,i} \right]
        \le \sum_{t=1}^T \frac{16\alpha^2}{c_1}
        = \frac{16\alpha^2}{c_1}\, T ,
    \end{align}
    valid for every $v$.

    By \cref{lem:assouad}, applied to the transcript distributions $\{\mathbb{P}_v\}$ with half-cube mixtures $\mathbb{P}_{\pm i} \coloneqq 2^{-(M-1)}\sum_{v:\, v_i = \pm1}\mathbb{P}_v$, each averaging over the $2^{M-1}$ hypotheses that fix the sign $v_i = \pm1$ 
    (in contrast to the single-flip pair $\mathbb{P}_v, \mathbb{P}_{v^{\oplus i}}$ controlled in the divergence bound above),
    the average Hamming risk is at least $\tfrac{M}{2} - \tfrac12\sum_{i=1}^M \TV(\mathbb{P}_{+i}, \mathbb{P}_{-i})$.
    Convexity of total variation links the two in the first step below, 
    where each mixture-pair distance $\TV(\mathbb{P}_{+i}, \mathbb{P}_{-i})$ is bounded by the averaged single-flip distance $\mathbb{E}_v\,\TV(\mathbb{P}_v, \mathbb{P}_{v^{\oplus i}})$.
    Success forces this risk to be at most $M/3$, so the total-variation sum must be large. 
    We now upper-bound it in terms of the query budget $T$, 
    via convexity of total variation, Pinsker's inequality $\TV(P,Q)\le\sqrt{\tfrac12\KL(P\Vert Q)}$, and Cauchy--Schwarz over the $M$ coordinates:
    \begin{align}
        \sum_{i=1}^M \TV(\mathbb{P}_{+i}, \mathbb{P}_{-i})
        &\;\le\; \mathbb{E}_v \sum_{i=1}^M \TV\!\left(\mathbb{P}_v, \mathbb{P}_{v^{\oplus i}}\right)
        \tag{convexity of $\TV$} \\
        &\;\le\; \mathbb{E}_v \sqrt{\tfrac{M}{2}\textstyle\sum_i \KL(\mathbb{P}_v \Vert \mathbb{P}_{v^{\oplus i}})}
        \tag{Pinsker $+$ Cauchy--Schwarz} \\
        &\;\le\; \sqrt{8/c_1}\,\alpha\sqrt{MT} .
        \tag{$\textstyle\sum_i \KL \le 16\alpha^2 T/c_1$}
    \end{align}
    Here Pinsker is applied per coordinate and Cauchy--Schwarz collapses the $M$ resulting square roots into a single factor $\sqrt{M}$. 
    This term-count $\sqrt{M}$ is distinct from the $M$-independent budget bound $\sum_i \KL(\mathbb{P}_v \Vert \mathbb{P}_{v^{\oplus i}}) \le 16\alpha^2 T/c_1$, so the $M$ on the left of the final squeeze (from the Hamming risk) and the $\sqrt{M}$ on the right (from Cauchy--Schwarz) have different origins and do not double-count.
    Combining this total-variation bound with the success cap $M/3$ on the average Hamming risk established above and the Assouad lower bound \cref{lem:assouad} gives
    \begin{align}
        \frac{M}{6}
        \;=\; \frac{M}{2} - \frac{M}{3}
        &\;\le\; \frac12\sum_{i=1}^M \TV(\mathbb{P}_{+i}, \mathbb{P}_{-i})
        \;\le\; \frac12\sqrt{8/c_1}\,\alpha\sqrt{MT}
        \\
        &\;\Longrightarrow\; T \ge \frac{c_1 M}{72\,\alpha^2} = \Omega(M/\alpha^2).
    \end{align}
\end{proof}

To turn this information-theoretic bound into a physical one, we exhibit a local Hamiltonian in which every prepare-and-measure round \emph{is} the single-bit mixture oracle of \cref{lem:classical_multi_lower_bound}: the Feynman--Kitaev clock, evolved for a time $t$, spreads the measured projector across the $M$ output qubits with clock-weight coefficients $w_q(t)$, so each snapshot returns one Bernoulli bit whose bias is a weighted average of the $M$ unknown projector marginals. Estimating all $M$ snapshots to precision $\eps$ therefore forces estimating these $M$ marginals, and the $\Omega(M/\alpha^2)$ cost of \cref{lem:classical_multi_lower_bound} becomes the claimed $\Omega(M/\eps^2)$ rounds.

\begin{theorem}
    [Classical $\Omega(M/\eps^2)$ lower bound for local \nameref{apd:prm:multi_time_obs}]
    \label{thm:joint_classical_lower_bound}
    Let $M$ be a positive integer and $\eps \in (0, \eps_0)$ for a sufficiently small absolute constant $\eps_0 > 0$.
    There exist an $\bigO(1)$-local Hamiltonian $\ham$ on $n = \bigO(M^2)$ qubits, a single $1$-local projector $P$, and $M$ time points $\{t_m\}_{m=1}^{M}$ such that any classical prepare-and-measure algorithm (\cref{def:pm_model}) estimating the snapshot expectations $\bra{\psi_0} e^{\ii\ham t_m}\, P\, e^{-\ii\ham t_m}\ket{\psi_0}$ at all $M$ time points to additive error $\eps$ with success probability at least $2/3$ requires $\Omega(M/\eps^2)$ rounds, in the worst case over the oracle-prepared initial state $\ket{\psi_0}$.
\end{theorem}
\begin{proof}
    We reuse the Feynman--Kitaev instance of \cref{apd:thm:local_lower_bound}:
    the $\bigO(1)$-local Hamiltonian $H_{\mathrm{FK}}$ on $\bigO(M^2)$ qubits
    (the same Feynman--Kitaev clock Hamiltonian, built from the cyclic-shift SWAP circuit $V = \Pi^M$ of $N = M^2$ nearest-neighbour gates), projector $P = \op{0}_{\oo,1}$, 
    time points $\{t_m\}_{m=1}^M$, clock weights $w_q(t)$, 
    and $M\times M$ weight matrix $W_{mq} = w_q(t_m)$ with $\gamma \coloneqq \norm{W^{-1}}_\infty = \bigO(1)$,
    the $\ell_\infty$ error-amplification incurred when inverting $W$ to recover the projector marginals (\cref{eq:local_echo_FK_reduced}). 
    Set $\eps_0 = 1/(8\gamma)$, 
    so $\eps < \eps_0$ ensures $\gamma\eps < 1/8$, the admissible window of \cref{lem:classical_multi_lower_bound}.
    The power-of-two hypothesis of \cref{apd:thm:local_lower_bound} is not needed here: it stems from the Apeldoorn--Huggins reduction in that proof, whereas the Feynman--Kitaev construction and the diagonal dominance of $W$ (\cref{eq:block_weight_lower}) hold for every positive integer $M$.
    The statistical hardness resides entirely in the family of initial states $\{\ket{\psi'_v}\}$ defined next, which carries that lemma's near-$\tfrac12$ hypercube.

    For $v \in \{-1,+1\}^M$,
    let $\ket{\psi'_v} = \bigotimes_{i=1}^{M}\ket{\phi_i}$ be the product state on output qubits $2,\ldots,M+1$ with marginals $p_{i+1} \coloneqq \bra{\phi_i}\op{0}\ket{\phi_i} = \tfrac12 + 2\gamma\eps\, v_i$, 
    and put $\ket{\psi_0} = \ket{0}_\cc \otimes \ket{0}_{\oo,1} \otimes \ket{\psi'_v}$;
    the oracle $U_\psi$ is built from $U_{\psi'_v}$ with $\bigO(1)$ fixed-state preparations.
    By the snapshot decomposition \cref{eq:local_echo_FK} (with $p_1 = 1$), 
    a PM round that evolves to time $t$ and measures the projector $P$ returns one bit
    \begin{align}
        \Ber\!\Big(\abs{c_0(t)}^2 + \textstyle\sum_{q=0}^{M-1} w_q(t)\, p_{q+2}\Big) = \Ber\big(\vb{a}^\top \vb{p} + (1 - \norm{\vb{a}}_1)\big),
        \qquad \vb{a} = \big(w_q(t)\big)_{q=0}^{M-1},
    \end{align}
    where $\vb{p} = (p_2,\ldots,p_{M+1}) \in [\tfrac12 - 2\gamma\eps,\, \tfrac12 + 2\gamma\eps]^M$ and $\norm{\vb{a}}_1 = \sum_q w_q(t) = 1 - \abs{c_0(t)}^2$.
    Each PM round is therefore a single-bit-oracle query of \cref{lem:classical_multi_lower_bound} for $\vb{p}$, the choice of evolution time $t$ selecting the weight vector $\vb{a}$.
    These realizable weights form a sub-family of the lemma's admissible vectors, but its lower bound transfers verbatim: 
    \emph{
    the per-query bound $\sum_i \KL(\mathbb{P}_v \Vert \mathbb{P}_{v^{\oplus i}}) \le 16\alpha^2/c_1$ holds \emph{uniformly} over all admissible $\vb{a}_t$, so restricting to the realizable $\vb{a}$ can only lower the information per round and hence cannot lower the query count.}
    An algorithm estimating all $M$ snapshots to additive error $\eps$ in particular estimates $\tilde{\vb{o}}$ to $\eps$; applying $W^{-1}$ (\cref{eq:local_echo_FK_reduced}) yields $\hat{\vb{p}}$ with $\norm{\hat{\vb{p}} - \vb{p}}_\infty \le \gamma\eps$.
    By \cref{lem:classical_multi_lower_bound} with accuracy $\alpha = \gamma\eps$ (admissible since $\gamma\eps < 1/8$, and matching its window $\tfrac12 \pm 2\alpha$), this requires $\Omega\big(M/(\gamma\eps)^2\big) = \Omega(M/\eps^2)$ rounds.
\end{proof}

The classical lower bound counts the same resource as the quantum upper bound: calls to $U_\psi$, i.e.\ PM rounds.
Because the PM class is a strict subclass of the coherent quantum algorithms, the universal quantum bound $\Omega(\sqrt{M}/\eps)$ of \cref{apd:thm:local_lower_bound} 
(valid in its high-precision regime $\eps = \bigO(1/\sqrt{M})$) applies to it as well. 
In that regime the classical $\Omega(M/\eps^2)$ exceeds it by the factor $\sqrt{M}/\eps \ge 1$ and is the binding classical bound.
On this single $\bigO(1)$-local instance, then, classical prepare-and-measure algorithms require $\Omega(M/\eps^2)$ rounds while the quantum algorithm of \cref{apd:cor:adaptive_snapshots} uses $\widetilde{\bigO}(\sqrt{M}/\eps)$:
a quadratic quantum speedup in \emph{both} the number of time points $M$ and the precision $1/\eps$, on one problem family.
The quantum guarantee of \cref{apd:cor:adaptive_snapshots} is in root-mean-square error. 
By Chebyshev's inequality it delivers additive error at most $3\eps$ at every time point with probability $\ge 2/3$, 
the same metric and confidence as the classical lower bound, 
so the two are compared on equal terms (up to a constant rescaling of $\eps$).

\fi

\end{document}